\documentclass[11pt,a4paper]{amsart}

\usepackage[foot]{amsaddr}
\usepackage{ifxetex}
\ifxetex
  \usepackage[no-math]{fontspec}
\else
\fi

\newcommand{\ifarxiv}[2]{#2}
\renewcommand{\ifarxiv}[2]{#1}

\usepackage{amsmath}
\usepackage{amsfonts}
\usepackage{amssymb}
\usepackage{amsthm}
\usepackage{verbatim}
\usepackage{dsfont}
\usepackage{fullpage}
\usepackage{microtype}
\usepackage{ragged2e}
\usepackage{mathtools}
\usepackage{tcolorbox}
\usepackage[libertine]{newtxmath}
\usepackage[tt=false]{libertine} 
\usepackage{caption}
\usepackage{bbm}
\usepackage{dsfont}
\usepackage{hyperref, color}
\hypersetup{colorlinks=true,citecolor=blue, linkcolor=blue, urlcolor=blue}
\usepackage[linesnumbered,boxed,ruled,vlined]{algorithm2e}
\usepackage{algorithmicx}
\usepackage{bm}
\usepackage{bbm}
\usepackage[numbers]{natbib}
\usepackage{xcolor}
\usepackage{enumerate} 
\usepackage{enumitem}
\usepackage{tabularx}
\usepackage{array}
\usepackage{cleveref}
\usepackage{pifont}
\usepackage{xifthen}
\usepackage{xstring}

\newcolumntype{L}[1]{>{\raggedright\arraybackslash}p{#1}}
\newcolumntype{C}[1]{>{\centering\arraybackslash}m{#1}}
\newcolumntype{R}[1]{>{\raggedleft\arraybackslash}p{#1}}

\usepackage{makecell}
\usepackage{footnote}
\makesavenoteenv{tabular}

\newcommand{\TODO}[1]{\typeout{TODO: \the\inputlineno: #1}\textbf{{\color{red}[[[ #1 ]]]}}}

\DeclareMathAlphabet{\mathcal}{OMS}{cmsy}{m}{n}

\newcommand{\MSC}[1]{\mathscr{{#1}}}
\newcommand{\qus}[1]{\+{Q}^{\star}_{#1}}

\newcommand{\pprime}{\alpha}

\newcommand{\rejsamp}{\textnormal{\textsf{RejectionSampling}}}
\newcommand{\newsample}{\textnormal{\textsf{MarginSample}}}

\newcommand{\pth}{\textnormal{\textsf{Path}}}

\newcommand{\simulator}{\textnormal{\textsf{Simulate}}}

\newcommand{\linbf}{\textnormal{\textsf{LinearBF}}}
\newcommand{\berrace}{\textnormal{\textsf{BernoulliRace}}}
\newcommand{\subbf}{\textnormal{\textsf{SubtractBF}}}

\newcommand{\eval}{\textnormal{\textsf{Eval}}}
\newcommand{\checkf}{\textnormal{\textsf{Frozen}}}

\newcommand{\recsample}{\textnormal{\textsf{MarginOverflow}}}

\newcommand{\tsamp}{\bar{t}_\textnormal{\textsf{MS}}}
\newcommand{\Lin}[1]{{{\mathsf{Lin}}}\left({#1}\right)} 
\newcommand{\rtsamp}{T_\textnormal{\textsf{MS}}}
\newcommand{\trecor}{\bar{t}_\textnormal{\textsf{MO}}}
\newcommand{\tbf}{T_\textnormal{\textsf{BFS}}}
\newcommand{\tbfup}{\bar{t}_\textnormal{\textsf{BF}}}

\newcommand{\induceddist}[2]{\psi^{{#1}}_{#2}}

\newcommand{\trej}{T_\textnormal{\textsf{Rej}}}

\newcommand{\tvar}{t_\textnormal{{\textsf{var}}}}
\def\Pr{\mathop{\mathbf{Pr}}\nolimits}

\newcommand{\Mod}[3]{{#1}_{{#2}\gets{#3}}}
\newcommand{\vbl}{\mathsf{vbl}}
\newcommand{\True}{\mathtt{True}}
\newcommand{\False}{\mathtt{False}}

\newcommand{\abs}[1]{\left\vert#1\right\vert}

\newcommand{\vstar}[1]{V^{#1}_{\star}}
\newcommand{\gvc}{G_{\mathsf{VC}}}

\SetKwRepeat{Do}{do}{while}

\newtheorem{theorem}{Theorem}[section]

\newtheorem*{claim*}{Claim}
\newtheorem{condition}[theorem]{Condition}
\newtheorem{example}[theorem]{Example}
\newtheorem{fact}[theorem]{Fact}
\newtheorem{lemma}[theorem]{Lemma}
\newtheorem{proposition}[theorem]{Proposition}
\newtheorem{corollary}[theorem]{Corollary}
\theoremstyle{definition}

\newtheorem{definition}[theorem]{Definition}
\newtheorem{remark}[theorem]{Remark}
\newtheorem*{remark*}{Remark}
\newtheorem{assumption}{Assumption}

\renewcommand{\Pr}[2][]{ \ifthenelse{\isempty{#1}}
  {\mathop{\mathbf{Pr}}\left[#2\right]} {\mathop{\mathbf{Pr}}_{#1}\left[#2\right]} }
\newcommand{\E}[1]{\mathbb{E}\left[{#1}\right]}
\newcommand{\Var}[2][]{ \ifthenelse{\isempty{#1}}
  {\mathbf{\mathbf{Var}}\left[#2\right]}
  {\mathbf{\mathbf{Var}}_{#1}\left[#2\right]} }
\newcommand{\var}[1]{{{\vbl}}\left({#1}\right)}  
\newcommand{\poly}{{\rm poly}}  
\newcommand{\nextvar}[1]{{{\mathsf{NextVar}}}\left({#1}\right)} 

\newcommand{\cfrozen}[1]{\+{C}^{#1}_{\mathsf{frozen}}}

\newcommand{\vfix}[1]{V^{#1}_{\mathsf{fix}}}

\newcommand{\hfix}[1]{H^{#1}_{\mathsf{fix}}}
\newcommand{\vinf}[1]{V^{#1}_{\star{\mathsf{\text{-}inf}}}}

\newcommand{\ccon}[1]{\+{C}^{#1}_{\star{\mathsf{\text{-}con}}}}

\newcommand{\vcon}[1]{V^{#1}_{\star{\mathsf{\text{-}con}}}}

\newcommand{\vst}[1]{V^{#1}_{\star}}
\newcommand{\csfrozen}[1]{\+{C}^{#1}_{\star{\mathsf{\text{-}frozen}}}}

\newcommand{\qs}{\+{Q}^*}

\newcommand{\Formal}[1]{\underline{\mathbf{#1}}}

\newcommand*{\midmathskip}{\hskip0.5\displaywidth\hskip-0.5\totwidth@}

\newcommand{\one}[1]{\mathbbm{1}\left[#1\right]}
\def\^#1{\mathbb{#1}} 
\def\*#1{\mathbf{#1}} 
\def\+#1{\mathcal{#1}} 
\def\-#1{\mathrm{#1}} 
\def\=#1{\boldsymbol{#1}} 
\newcommand{\set}[1]{\left\{#1\right\}}

\newcommand{\defeq}{\triangleq}

\usepackage{amssymb}
\usepackage{stackengine}
\usepackage{scalerel}
\usepackage{xcolor}
\usepackage{graphicx}
\newcommand\openbigstar[1][0.7]{%
  \scalerel*{%
    \stackinset{c}{-.125pt}{c}{}{\scalebox{#1}{\color{white}{$\bigstar$}}}{%
      $\bigstar$}%
  }{\bigstar}
}

\newcommand{\hollowstar}{\text{$\scriptstyle\openbigstar[.7]$}}

\usepackage[textsize=tiny,textwidth=20mm]{todonotes}

\title{Sampling Lov\'{a}sz local lemma for general constraint satisfaction solutions in near-linear time}

\ifarxiv{\author{
Kun He
}
\author{
Chunyang Wang
}
\author{
Yitong Yin
}
\address[Kun He]{The Key Lab of Data Engineering and Knowledge Engineering, MOE, Renmin University of China, No. 59 Zhongguancun Street, Haidian District, Beijing, China. \textnormal{E-mail: \url{hekun.threebody@foxmail.com}}. The research of Kun He is supported by the Strategic Priority Research Program of Chinese Academy of Sciences under Grant
No. XDA27000000, the National Natural Science Foundation
of China Grants No. 62002231, 61832003.}
\address[Chunyang Wang, Yitong Yin]{ State Key Laboratory for Novel Software Technology, Nanjing University, 163 Xianlin Avenue, Nanjing, Jiangsu Province, China. \textnormal{E-mails: \url{wcysai@smail.nju.edu.cn}, \url{yinyt@nju.edu.cn} 
}}}
\date{}
    
\begin{document}


\allowdisplaybreaks
\maketitle



\begin{abstract}
We give a fast algorithm for sampling uniform solutions of \emph{general} constraint satisfaction problems (CSPs) in a local lemma regime. 
Suppose that the CSP has $n$ variables with domain size at most $q$, 
each constraint contains at most $k$ variables, shares variables with at most $\Delta$ constraints, and is violated with probability at most $p$ by a uniform random assignment.
The algorithm returns an almost uniform satisfying assignment in expected $\mathrm{poly}(q,k,\Delta)\cdot\tilde{O}(n)$ time,
as long as a local lemma condition is satisfied:
\[
k\cdot p\cdot q^2\cdot \Delta^5\le C_0\quad\text{for a suitably small absolute constant }C_0.
\]
Previously, under similar local lemma conditions, sampling algorithms with running time polynomial in both $n$ and $\Delta$
were only known for the almost atomic case, where each constraint is violated by a small number of forbidden local configurations.
The key term $\Delta^5$ in our local lemma condition also improves the previously best known $\Delta^7$ for general CSPs~\cite{Vishesh21towards} and $\Delta^{5.714}$ for atomic CSPs, including the special case of $k$-CNF~\cite{Vishesh21sampling, HSW21}.

%

Our sampling approach departs from previous fast algorithms for sampling LLL, which were based on Markov chains.
A crucial step of our algorithm is a recursive marginal sampler that is of independent interests. 
Within a local lemma regime, 
this marginal sampler can draw a random value for a variable according to its marginal distribution,
at a cost independent of the size of the CSP.
%
\end{abstract}

\setcounter{tocdepth}{1}
\tableofcontents

\setcounter{page}{0} \thispagestyle{empty} \vfill
\pagebreak

\section{Introduction}\label{sec:intro}
Constraint satisfaction problems (CSPs) are one of the most fundamental objects in computer science.
A CSP is described by a collection of constraints defined on a set of variables. 
Formally, an instance of constraint satisfaction problem, 
called a \emph{CSP formula}, 
is denoted by $\Phi=(V,\+{Q},\+{C})$.
Here, $V$ is a set of $n=|V|$ variables;
$\+{Q}\triangleq\bigotimes_{v\in V}Q_v$ is a product space of all assignments of variables, 
where each $Q_v$ is a finite domain of size $q_v\triangleq\abs{Q_v}\ge 2$ over where the variable $v$ ranges;
and $\+{C}$ gives a collection of local constraints, 
such that each  $c\in \+{C}$ is a constraint function $c:\bigotimes_{v\in \vbl(c)}Q_v\to\{\True,\False\}$ defined on a subset of variables, denoted by $\vbl(c)\subseteq V$.
An assignment $\={x}\in \+Q$ is called \emph{satisfying} for $\Phi$ if 
\[
\Phi(\={x})\triangleq\bigwedge\limits_{c\in\+{C}} c\left(\={x}_{\vbl(c)}\right)=\True.
\]
%
%
%
%
The followings are some key parameters of a CSP formula $\Phi=(V,\+{Q},\+{C})$:
\begin{itemize}
    \item \emph{domain size} $q=q_\Phi\triangleq\max\limits_{v\in V}\abs{Q_{v}}$  
    and \emph{width} $k=k_\Phi\triangleq\max\limits_{c\in \+{C}}\abs{ {\vbl}(c)}$;
    \item \emph{constraint degree} $\Delta=\Delta_\Phi\triangleq\max\limits_{c\in \+{C}}\abs{\{c'\in \+{C}\mid \vbl(c)\cap \vbl(c')\neq\emptyset\}}$;\footnote{The constraint degree $\Delta$ should be distinguished from the \emph{dependency degree} $D$, which is the maximum degree of the dependency graph: $D\triangleq\max_{c\in \+{C}}\abs{\{c'\in \+{C}\setminus\{c\}\mid \vbl(c)\cap \vbl(c')\neq\emptyset\}}$. Note that $\Delta=D+1$.}
    \item \emph{violation probability} $p=p_{\Phi}\triangleq\max\limits_{c\in \+{C}}\mathbb{P}[\neg c]$, where $\mathbb{P}$ denotes the law for the uniform assignment, 
    in which each $v\in V$ draws its evaluation from~$Q_v$ uniformly and independently at random.
\end{itemize}

The famous \emph{Lov\'{a}sz Local Lemma (LLL)}~\cite{LocalLemma} provides a sufficient criterion for the satisfiability of~$\Phi$.
Specifically, a satisfying assignment for a CSP formula $\Phi$ exists if
\begin{align}\label{eq:classic-LLL-condition}
    \mathrm{e}p\Delta\le 1.
\end{align}
Due to a lower bound of Shearer~\cite{shearer85},
such ``{LLL condition}'' for the existence of satisfying solution is essentially tight if only knowing $p$ and $\Delta$. 
On the other hand,
the \emph{algorithmic} or \emph{constructive LLL} seeks to find a solution efficiently.
A major breakthrough was the Moser-Tardos algorithm~\cite{moser2010constructive},
which guarantees to find a satisfying assignment efficiently under the LLL condition in~\eqref{eq:classic-LLL-condition}.

\medskip
\noindent
\textbf{The sampling LLL.}\,\,
We are concerned with the problem of \emph{sampling Lov\'{a}sz Local Lemma}, 
which has drawn considerable attention in recent years~\cite{GJL19,Moi19,guo2019counting,galanis2019counting,harris2020new,FGYZ20,feng2021sampling,Vishesh21sampling,Vishesh21towards,HSW21,galanis2021inapproximability,feng2022improved}.
In the context of CSP,
it seeks to provide an efficient sampling algorithm for (nearly) uniform generation of satisfying assignments for the CSPs in an LLL-like regime. 
This sampling LLL problem is closely related to the problem of estimating the volume of solution spaces or the partition functions of statistical physics systems,
and is motivated by fundamental tasks, including the probabilistic inferences in graphical models~\cite{Moi19} and the network reliability problems~\cite{GJL19,guo2019polynomial,guo2020tight}. 
%

This problem of sampling LLL turns out to be  computationally more challenging than the traditional algorithmic LLL, 
which requires constructing an arbitrary satisfactory assignment, not necessarily following the correct distribution.
For example, when used as a sampling algorithm, 
the Moser-Tardos algorithm can only guarantee correct sampling on restrictive classes of CSPs~\cite{GJL19}.
Due to the computational lower bounds shown in \cite{BGGGS19,galanis2021inapproximability}, 
a strengthened LLL condition with $c\ge 2$:
\begin{align}\label{eq:sampling-LLL-condition-abstract}
    p\Delta^c\lesssim 1,
\end{align}
is necessary for the tractability of sampling LLL,
even restricted to typical specific sub-classes of CSPs, such as CNF or hypergraph coloring.
Here $\lesssim$ ignores the lower-order terms and the constant factor.

In a seminal work of Moitra \cite{Moi19},
a very innovative algorithm was given for sampling almost uniform $k$-CNF solutions assuming an LLL condition $p\Delta^{60}\lesssim 1$.
%
This sampling algorithm was based on deterministic approximate counting 
by solving linear programs on properly factorized formulas
and has a running time of $n^{\poly(k,\Delta)}$. 
This LP-based approach was later extended to hypergraph coloring \cite{guo2019counting} and random CNF formulas \cite{galanis2019counting},
and finally in a work of Jain, Pham and Vuong \cite{Vishesh21towards} to all CSPs satisfying a substantially improved LLL condition $p\Delta^{7}\lesssim 1$.
All these deterministic approximate counting based algorithms suffered from an $n^{\poly(k,\Delta)}$ time cost.

Historically,
rapidly mixing Markov chains have been the canonical sampling algorithms,
and often have near-linear time efficiency.
However, for sampling LLL, 
there used to be a fundamental barrier for Markov chains.
That is,
despite the ubiquity of solutions, the solution space of CSPs may be highly disconnected through the transition of local Markov chains.

This  barrier of disconnectivity was circumvented in a breakthrough of Feng~\emph{et al.}~\cite{FGYZ20},
in which a rapidly mixing \emph{projected} random walk was simulated efficiently 
on a subset of variables constructed using the marking/unmarking strategy invented in \cite{Moi19}.
Assuming an LLL condition $p\Delta^{20}\lesssim 1$, 
this new algorithm could generate an almost uniform $k$-CNF solution using a time cost within $\poly(k,\Delta)\cdot n^{1.0001}$, 
which is close to linear in the number of variables~$n$.
By observing that this marking/unmarking of variables  
was, in fact a specialization in the Boolean case of compressing variables' states, 
this Markov chain based fast sampling approach was generalized in \cite{feng2021sampling} 
to CSPs beyond the Boolean domain, 
specifically, to all almost \emph{atomic} CSPs (which we will explain later),
assuming an LLL condition $p\Delta^{350}\lesssim 1$.
This bound was remarkably improved to $p\Delta^{7.04}\lesssim 1$ in another work of Jain, Pham and Vuong~\cite{Vishesh21sampling} through a very clever witness-tree-like information percolation analysis of the mixing time,
which was also used later to support a perfect sampler through the coupling from the past (CFTP) in~\cite{HSW21}
with a further improved condition $p\Delta^{5.71}\lesssim 1$.
%

%
%
%

All these fast algorithms for sampling LLL are restricted to the (almost) atomic CSPs, 
in which each constraint $c$ is violated by exactly one (or very few) forbidden assignment(s) on $\vbl(c)$.

\medskip
\noindent
\textbf{Challenges for general CSP.}\,\,
New techniques are needed for fast sampling LLL for general CSPs.
All existing fast algorithms for sampling LLL relied on some projection of the solution space to a much smaller space where the barrier of disconnectivity could be circumvented because the images of the projection might collide and were well connected. 
In order to efficiently simulate the random walk on the projected space and to recover a random solution from a random image, one would hope that the CSP formula were well ``factorized'' into small clusters most of the time because many constraints had  already been satisfied for sure given the current image, 
which was indeed the case for fast  sampling LLL for atomic CSPs \cite{FGYZ20, feng2021sampling, Vishesh21sampling, HSW21}. 
But for general non-atomic CSPs, it may no longer be the case,
because now a bad event (violation of a constraint) may be highly non-elementary,
and hence is no longer that easy to  avoid cleanly after projection, 
which breaks the factorization.

It is possible that the non-atomicity of general CSPs might have imposed greater challenges to the sampling LLL than to its constructive counterpart.
To see this, note that general CSPs can be simulated by atomic ones:
by replacing each general constraint $c$ having $N$ forbidden assignments on $\vbl(c)$,
with $N$ atomic constraints on the same $\vbl(c)$ each forbidding one assignment. 
Such simulation would increase the constraint degree $\Delta$  by a factor of at most $N$ and also decrease the violation probability $p$ by a factor of $N$.
For the classic LLL condition~\eqref{eq:classic-LLL-condition}
where $p$ and $\Delta$ are homogeneous, 
this would not change the LLL condition;
but the regime for the sampling LLL captured by~\eqref{eq:sampling-LLL-condition-abstract} 
would be significantly reduced,
since there $p$ and $\Delta$ are necessarily not homogeneous due to the lower bounds in \cite{BGGGS19,galanis2021inapproximability}.
This situation seems to suggest that  
the non-atomicity of general CSPs might impose bigger challenges
to the sampling LLL than to the existential/constructive LLL.

Indeed, prior to our work, 
it was not known  for general CSPs 
with unbounded width $k$ and degree $\Delta$, 
whether the sampling problem is polynomial-time tractable 
under an LLL condition like \eqref{eq:sampling-LLL-condition-abstract}.

\subsection{Our results}
In this paper, we answer the above open question positively.
We give a new algorithm that departs from all prior fast samplers based on Markov chains and achieves, for the first time, 
a fast sampling of almost uniform satisfying solutions for general CSPs in an improved local lemma regime. 
%


As in the case of algorithmic LLL~\cite{moser2010constructive,HV15}, 
we assume  an abstraction of constraint evaluations, 
because arbitrary constraint functions defined on a super-constant number of variables
can be highly nontrivial to express and evaluate.
%
%
%
Specifically, we assume the following evaluation oracle 
for checking whether a constraint is already satisfied by a partially specified assignment.
\begin{assumption}[\textbf{evaluation oracle}]\label{definition:evaluation-oracle}
There is an \emph{evaluation oracle}  for $\Phi=(V,\+{Q},\+{C})$ such that 
given any constraint $c\in \+C$, any assignment $\sigma\in\+Q_{\Lambda}\triangleq \bigotimes_{v\in \Lambda}Q_v$ specified on a subset $\Lambda\subseteq\vbl(c)$ of variables, 
the oracle answers whether $c$ is already satisfied by $\sigma$, i.e.~$c(\tau)=\True$ for all $\tau\in\+Q_{\vbl(c)}$ that $\tau_{\Lambda}=\sigma_{\Lambda}$.
\end{assumption}

For specific classes of CSPs, e.g.~$k$-CNF or hypergraph coloring, such an oracle is easy to realize.

Assuming such an oracle for constraint evaluations, 
we give the following fast, almost uniform sampler for general CSPs in a local lemma regime. 
Recall the parameters $q,k,p,\Delta$ of a CSP formula $\Phi$.

\begin{theorem}[informal]\label{thm:main-sampling}
There is an algorithm such that given as input any $\varepsilon\in(0,1)$ and any CSP formula $\Phi=(V,\+{Q},\+{C})$ with $n$ variables  satisfying 
\begin{align}\label{eq:main-thm-LLL-condition}
k\cdot p\cdot q^2\cdot \Delta^5\leq \frac{1}{150\mathrm{e}^3},
\end{align}
the algorithm terminates within  $\poly(q,k,\Delta)\cdot n\log\left(\frac{ n}{\varepsilon}\right)$ time in expectation and outputs an almost uniform sample of satisfying assignments for $\Phi$ within $\varepsilon$ total variation distance.
\end{theorem}
The formal statement of the theorem is in \Cref{maincor} (for termination and correctness of sampling) and in  \Cref{mainef-approx} (for efficiency of sampling). 
%

The condition in \eqref{eq:main-thm-LLL-condition} becomes  $p\Delta^{5+o(1)}\lesssim 1$ 
when $p\le (qk)^{-\omega(1)}$, 
while a typical case is usually given by a much smaller $p\le q^{-\Omega(k)}$. 
The  previous best bound for sampling general CSP solutions was that 
$q^3kp\Delta^7<c$ for a small constant $c$,
achieved by the deterministic approximate counting based algorithm in \cite{Vishesh21towards} 
whose  running  time was $(n/\varepsilon)^{\poly(k,\Delta,\log q)}$. 
We remark that our bound also improves the previous best bound, $p\Delta^{5.714}\lesssim 1$, for sampling almost atomic CSP and $k$-SAT~\cite{HSW21,Vishesh21sampling}.

Let $Z$ be the total number of satisfying assignments for $\Phi$. 
A $\hat{Z}$ is called an \emph{$\varepsilon$-approximation} of $Z$ if $(1-\varepsilon)Z\le \hat{Z}\le (1+\varepsilon)Z$.
By routinely going through the non-adaptive annealing process in \cite{FGYZ20}, 
the approximate sampler in \Cref{thm:main-sampling} can be used as a black-box to give 
for any $\varepsilon\in(0,1)$
an {$\varepsilon$-approximation} of $Z$ in time $\poly\left(q,k,\Delta\right)\cdot\tilde{O}\left(n^2\varepsilon^{-2}\right)$ with high probability.




\subsubsection{Perfect sampler}
%
The evaluation oracle in \Cref{definition:evaluation-oracle} in fact checks the sign of $\mathbb{P}[\neg c\mid \sigma]$, the probability that a constraint $c$ is violated given a partially specified assignment $\sigma$.
If further this probability can be estimated efficiently,
then the sampling in \Cref{thm:main-sampling} can be made perfect, 
where the output sample follows exactly the target distribution.
%

\begin{theorem}[informal]\label{thm:main-sampling-perfect}
For the input class of CSPs, if there is such an FPTAS for violation probability: 
\begin{itemize}
\item for any constraint $c\in\+{C}$, any assignment $\sigma\in\+Q_{\Lambda}$ specified on a subset $\Lambda\subseteq\vbl(c)$, and $0<\varepsilon<1$,
an $\varepsilon$-approximation of $\mathbb{P}[\neg c\mid \sigma]$  is returned deterministically within $\poly(q,k,1/\varepsilon)$ time,
\end{itemize}
then the sampling algorithm in \Cref{thm:main-sampling} returns a perfect sample of uniform satisfying assignment within $\poly(q,k,\Delta)\cdot n$  time in expectation under the same condition \eqref{eq:main-thm-LLL-condition}.
\end{theorem}

The formal statement of the theorem is in \Cref{maincor} (for termination and correctness of sampling) and in  \Cref{mainef} (for efficiency of sampling). 
In fact, we prove this perfect sampler first, and then realize the FPTAS assumed in \Cref{thm:main-sampling-perfect} using Monte Carlo experiments, 
which introduces a bounded bias to the sampling and gives us the approximate sampler claimed in \Cref{thm:main-sampling}.

For concrete classes of CSPs 
defined by simple local constraints,
it is no surprise to see that the probability $\mathbb{P}[\neg c\mid \sigma]$ 
almost always has an easy-to-compute closed-form expression,
in which case we have  a perfect sampler without assuming 
the oracles in \Cref{definition:evaluation-oracle} and in \Cref{thm:main-sampling-perfect}.

The followings are two examples of non-atomic CSPs which admit linear-time perfect samplers.

\begin{example}[$\delta$-robust $k$-SAT]
The $n$ variables are Boolean, each clause contains exactly $k$ literals, and a clause is satisfied if and only if at least $\delta k$ of its literals have the outcome $\True$.
\begin{itemize}
    \item
    For $\delta$-robust $k$-SAT with variable degree $d$ (each variable appears in at most $d$ clauses) satisfying 
    \[0<\delta<\frac{1}{2}, \quad 
    k\geq \frac{24\ln{k}+20\ln{d}+40}{\left(1-2\delta\right)^2},
    \]
    a perfect sample of uniform satisfying solutions is returned within expected time $\poly(k,d)\cdot n$ .
\end{itemize}
\end{example}

\begin{example}[$\delta$-robust hypergraphs $q$-coloring]\label{example:robust-coloring}
Each vertex is colored with one of the $q$ colors, each hyperedge is $k$-uniform and is satisfied if and only if there are no $(1-\delta)k$ vertices with the same color.
\begin{itemize}
     \item For $k$-uniform hypergraphs on $n$ vertices with maximum vertex degree $d$ satisfying 
    \[
    (1-\delta)k\geq 15,\quad q\geq \frac{7d^{\frac{5}{(1-\delta)k-3}}\cdot 4^{\frac{1}{(1-\delta)}}}{(1-\delta)^{1.25}},
    \]
    a perfect sample of uniform satisfying coloring is returned within expected time $\poly(q,k,d)\cdot n$. 
\end{itemize}
\end{example}


\subsubsection{Marginal sampler}
The core component of our sampling algorithm is a \emph{marginal sampler} for drawing from marginal distributions.
Let $\mu=\mu_{\Phi}$ denote the uniform distribution over all satisfying assignments for $\Phi$, and for each $v\in V$, let $\mu_v$ denote the marginal distribution at $v$ induced by $\mu$.
%

\begin{theorem}[informal]\label{thm:main-margin-sampling}
There is an algorithm such that given as input any $\varepsilon\in(0,1)$, any CSP formula $\Phi=(V,\+{Q},\+{C})$ satisfying \eqref{eq:main-thm-LLL-condition}, and any $v\in V$, 
the algorithm returns 
a random value $x\in Q_v$ distributed approximately as $\mu_v$ within total variation distance $\varepsilon$,
within $\poly\left(q,k,\Delta,\log(1/\varepsilon)\right)$ time in expectation.
\end{theorem}
This marginal sampler is also perfect under the same assumption as in \Cref{thm:main-sampling-perfect}.
Another byproduct of this marginal sampler is the following algorithm for probabilistic inference.

\begin{theorem}[informal]\label{thm:main-inference}
There is an algorithm such that given as input any $\varepsilon,\delta\in(0,1)$, any CSP formula $\Phi=(V,\+{Q},\+{C})$ satisfying \eqref{eq:main-thm-LLL-condition}, and any $v\in V$, 
the algorithm returns for every $x\in Q_v$ an $\varepsilon$-approximation of the marginal probability $\mu_v(x)$  within $\poly\left(q,k,\Delta,1/\varepsilon,\log(1/\delta)\right)$ time with probability at least $1-\delta$.
\end{theorem}

The above two theorems are formally restated and proved in \Cref{thm:margin-sample-inference}.

By a self-reduction,
the sampling and inference algorithms in Theorems \ref{thm:main-margin-sampling} and \ref{thm:main-inference} remain to hold for the marginal distributions $\mu_v^\sigma$ conditioning on a feasible partially specified assignment $\sigma$,
as long as the LLL condition \eqref{eq:main-thm-LLL-condition} is satisfied by the new instance $\Phi^\sigma$ obtained from pinning $\sigma$ onto $\Phi$.

Both the above algorithms for marginal sampling and probabilistic inference are \textbf{local algorithms whose  costs are independent of $n$}.
Previously, 
in order to simulate or estimate the marginal distribution of a variable, 
it was often necessary to generate a full assignment on all $n$ variables, or at least pay no less than that.
%
One might have asked the following natural question: 
\begin{center}
    \emph{Can these locally defined sampling or inference problems be solved at a local cost?}
\end{center}
%
However, decades have passed, and only recently has such a novel local algorithm been discovered for marginal distributions in infinite spin systems \cite{anand2021perfect}, 
which is also our main source of inspiration.

\subsection{Technique overview}\label{subsec:technique}
As we have explained before, 
non-atomicity of constraints causes a barrier for the current Markov chain based algorithms~\cite{FGYZ20,feng2021sampling,Vishesh21sampling,HSW21,feng2022improved}.
There is another family of sampling algorithms, which we call ``resampling based'' algorithms~\cite{fill2000randomness,GJL19,FVY19,jerrum2021fundamentals,feng2022perfect}.
These algorithms use resampling of variables to fix the assignment until it follows the right distribution, 
morally like the Moser-Tardos algorithm,
and they are not as affected by disconnectivity of solution space as Markov chains, 
but here a principle to ensure the correct sampling 
is to resample the variables that the algorithm has observed and conditioned on,
which also causes trouble on non-atomic constraints,
because to ensure such constraints are satisfied,
the algorithm has to observe too many variables,
whose resampling would cancel the progress of the algorithm.

We adopt a new idea of sampling, which we call the \emph{recursive marginal sampler}.
%
It is somehow closer to the resampling based algorithms than to the Markov chains, 
but thanks to its recursive nature, the algorithm avoids excessive resampling.
This algorithm is inspired by a recent novel  algorithm of Anand and Jerrum \cite{anand2021perfect} for perfectly sampling in infinite spin systems,
where a core component is such a marginal sampler that can draw a spin according to its marginal distribution. 

Now let us consider the uniform distribution $\mu$ over all satisfying assignments of a CSP formula $\Phi=(V,\+{Q},\+{C})$, and its marginal distribution $\mu_v$ at a variable $v\in V$, say over domain $Q_v=[q]$.
To sample from this $\mu_v$ over $[q]$,
an idea is to exploit the so-called ``local uniformity" property \cite{haeupler2011new}, 
which basically says that $\mu_v$ should not be far from a uniform distribution over $[q]$ in total variation distance when $\Phi$ satisfies some local lemma condition.
Therefore, a uniform sample from $[q]$ already gives a coarse sample of $\mu_v$.
It remains to boost such a coarse sampler to a sampler with arbitrary precision. 

By the local uniformity, 
there exists a $\theta< \frac{1}{q}$ close enough to $\frac{1}{q}$, such that 
\begin{align}
\forall x\in[q],\quad
\mu_v(x)\geq \theta.\label{eq:local-uniformity-example}
\end{align}
The marginal distribution $\mu_v$ can then be divided as 
\[
q\theta \cdot \+{U} + (1-q\theta) \cdot \+{D},
\]
where $\+{U}$ is the uniform distribution over $[q]$ and $\+{D}$ gives a distribution of ``overflow'' mass such that:
\[
\forall x\in[q],\quad
\+{D}(x)=\frac{\mu_v(x)-\theta}{1-q\theta}.
\]

Sampling from $\mu_v$ then can follow this strategy:
with probability $q\theta$, 
the algorithm falls into the ``zone of local uniformity'' and returns a uniform sample from $\+{U}$;
and with probability $1-q\theta$, 
the algorithm falls into the ``zone of indecision'' and has to draw a sample from this overflow distribution $\+{D}$,
which can be done by constructing a Bernoulli factory that accesses $\mu_v$ as an oracle. 
But wait, if we had such an oracle for $\mu_v$ in the first place, why would sampling from $\mu_v$ even be a problem?

The above ``chicken or egg'' paradox is somehow resolved by a simple observation: 
if enough many other variables had already been sampled correctly, say with outcome $X$, 
then assuming a strong enough LLL condition,
there is a good chance that the resulting formula $\Phi^X$ was ``factorized'' into small clusters, 
from where a standard rejection sampling on $\Phi^X$ would be efficient for sampling from $\mu_v^X$, and overall from $\mu_v$.
Therefore, the sampling strategy is now corrected as:
after falling into the ``zone of indecision'' and before trying to draw from the overflow distribution $\+{D}$,
the algorithm picks another variable $u$ whose successful sampling might help factorize $\Phi$,
and recursively apply the marginal sampler at $u$ to draw from $u$'s current marginal distribution first.
The only loose end now is that the LLL condition is not self-reducible, meaning it is not invariant under arbitrary pinning. 
We adopt the idea of ``freezing" constraints used in \cite{Vishesh21towards} to guide the algorithm to pick variables for sampling. 
The LLL condition is replaced by a more refined invariant condition that guarantees 
for each variable picked for sampling, 
the same local uniformity as in \eqref{eq:local-uniformity-example} 
to persist throughout the algorithm,
and  also guarantees a good chance of factorization 
while there are no other variables to pick.


To show the fast convergence of the recursive sampler, 
in \cite{anand2021perfect} the strategy was to show that the branching process given by the recursion tree always has decaying offspring number in expectation given the worst-case boundary condition, 
which is not true here.
Instead, 
we apply a more average-case style analysis
and bound the expected cost of the sampler according to the recursion tree directly.

To achieve a sharper bound, we design a new combinatorial structure named generalized $\{2,3\}$-tree.
In most works on counting/sampling LLL, 
two types of bad events are considered: 
one is that the assignment of a marked variable does not fall into the zone of local uniformity;
the other is that a constraint is still not satisfied after that a large proportion of its variables are assigned~\cite{Moi19,feng2021sampling,guo2019counting,Vishesh21towards}.
In previous work,
these two bad events are treated similarly and bounded using a combinatorial structure named $\{2,3\}$-tree~\cite{alon1991parallel}.
A crucial observation is that the densities of these two types of bad events are different, which inspires our design of this new combinatorial structure to take advantage of this property and push the bounds beyond state-of-the-arts.

\section{Notations for CSP}\label{sec:CSP-notation}
We recall the definition of CSP formula $\Phi=(V,\+{Q},\+{C})$ in \Cref{sec:intro}.
We use $\Omega=\Omega_{\Phi}$ to denote the set of all satisfying assignments of $\Phi$,
and use $\mu=\mu_{\Phi}$ to denote the uniform distribution over $\Omega$. 
Recall that $\mathbb{P}$ denotes the law for the uniform product distribution over $\+{Q}$.
For $C\subseteq\+{C}$, denote $\vbl(C)\triangleq\bigcup_{c\in C}\vbl(c)$;
and for $\Lambda\subseteq V$, denote ${\+Q}_\Lambda\triangleq\bigotimes_{v\in \Lambda}Q_v$.
%
%
We introduce a notation for {partial assignments}.
%
%
%
\begin{definition}[partial assignment]\label{def:partial-assignment}
Given a CSP formula $\Phi=(V,\+{Q},\+{C})$, define: 
\[
{\+Q}^\ast\triangleq\bigotimes_{v\in V}\left(Q_v\cup\{\star,\hollowstar\}\right),
\]
where $\star$ and $\hollowstar$ are two special symbols not in any $Q_v$.
Each $\sigma\in {\+Q}^\ast$ is called a \emph{partial assignment}.
\end{definition}
In a partial assignment $\sigma\in {\+Q}^\ast$, each variable $v\in V$ is classified as follows:
\begin{itemize}
\item $\sigma(v)\in Q_v$ means that  $v$ is \emph{accessed}  by the algorithm and \emph{assigned} with the value $\sigma(v)\in Q_v$;
\item $\sigma(v)=\star$ means that  $v$ is just \emph{accessed} by the algorithm but  \emph{unassigned} yet with a value in $Q_v$;
\item $\sigma(v)=\hollowstar$  means that  $v$ is \emph{unaccessed} by the algorithm and hence \emph{unassigned} with any value.
\end{itemize}
Furthermore, 
we use $\Lambda(\sigma)\subseteq V$ and $\Lambda^{+}(\sigma)\subseteq V$ to respectively denote the sets of assigned and accessed variables in a partial assignment $\sigma\in {\+Q}^\ast$, that is:
\begin{align}
\Lambda(\sigma)
\triangleq \{v\in V\mid  \sigma(v) \in Q_v\}
\quad\text{ and }\quad
\Lambda^{+}(\sigma)
\triangleq \{v\in V\mid  \sigma(v)\neq \hollowstar\}.\label{eq:def-Lambda}
\end{align}
Given any partial assignment $\sigma\in {\+Q}^\ast$ and variable $v\in V$, we further denote by $\Mod{\sigma}{v}{x}$ the partial assignment obtained from modifying $\sigma$ by replacing $\sigma(v)$ with $x\in Q_v\cup\{\star,\hollowstar\}$.
%

%
A partial assignment $\tau\in\qs$ is said to \emph{extend} a partial assignment $\sigma\in {\+Q}^\ast$ if 
$\Lambda(\sigma)\subseteq\Lambda(\tau)$, $\Lambda^+(\sigma)\subseteq\Lambda^+(\tau)$, 
and $\sigma,\tau$ agree with each other over all  variables in $\Lambda(\sigma)$.
A partial assignment $\sigma\in {\+Q}^\ast$ is said to satisfy a constraint $c\in\+{C}$ if $c$ is satisfied by all full assignments $\tau\in\+Q$ that extend $\sigma$.
A partial assignment $\sigma\in {\+Q}^\ast$ is called \emph{feasible} if there is a satisfying assignment $\tau\in\Omega$ that extends $\sigma$.

Given any feasible $\sigma\in {\+Q}^\ast$ and any $S\subseteq V$, we use $\sigma_S$ to denote $\bigotimes_{v\in S}\sigma(v)$ and $\mu_S^{\sigma}$ to denote the marginal distribution  induced by $\mu$ on $S$ conditioning on $\sigma$. 
For each $ \tau\in \+{Q}_{S}$, we have
$\mu_S^{\sigma}(\tau)=\Pr[X\sim\mu]{X_S=\tau\mid \forall v\in\Lambda(\sigma), X(v)=\sigma(v)}$.
We further write $\mu_v^{\sigma}=\mu_{\{v\}}^{\sigma}$.
Similar notation is used for the law $\mathbb{P}$ for the uniform product distribution over $\+{Q}$.  For  $\sigma\in {\+Q}^\ast$ and any event $A\subseteq \+{Q}$,
we have
$\mathbb{P}[A\mid \sigma]= \Pr[X\in\+{Q}]{X\in A\mid  \forall v\in\Lambda(\sigma), X(v)=\sigma(v)}$.

\section{The Sampling Algorithm}
We give our main algorithm for sampling almost uniform satisfying assignments for a CSP formula. 
Our presentation uses notations defined in \Cref{sec:CSP-notation}.

\subsection{The main sampling algorithm}\label{sec:main-sampling-algorithm}
Our  main sampling  algorithm takes as input a CSP formula $\Phi=(V,\+{Q},\+{C})$
with domain size $q=q_\Phi$, 
width $k=k_\Phi$, 
constraint degree  $\Delta=\Delta_\Phi$, 
and violation probability $p=p_{\Phi}$, where the meaning of these parameters are as defined in \Cref{sec:intro}.
%

We suppose that the $n=|V|$ variables are enumerated as $V=\{v_1,v_2,\ldots,v_n\}$ in an arbitrary order.
The CSP formula  $\Phi=(V,\+{Q},\+{C})$ is presented to the algorithm 
by the evaluation oracle in \Cref{definition:evaluation-oracle}.
Also assume that given any $c\in\+{C}$ (or $v\in V$), the $\vbl(c)$ (or $\{c\in\+{C}\mid v\in\vbl(c)\}$) can be retrieved.
%


The main algorithm (\Cref{Alg:main}) is the same as the main sampling frameworks in \cite{Vishesh21towards,guo2019counting}.
A partial assignment $X\in \qs$ is maintained, initially as the empty assignment $X=\hollowstar^V$.
\begin{enumerate}
\item
In the 1st phase, 
at each step it adaptively picks (in a predetermined order) a  variable $v$ that has enough ``freedom'' because it is not involved in any easy-to-violate constraint given  the current $X$, 
and replaces $X(v)$ with a random value drawn by a subroutine \newsample{} according to the correct marginal distribution $\mu_v^X$. 
\item
When no such variable with enough freedom remains,
the formula is supposed to be ``factorized" enough into small clusters and the algorithm enters the 2nd phase,
from where the partial assignment constructed in the 1st phase is completed to a uniform random satisfying assignment by a standard
\rejsamp{} subroutine.
\end{enumerate}

A key threshold $\pprime$ for the violation probability is fixed as below:
\begin{align}\label{eq:parameter-p-prime}
\pprime=\left(18\mathrm{e}^2q^2k\Delta^2\right)^{-1},
\end{align}
which satisfies $\pprime>p=p_{\Phi}$, assuming the LLL condition in \eqref{eq:main-thm-LLL-condition}.


For the ease of exposition, we assume an oracle for approximately deciding whether a constraint becomes too easy to violate given the current partial assignment.

\begin{assumption}\label{assumption:frozen-oracle}
There is an oracle such that 
given any partial assignment $\sigma\in\+{Q}^*$ and any constraint $c\in\+{C}$, 
the oracle distinguishes between the two cases: $\mathbb{P}[\neg c \mid \sigma]>\pprime$ and $\mathbb{P}[\neg c \mid \sigma]<0.99\pprime$, 
and answers arbitrarily and consistently if otherwise, which means that the answer to the undefined case $\mathbb{P}[\neg c \mid \sigma]\in[0.99\pprime,\pprime]$ can be either ``yes'' or ``no''  but
remains the same for the same $\sigma_{\vbl(c)}$.
\end{assumption}

Such an oracle is clearly implied by the FPTAS for violation probability assumed in \Cref{thm:main-sampling-perfect}
and will be explicitly realized  later  in \Cref{section-samp-ef}.
For now, with respect to such an oracle, 
the classes of easy-to-violate constraints and their involved variables are defined as follows.

\begin{definition}[frozen and fixed]\label{definition:frozen-fixed}
Assume \Cref{assumption:frozen-oracle}.
Let $\sigma\in \+{Q}^*$ be a partial assignment.
\begin{itemize}
\item  
A constraint $c\in \+{C}$ is called \emph{$\sigma$-frozen} if it is reported $\mathbb{P}[\neg c \mid \sigma] > \pprime$ by the oracle in  \Cref{assumption:frozen-oracle}.
Denote by $\cfrozen{\sigma}$ the set of all $\sigma$-frozen constraints:
%
\begin{align*}
\cfrozen{\sigma}
&\triangleq \left\{c\in\+{C}\mid \text{$c$ is reported by the oracle to satisfy }\mathbb{P}[\neg c \mid \sigma] > \pprime\right\}. 
\end{align*}
\item 
A variable $v\in V$ is called \emph{$\sigma$-fixed} if $v$ is accessed in $\sigma$ or  is involved in some $\sigma$-frozen constraint. 
Denote by $\vfix{\sigma}$ the set of all $\sigma$-fixed variables:
\begin{align*}
\vfix{\sigma}
&\triangleq  \Lambda^+(\sigma) \cup \bigcup_{c\in \cfrozen{\sigma}}\vbl(c). 
\end{align*}
\end{itemize}
\end{definition}
\noindent
Similar ideas of freezing appeared in previous works on sampling and algorithmic LLL~\cite{Vishesh21towards, beck1991algorithmic}.

\begin{remark}[one-sided error for frozen/fixed decision]\label{remark:one-sided-error-frozen-oracle}
By the property of the oracle in  \Cref{assumption:frozen-oracle},
any constraint $c\in\+{C}$ with $\mathbb{P}[\neg c \mid \sigma] > \pprime$ must be in $\cfrozen{\sigma}$, 
and any variable $v\in V$ involved in such a constraint must be in $\vfix{\sigma}$;
%
conversely, 
any $\sigma$-frozen constraint $c\in\cfrozen{\sigma}$ must have $\mathbb{P}[\neg c \mid \sigma] \geq 0.99\pprime$ 
and any unaccessed $\sigma$-fixed variable $v\in\vfix{\sigma}$ must be involved in at least one of such constraints.
\end{remark}

%


\begin{algorithm}
  \caption{The sampling algorithm} \label{Alg:main}
    \SetKwInput{KwPar}{Parameter}
  \KwIn{a CSP formula $\Phi=(V,\+{Q},\+{C})$;}  
  \KwOut{a uniform random satisfying assignment $X\in \Omega_\Phi$;}
 $X\gets \hollowstar^V$\;\label{line-main-init}
 \For{$i=1$ to $n$\label{line-main-for}}
 {
    \If{$v_i$ is not $X$-fixed\label{line-main-if}}
    {
         $X(v_i) \gets \newsample{}(\Phi,X,v_i)$\; \label{line-main-sample}
    }
}
 $X_{V\setminus \Lambda(X)}\gets \rejsamp{}(\Phi,X,V\setminus \Lambda(X))$\;\label{line-main-partialass}
 \textbf{return} $X$\;
\end{algorithm}

The following invariant is satisfied in the \textbf{for} loop in \Cref{Alg:main} (formally proved in \Cref{lemma:invariant-marginsample}). 
The correctness of the \newsample{} subroutine is guaranteed by this invariant.

\begin{condition}[invariant for \newsample{}]\label{inputcondition-magin}
The following holds for the input tuple $(\Phi, \sigma, v)$:
\begin{itemize}
\item $\Phi=(V,\+{Q},\+{C})$ is a CSP formula, $\sigma\in\+{Q}^*$ is a feasible partial assignment, and  $v\in V$ is a variable;
\item $v$ is not $\sigma$-fixed and $\sigma(v)=\hollowstar$, and for all $u\in V$, $\sigma(u)\in Q_u\cup\{\hollowstar\}$;
\item $\mathbb{P}[\neg c\mid \sigma]\leq \pprime q$ for all $c\in \mathcal{C}$.
\end{itemize}
\end{condition}

The correctness of \Cref{Alg:main} follows from the correctness of \newsample{} and \rejsamp{}   for sampling from the correct marginal distributions, which is formally proved in \Cref{maincor}.

In fact, the sampling in \Cref{Alg:main} is \emph{perfect}.
It will only become approximate after the oracle in \Cref{assumption:frozen-oracle} realized by a Monte Carlo program that may bias the sampling.


\subsection{The rejection sampling}\label{sec:rejection-sampling}
We first introduce the \rejsamp{}, which is a standard procedure.
Our rejection sampling takes advantages of simplification and decomposition of a CSP formula. 

A simplification of $\Phi=(V,\+{Q},\+{C})$ under partial assignment $\sigma\in \qs$, denoted by $\Phi^\sigma=(V^\sigma,\+{Q}^\sigma,\+{C}^\sigma)$, 
is a new CSP formula such that $V^\sigma=V\setminus \Lambda(\sigma)$ and $\+{Q}^\sigma=\+{Q}_{V\setminus \Lambda(\sigma)}$, 
and the $\+{C}^\sigma$ is obtained from $\+{C}$ by: 
\begin{enumerate}
\item
removing all the constraints that have already been satisfied
 by $\sigma$;
\item
for  the remaining constraints, 
replacing the variables $v\in \Lambda(\sigma)$ with their values $\sigma(v)$.
\end{enumerate}
It is easy to see that $\mu_{\Phi^\sigma}=\mu^{\sigma}_{V\setminus\Lambda(\sigma)}$ for the uniform distribution  $\mu_{\Phi^\sigma}$ over  satisfying assignments of $\Phi^\sigma$.

%

A CSP formula $\Phi = (V, \+{Q}, \+{C})$ can be naturally represented as a (multi-)hypergraph $H_{\Phi}$,
where each variable $v\in V$ corresponds to a vertex in $H_{\Phi}$ and each constraint $c\in\+{C}$ corresponds to a hyperedge $\vbl(c)$ in $H_{\Phi}$.
We slightly abuse the notation and write $H_{\Phi}=(V, \+{C})$.

%
Let $H_{i}=(V_i, \+{C}_i)$ for $1\le i\le K$ denote all $K\ge 1$ connected components in $H_{\Phi}$, and $\Phi_i=(V_i, \+{Q}_{V_i}, \+{C}_i)$ their formulas.
Obviously $\Phi=\Phi_1\land\Phi_2\land \dots \land \Phi_K$ with disjoint $\Phi_i$, and $\mu_{\Phi}$ is the product of all $\mu_{\Phi_i}$.
%
Also $\mu_S$ on a subset of variables $S\subseteq V$ is determined only by those components with $V_i$ intersecting $S$.

For each $\sigma\in \qs$ and $v\in V^{\sigma}$, let $H_v^{\sigma}=(V_v^\sigma,\+{C}_v^{\sigma})$ denote the connected component in $H^{\sigma}$ that contains the vertex/variable $v$. This definition will be useful later.

\begin{algorithm}  
  \caption{\rejsamp{}$(\Phi,\sigma,S)$}  \label{Alg:rej}
  \KwIn{a CSP formula $\Phi=(V,\+{Q},\+{C})$, a feasible partial assignment $\sigma \in \+{Q}^\ast$, and  a subset $S \subseteq V\setminus\Lambda(\sigma)$ of unassigned variables  in $\sigma$;}  
  \KwOut{an assignment $X_S\in \+{Q}_S$ distributed as $\mu^\sigma_S$;}
  find all the connected components $\{H_i^\sigma= (V_i^\sigma,\+{C}_i^\sigma)\}\mid 1\leq i\leq K\}$ in $H_{\Phi^\sigma}$ s.t.~$V_i^\sigma$ intersects $S$, where $\Phi^\sigma$ denotes the simplification of $\Phi$ under $\sigma$\;\label{line-rs-find}
  \For{$1\leq i\leq K$ \label{line-rs-for}}
  {
    \Repeat{all the constraints in $\+{C}_i^\sigma$ are satisfied by $X_{V_i^\sigma}$\label{line-rs-until}}
     {  
        generate $X_{V_i^\sigma}\in \+{Q}_{V_i^\sigma}$ uniformly and independently at random\;\label{line-rs-sampling}
    }
  }
 
\textbf{return } $X_S$ where $X$ is the concatenation of all $X_{V_i^\sigma}$\;
\end{algorithm} 

Our rejection sampling algorithm for drawing from a marginal distribution $\mu_S^{\sigma}$ is given in \Cref{Alg:rej}.
The correctness of this algorithm is folklore. We state without proof.
\begin{theorem}\label{rejcorrect}
On any input $(\Phi,\sigma,S)$ as specified in \Cref{Alg:rej},
\rejsamp{} terminates with probability $1$, and upon termination it returns an assignment $X_S\in \+{Q}_S$ distributed as $\mu_S^\sigma$.
\end{theorem}

\subsection{The marginal sampler}\label{sec:marginal-sampler}
We now introduce the  the core part of our sampling algorithm, the $\newsample{}$ subroutine.
This procedure is a ``marginal sampler'': 
it can draw a random value for a variable $v\in V$ according to its marginal distribution $\mu_v^{\sigma}$.
%
%
Our marginal sampling algorithm is inspired by a recent novel sampling algorithm of Anand and Jerrum for infinite spin systems~\cite{anand2021perfect}.

For each variable $v\in V$, we suppose that an arbitrary order is  assumed over all values in $Q_v$; 
we use $q_v\triangleq\abs{Q_v}$ to denote the domain size of $v$ and fix the following  parameters:
\begin{align}
\theta_v \triangleq
\frac{1}{q_v}-\eta-\zeta
\quad\text{and}\quad
\theta\triangleq\frac{1}{q}-\eta-\zeta
\quad\text{ where}\quad 
\begin{cases}
\eta=\left(1-\mathrm{e}\pprime q\right)^{-\Delta}-1\\
\zeta=\left(16\mathrm{e}qk\Delta\right)^{-1}
\end{cases}
\label{eq:parameter-theta}
\end{align}
Note that $\zeta<\frac{1}{q}-\eta$ is guaranteed by the LLL condition in \eqref{eq:main-thm-LLL-condition}, and hence $\theta_v,\theta>0$.

The $\newsample{}$ subroutine for drawing from a marginal distribution $\mu_v^{\sigma}$ is given in \Cref{Alg:eq}.

\begin{algorithm}
  \caption{$\newsample{}(\Phi, \sigma ,v)$} \label{Alg:eq}
  \SetKwInput{KwPar}{Parameter}
  \KwIn{a CSP formula $\Phi=(V,\+{C})$, a feasible partial assignment $\sigma \in \+{Q}^\ast$, and a variable $v\in V$;} 
  \KwOut{a random $x\in Q_v$ distributed as $\mu^\sigma_v$;}
  choose $r\in [0,1)$ uniformly at random\;\label{Line-new-st}
  \eIf(\tcp*[f]{$r$ falls into the zone of local uniformity}\label{Line-if-eq}){$r<q_v\cdot \theta_v$}{\Return{} the $\lceil r/\theta_v \rceil$-th value in $Q_v$\;}(\tcp*[f]{$r$ falls into the zone of indecision}){\Return{$\recsample{}(\Phi,\Mod{\sigma}{v}{\star},v)$}\;\label{Line-rec-eq}}
\end{algorithm}

An invariant satisfied by \Cref{Alg:eq} guarantees that $\theta_v$ always lower bounds the marginal probability with gap $\zeta$.
%
This is formally proved in \Cref{sec:prelim} by a ``local uniformity'' property (\Cref{generaluniformity}).

\begin{proposition}\label{prop:theta-marginal-lower-bound}
Assuming \Cref{inputcondition-magin} for the input $(\Phi, \sigma ,v)$, 
it holds that $\min\limits_{x\in Q_v}\mu^{\sigma}_v(x)\ge  \theta_v+\zeta$.
\end{proposition}

Therefore, the function $\+{D}$ constructed below is a well-defined distribution over $Q_v$:
\begin{align}\label{eq:definition-margin-overflow}
\forall x\in Q_v,\qquad
\+{D}(x)\triangleq\frac{\mu_v^{\sigma}(x)-\theta_v}{1-q_v\cdot \theta_v}.
\end{align}


Consider the following thought experiment.
Partition $[0,1)$ into $(q_v+1)$ intervals $I_1,I_2\dots,I_{q_v}$ and $I'$, where $I_i=[(i-1)\theta_v,\, i\theta_v)$ for $1\le i\le q_v$ are of equal size $\theta_v$, and $I'=[q_v\cdot \theta_v,\,1)$ is the remaining part.
We call $\bigcup_{i=1}^{q_v}I_i=[0,\,q_v\cdot \theta_v)$ the ``zone of local uniformity" and $I'=[q_v\cdot \theta_v,\,1)$ the ``zone of indecision". 

Drawing from $\mu^\sigma_v$ can then be simulated as: first drawing a uniform random $r\in [0,1)$, if $r<q_v\cdot \theta_v$, i.e.~it falls into the ``zone of local uniformity", then returning the $i$-th value in $Q_v$ if $r\in I_i$; if otherwise $r\in I'$, i.e.~it falls into the ``zone of indecision", then returning a random value drawn from the above $\+{D}$.
It is easy to verify  that the generated sample is distributed as $\mu^\sigma_v$.
And this is exactly what \Cref{Alg:eq} is doing, assuming that the subroutine $\recsample{}(\Phi,\Mod{\sigma}{v}{\star},v)$ correctly draws from $\+{D}$.


\subsection{Recursive sampling for  margin overflow}\label{sec:margin-overflow-sampler}
The goal of the $\recsample{}$ subroutine is to draw from the distribution $\+{D}$ which is computed from the marginal distribution $\mu_v^{\sigma}$ as defined in \eqref{eq:definition-margin-overflow}.
%
%
Now suppose that we are given access to an oracle for drawing from $\mu_v^{\sigma}$ (such an oracle can be realized by $\rejsamp{}(\Phi,\sigma,\{v\})$ in \Cref{Alg:rej}).
Then, drawing from $\+D$ that is a linear function of $\mu_v^{\sigma}$, by accessing an oracle for drawing from $\mu_v^{\sigma}$, 
 can be resolved using the existing approaches of \emph{Bernoulli factory} \cite{Nacu2005FastSO,Hub16Bernoulli,Dughmi17Bernoulli}.

This sounds silly because if such oracle for $\mu_v^{\sigma}$ were efficient we would have been using it to output a sample for $\mu_v^{\sigma}$ in the first place, which is exactly the reason why we ended up trying to draw from  $\+{D}$.

Nevertheless, such Bernoulli factory for sampling from $\+{D}$ may serve as the basis of a recursion, 
where sufficiently many variables with enough ``freedom'' would have been sampled successfully in their zones of local uniformity during the recursion, 
and hence the remaining CSP formula would have been ``factorized'' into small connected components, 
in which case an oracle for  $\mu_v^{\sigma}$ would be efficient to realize by the $\rejsamp{}(\Phi,\sigma,\{v\})$, and the Bernoulli factory for $\+{D}$ could apply.

%

%
We define a class of variables that are candidates for sampling with priority in the recursion.


\begin{definition}[$\star$-influenced variables]\label{definition:boundary-variables}
Let $\sigma\in \+{Q}^*$ be a partial assignment.
Let $H^\sigma=H_{\Phi^{\sigma}}=(V^{\sigma},\+{C}^{\sigma})$ be the hypergraph for simplification $\Phi^{\sigma}$. 
%
Let $\hfix{\sigma}$ be the sub-hypergraph of $H^\sigma$ induced by $V^{\sigma}\cap\vfix{\sigma}$.
\begin{itemize}
\item
Let $\vcon{\sigma}\subseteq V^{\sigma}\cap\vfix{\sigma}$ be the set of vertices belong to the connected components in $\hfix{\sigma}$ that contain any $v$ with $\sigma(v)=\star$.
\item
Let $\vinf{\sigma}\triangleq\left\{u\in V^{\sigma}\setminus \vcon{\sigma}\mid \exists c\in \+{C}^{\sigma}, v\in \vstar{\sigma}:u,v\in\vbl(c)\right\}$ be the vertex boundary of $\vcon{\sigma}$ in~$H^\sigma$.
%
%
\item 
Let $\ccon{\sigma}$ be the set of constraints $c\in\+{C}$ that intersect $\vcon{\sigma}$.
\item 
Define $\nextvar{\sigma}$ by
\begin{align}
\nextvar{\sigma}
\triangleq
\begin{cases}
v_i\in \vinf{\sigma}\text{ with smallest $i$} & \text{if }\vinf{\sigma}\neq\emptyset,\\
\perp & \text{otherwise}.
\end{cases}
\label{eq:definition-var}
\end{align}
\end{itemize}
\end{definition}

\begin{remark}

The $\ccon{\sigma}$ defined above is not used here, but is important in the analysis.
Same as in \Cref{definition:frozen-fixed}, the $\vfix{\sigma}$ is defined with respect to the oracle in \Cref{assumption:frozen-oracle}.

In \Cref{subsection-var}, a dynamic data structure is given to efficiently compute $\nextvar{\sigma}$.
\end{remark}

With this construction of $\nextvar{\cdot}$,  the \recsample{} subroutine is described in \Cref{Alg:Recur}.

\begin{algorithm}
\caption{$\recsample{}(\Phi,\sigma,v)$} \label{Alg:Recur}
  \SetKwInput{KwPar}{Parameter}
\KwIn{a CSP formula $\Phi=(V,\+{C})$, a feasible partial assignment $\sigma\in \+{Q}^*$, and a variable $v\in V$;} 
\KwOut{a random $x\in Q_v$ distributed as $\+{D}\triangleq\frac{1}{(1-q_v\cdot \theta_v)}(\mu_v^{\sigma}-\theta_v)$;}
$u \gets \nextvar{\sigma}$ where $\nextvar{\sigma}$ is defined as in \eqref{eq:definition-var}\;\label{Line-u-rec}
\eIf{$u\neq\perp$\label{Line-if-1-rec}}
{
    choose $r\in [0,1)$ uniformly at random\;\label{Line-r-rec}
      \eIf(\tcp*[f]{$r$ falls into the zone of local uniformity}\label{Line-if-2-rec}){$r<q_u\cdot \theta_u$}{$\sigma(u) \gets$ the $\lceil r/\theta_u \rceil$-th value in $Q_u$\;}(\tcp*[f]{$r$ falls into the zone of indecision}\label{Line-else-rec}){$\sigma(u)\gets \recsample{}(\Phi,\Mod{\sigma}{u}{\star},u)$\;\label{line-rec-1-rec}}
\tcc{\Cref{Line-if-2-rec} to \Cref{line-rec-1-rec} together draw $\sigma(u)$ according to $\mu_u^{\Mod{\sigma}{u}{\star}}$}
    \Return $\recsample{}(\Phi,\sigma,v)$\; \label{line-rec-2-rec}
}
(\tcp*[f]{All non-$\sigma$-fixed variables are d/c.~from $v$ and ancestors.})
{
sample a random $x\in Q_v$ according to $\+{D}\triangleq\frac{1}{(1-q_v\cdot \theta_v)}(\mu_v^{\sigma}-\theta_v)$ using the \emph{Bernoulli factory} in \Cref{sec:bernoulli-factory} that accesses $\rejsamp{}(\Phi,\sigma,\{v\})$ as an oracle\;\label{Line-bfs-rec} 
\Return{$x$}\;
}
\end{algorithm}

Basically, a variable $u$ is a good candidate for sampling if it currently has enough ``freedom'' (since $u$ is not $\sigma$-fixed) and can ``influence'' the variables that we are trying to sample  in the recursion (which are marked by $\star$)
through a chain of constraints in the simplification of $\Phi$ under the current $\sigma$.
Such variables are enumerated by $\nextvar{\sigma}$.

The idea of \Cref{Alg:Recur} is simple.
In order to draw from the overflow distribution $\+{D}$ for a variable $v\in V$: 
if there is another candidate variable $u=\nextvar{\sigma}$ that still has enough freedom so its sampling might be easy, and also is relevant to the sampling at $v$ or its ancestors, 
we try to sample $u$'s marginal value first 
(hopefully within its zone of local uniformity and compensated by a recursive call for drawing from its margin overflow);
and if there is no such candidate variable to sample first, 
we finally draw  from $\+{D}$ 
using the Bernoulli factory.


The following invariant is satisfied by the \recsample{} subroutine called within the \newsample{} subroutine and the \recsample{} itself (formally proved in \Cref{lemma:invariant-marginsample}).

\begin{condition}[invariant for \recsample{}]\label{inputcondition-recur}
The following holds for the input tuple $(\Phi, \sigma, v)$:
\begin{itemize}
\item $\Phi=(V,\+{Q},\+{C})$ is a CSP formula, $\sigma\in\+{Q}^*$ is a feasible partial assignment, and  $v\in V$ is a variable;
\item $\sigma(v)=\star$;
\item $\mathbb{P}[\neg c\mid \sigma]\leq \pprime q$ for all $c\in \mathcal{C}$.
\end{itemize}
\end{condition}

The following marginal lower bound follows from  the ``local uniformity'' property (\Cref{generaluniformity})  in the same way as in \Cref{prop:theta-marginal-lower-bound} and is formally proved in \Cref{sec:prelim}.

\begin{proposition}\label{prop:theta-marginoverflow-lower-bound}
Assuming \Cref{inputcondition-recur} for the input $(\Phi, \sigma ,v)$, 
it holds that $\min\limits_{x\in Q_v}\mu^{\sigma}_v(x)\ge  \theta_v+\zeta$ and for $u=\nextvar{\sigma}$, if $u\neq\perp$ then it also holds that $\min\limits_{x\in Q_u}\mu^{\sigma}_u(x)\ge  \theta_u+\zeta$.
\end{proposition}


The Bernoulli factory used in \Cref{Alg:Recur} is achieved by a combination of existing constructions  (to be specific, the  Bernoulli factory for subtraction in \cite{Nacu2005FastSO}, composed with the linear Bernoulli factory in \cite{Hub16Bernoulli} and the Bernoulli race in \cite{Dughmi17Bernoulli}),
given access to an oracle for drawing from the marginal distribution $\mu_v^{\sigma}$, which is realized by $\rejsamp{}(\Phi,\sigma,\{v\})$ in \Cref{Alg:rej}.
This is formally stated by the following lemma.

\begin{lemma}[correctness of Bernoulli factory]\label{bercorrect}
Assuming \Cref{inputcondition-recur} for the input $(\Phi, \sigma ,v)$, 
there is a Bernoulli factory accessing $\rejsamp{}(\Phi,\sigma,\{v\})$ as an oracle
that terminates with probability $1$, 
and upon termination it returns a random $x\in Q_v$ distributed as the $\+{D}$ defined in \eqref{eq:definition-margin-overflow}.
\end{lemma}

The construction of the Bernoulli factory stated in above lemma is somehow standard, and is deferred to   \Cref{sec:bernoulli-factory},
where \Cref{bercorrect} is proved and the efficiency of the Bernoulli factory is also analyzed.



\section{Preliminary on Lov\'{a}sz Local Lemma}\label{sec:prelim}

The following is the asymmetric Lov\'{a}sz Local Lemma stated in the context of CSP.

\begin{theorem}[Erd\"{o}s and Lov\'{a}sz~\cite{LocalLemma}]\label{locallemma}
    Given a CSP formula $\Phi=(V,\+{Q},\+{C})$, if the following holds
    \begin{align}\label{llleq}
    \exists x\in (0, 1)^\+{C}\quad \text{ s.t.}\quad \forall c \in \+{C}:\quad
        {\mathbb{P}[\neg c]\leq x(c)\prod_{\substack{c'\in \+{C}\\ {\vbl}(c)\cap {\vbl}(c')\neq \emptyset}}(1-x(c'))},
    \end{align}
    then  
    $$
        {\mathbb{P}\left[ \bigwedge\limits_{c\in \+{C}} c\right]\geq \prod\limits_{c\in \+{C}}(1-x(c))>0},
    $$
\end{theorem}

When the condition \eqref{llleq} is satisfied, 
the probability of any event in the uniform distribution $\mu$ over all satisfying assignments can be well approximated by the probability of the event in the product distribution.
This was observed in \cite{haeupler2011new}:

\begin{theorem}[Haeupler, Saha, and Srinivasan \cite{haeupler2011new}]\label{HSS}
Given a CSP formula $\Phi=(V,\+{Q},\+{C})$, if $\eqref{llleq}$ holds, 
then for any event $A$ that is determined by the assignment on a subset of variables $\var{A}\subseteq V$, 
\[
    \Pr[\mu]{A}=\mathbb{P}\left[A\mid \bigwedge\limits_{c\in \+{C}}  c\right]\leq \mathbb{P}[A]\prod_{\substack{c\in \+{C}\\ \var{c}\cap \var{A}\neq \emptyset}}(1-x(c))^{-1},
\]
where $\mu$ denotes the uniform distribution over all satisfying assignments of $\Phi$ and $\mathbb{P}$ denotes the law of the uniform product distribution over $\+{Q}$.
\end{theorem}

The following ``local uniformity" property is a straightforward corollary to \Cref{HSS} by setting $x(c)=\mathrm{e}p$ for every $c\in \+{C}$ (and the lower bound is calculated by $\mu_v(x)=1-\sum_{y\in Q_v\setminus\{x\}}\mu_v(y)$). 

\begin{corollary}[local uniformity]\label{generaluniformity}
Given a CSP formula $\Phi=(V,\+{Q},\+{C})$, if $\mathrm{e}p\Delta<1$, then for any variable $v\in V$ and any value $x\in Q_v$, it holds that
\[
\frac{1}{q_v}-\left((1-\mathrm{e}p)^{-\Delta}-1\right)\leq \mu_v(x) \leq \frac{1}{q_v}+\left((1-\mathrm{e}p)^{-\Delta}-1\right).
\]
\end{corollary}



The following corollary implied by the ``local uniformity" property simultaneously  proves \Cref{prop:theta-marginal-lower-bound} and \Cref{prop:theta-marginoverflow-lower-bound}.
Recall $\pprime$ defined in \eqref{eq:parameter-p-prime} and $\theta_v, \zeta, \eta$ defined in \eqref{eq:parameter-theta}.

\begin{corollary}\label{localuniformitycor}
For any CSP formula $\Phi=(V,\+{Q},\+{C})$ and any partial assignment $\sigma\in\+{Q}^*$,
if
\[
\forall c\in \mathcal{C},\quad \mathbb{P}[\neg c\mid \sigma]\leq \pprime q,
\]  
then $\sigma$ is feasible, and for any $v\in V\setminus\Lambda(\sigma)$ and any $x\in Q_v$,
\[
 \theta_v+\zeta\le \mu^{\sigma}_v(x) \le \theta_v+2\eta+\zeta.
\]
\end{corollary}
\begin{proof}
Let $\Phi^\sigma=(V^\sigma,\+{Q}^\sigma,\+{C}^\sigma)$ be the simplification of $\Phi$ over $\sigma$, where $V^\sigma=V\setminus\Lambda(\sigma)$. 
We have  
$$\forall c\in \+{C},\quad \mathbb{P}_{\Phi^\sigma}[{\neg c}]=\mathbb{P}_{\Phi}[\neg c| \sigma]\leq \pprime q,$$
which means the simplified instance $\Phi^\sigma$ has violation probability $p_{\Phi^{\sigma}}\le \pprime q$.
By our choice of $\pprime$  in \eqref{eq:parameter-p-prime}, we still have $\mathrm{e}p_{\Phi^{\sigma}}\Delta_{\Phi^{\sigma}}<1$ where $\Delta_{\Phi^{\sigma}}\le \Delta$ is the constraint degree of simplified instance $\Phi^\sigma$.
Then by \Cref{locallemma}, $\Phi^\sigma$ is satisfiable, i.e.~$\sigma$ is feasible.

Note that the marginal distribution at $v$ induced by the $\mu_{\Phi^\sigma}$ over satisfying assignments of $\Phi^\sigma$ is precisely $\mu^{\sigma}_v$.
By \Cref{generaluniformity}, 
for any $v\in V^\sigma=V\setminus\Lambda(\sigma)$
and any $x\in Q_v$,
\[
\theta_v+\zeta
=
\frac{1}{q_v}-\eta
=
\frac{1}{q_v}-\left((1-\mathrm{e}\pprime q)^{-\Delta}-1\right)
\le
\mu^{\sigma}_v(x)
\le 
\frac{1}{q_v}+\left((1-\mathrm{e}\pprime q)^{-\Delta}-1\right)
=
\frac{1}{q_v}+\eta
=
\theta_v+2\eta+\zeta.
\]
\end{proof}

\section{Correctness of Sampling}\label{section-samp-cor}

In this section, we prove the correctness of \Cref{Alg:main}.
All theorems in this section assume 
the setting of parameters in  \eqref{eq:parameter-p-prime} and   \eqref{eq:parameter-theta},
and
the oracles in \Cref{definition:evaluation-oracle} and \Cref{assumption:frozen-oracle}.

We show that our main sampling algorithm \Cref{Alg:main} is correct.


\begin{theorem}\label{maincor}
On any input CSP formula $\Phi=(V,\+{Q},\+{C})$ satisfying \eqref{eq:main-thm-LLL-condition}, 
\Cref{Alg:main} terminates with probability $1$, and returns a uniform random satisfying assignment of $\Phi$ upon termination.
\end{theorem}
\begin{remark}[perfectness of sampling]
Note that the sampling in above theorem is \emph{perfect}: 
\Cref{Alg:main} returns a sample that is distributed exactly as the uniform distribution $\mu$ over all satisfying assignments of $\Phi$. 
Later in \Cref{section-samp-ef}, the oracle assumed in \Cref{assumption:frozen-oracle} will be realized by a Monte Carlo routine, which will further generalize the sampling algorithm to assume nothing beyond an evaluation oracle, in a price of a bounded bias introduced to the sampling.
\end{remark}

\begin{remark}[a weaker LLL condition]
The LLL condition \eqref{eq:main-thm-LLL-condition} is assumed mainly to guarantee the efficiency of the algorithm.
\Cref{maincor} in fact holds under a much weaker LLL condition:
\begin{align*}
2\mathrm{e}\cdot q^2\cdot p\cdot\Delta<1.
\end{align*}
Under this condition, there exists such choices of parameters $\pprime$ and $\zeta$  that satisfy
$p<\pprime<\frac{1}{2\mathrm{e}q^2\Delta}$ and $0<\zeta<\frac{1}{q}-\left(1-\mathrm{e}\pprime q\right)^{-\Delta}+1$.
For any such choice of parameters, the same analysis persists and \Cref{Alg:main} is as correct as claimed in \Cref{maincor}. 
%
\end{remark}

The following lemma guarantees  that the invariants in \Cref{inputcondition-magin} and \Cref{inputcondition-recur} 
are satisfied respectively by the inputs to \Cref{Alg:eq} and \Cref{Alg:Recur}.

\begin{lemma}\label{lemma:invariant-marginsample}
During the execution of \Cref{Alg:main} on a CSP formula $\Phi=(V,\+{Q},\+{C})$ satisfying \eqref{eq:main-thm-LLL-condition}:
\begin{enumerate}
\item\label{item:invariant-marginsample} 
whenever $\newsample{}(\Phi, X, v)$ is called, \Cref{inputcondition-magin} is satisfied by its input $(\Phi, X, v)$;
\item \label{item:invariant-marginoverflow}
whenever \recsample{}$(\Phi, \sigma, v)$ is called, \Cref{inputcondition-recur}  is satisfied by its input $(\Phi, \sigma, v)$.
\end{enumerate}
\end{lemma}

Before proving this lemma, 
we show that these invariants can already  imply the correctness of \newsample{},
which is critical for the correctness of the main sampling algorithm (\Cref{Alg:main}),
because the correctness of \rejsamp{} is standard (\Cref{rejcorrect}).


\begin{theorem}\label{sampcor}
The following holds for \Cref{Alg:eq} and \Cref{Alg:Recur}:
\begin{enumerate}
\item
Assuming \Cref{inputcondition-magin},  $\newsample{}(\Phi,\sigma,v)$ terminates with probability $1$, and  it returns a random value {$x\in Q_v$} distributed as $\mu^{\sigma}_v$ upon termination.
\item
Assuming \Cref{inputcondition-recur},  $\recsample{}(\Phi,\sigma,v)$ terminates with probability $1$, and upon termination  it returns a random value {$x\in Q_v$} distributed as the $\+{D}\triangleq\frac{\mu_v^{\sigma}-\theta_v}{1-q_v\cdot \theta_v}$ defined in \eqref{eq:definition-margin-overflow}.
\end{enumerate}
\end{theorem}

%
%

\begin{proof}
We verify the correctness of $\recsample{}$ by a structural induction. Then the correctness of \newsample{} follows straightforwardly.

Suppose that $\recsample{}$ is run on input $(\Phi,\sigma,v)$ satisfying \Cref{inputcondition-recur}.

The induction basis  is when $\nextvar{\sigma}=\perp$, in which case no further recursive calls to the \recsample{} are incurred.
In this case, due to the correctness of the Bernoulli factory  stated in \Cref{bercorrect} under \Cref{inputcondition-recur},
$\recsample{}(\Phi,\sigma,v)$ terminates with probability $1$ and
returns a random value $x\in Q_v$ distributed as $\+{D}$.

For the induction step, 
we assume that $\nextvar{\sigma}=u\in V$.
By the induction hypothesis, 
the recursive calls to $\recsample{}$ at \Cref{line-rec-1-rec} and \Cref{line-rec-2-rec} in \Cref{Alg:Recur} terminate with probability 1.
All other computations are finite. 
By induction, $\recsample{}(\Phi,\sigma,v)$ terminates with probability $1$. 

We then verify the correctness of sampling.
Let $W$ denote the value of $\sigma(u)$ generated in Lines \ref{Line-r-rec}-\ref{line-rec-1-rec} of in \Cref{Alg:Recur}.
It is easy to verify that for every $a\in Q_u$,
\begin{equation}
\begin{aligned}\label{eq-z-a}
\Pr{W=a}
= &\Pr{r<q_u\cdot \theta_u}\cdot \frac{1}{q_u}
 +  \Pr{r\geq q_u\cdot \theta_u}\cdot\Pr{\recsample(\Phi,\sigma_{u\gets \star},u) = a}
 \\=&\theta_u + \mu^\sigma_u(a)-\theta_u= \mu^\sigma_u(a),
\end{aligned}    
\end{equation}
where the second equality is due to the induction hypothesis (I.H.).
Thus, for every $a\in Q_u$,
\begin{align*}
\Pr{\recsample(\Phi,\sigma,v) = a} 
&=
\sum_{b\in Q_u}
\left(\Pr{W=b}\cdot\Pr{\recsample(\Phi,\Mod{\sigma}{u}{b},v) = a}\right)\\
\text{(by \eqref{eq-z-a} and  I.H.)}\qquad
&= 
\sum_{b\in Q_u}\left(
\mu^\sigma_u(b)\cdot
\frac{\mu^{\sigma_{u\gets b}}_v(a) - \theta_v}{1 -q_v\cdot \theta_v}\right)\\
&=
\frac{\mu^{\sigma}_v(a) - \theta_v}{1 -q_v\cdot \theta_v},
 \end{align*}
which means that the value returned by $\recsample{}(\Phi,\sigma,v)$ is distributed as $\+{D}$.
This finishes the induction and proves the correctness of \recsample{} assuming \Cref{inputcondition-recur}.

It is then straightforward to verify the correctness of $\newsample{}$ under \Cref{inputcondition-magin}.
Let $(\Phi,\sigma,v)$  be an arbitrary input to $\newsample{}$  satisfying \Cref{inputcondition-magin}.
It is easy to verify that $(\Phi,\Mod{\sigma}{v}{\star},v)$ satisfies \Cref{inputcondition-recur}, and hence the correctness of \recsample{} can apply.
Therefore, $\newsample{}(\Phi,\sigma,v)$ terminates with probability $1$ and for every $a\in Q_v$,
\begin{align*}
&\Pr{\newsample{}(\Phi,\sigma,v)=a}\\
= &\Pr{r<q_v\cdot \theta_v}\cdot \frac{1}{q_v}
 +  \Pr{r\geq q_v\cdot \theta_v}\cdot\Pr{\recsample(\Phi,\sigma_{v\gets \star},v) = a}\\
 =&\theta_v + \mu^{\sigma}_v(a)-\theta_v \\
= &\mu^{\sigma}_v(a).
\end{align*}
This shows the termination and correctness of $\newsample{}$.
\end{proof}

We then verify the invariant conditions claimed in \Cref{lemma:invariant-marginsample}.
Before that, we formally define the sequence of partial assignments that evolve in \Cref{Alg:main}.

\begin{definition}[partial assignments in \Cref{Alg:main}]\label{def-pas-main}
Let $X^0,X^1,\dots,X^n\in\+{Q}^*$ denote the sequence of partial assignments, 
where $X^0=\hollowstar^V$ 
and for every $1\le i\le n$, 
$X^i$ is the partial assignments $X$ in \Cref{Alg:main} after the $i$-th iteration of the \textbf{for} loop in Lines \ref{line-main-init}-\ref{line-main-sample}.
\end{definition}


\begin{fact}\label{fact-on-alg-main}
For each $1\le i\le n$, either $v_i$ is $X^{i-1}$-fixed, in which case $X^i=X^{i-1}$,
or otherwise, in which case $X^i$ extends $X^{i-1}$ by assigning $v_i$ a value in $Q_{v_i}$.
Consequently, for each $1\le i\le n$, if $v_i$ is not  $X^{i-1}$-fixed, then $X^*(v_i) = X^i(v_i)$, where $X^*$ denotes the output of \Cref{Alg:main}.
\end{fact}

\begin{lemma}\label{lemma:invariant-p-prime-q-bound}
For the $X^0,X^1,\dots,X^n$ in \Cref{def-pas-main}, it holds for all $0\le i\le n$ that $X^i$ is feasible and
\begin{align}
\forall c\in \+{C},\qquad \mathbb{P}[\neg c\mid X^i]< \pprime q. \label{eq:invariant-p-prime-q-bound}
\end{align}
\end{lemma}
\begin{proof}
We only need to prove \eqref{eq:invariant-p-prime-q-bound}. Then the feasibility of $X^i$ follows from \Cref{localuniformitycor}.

Fix any $c\in\+{C}$. Recall that assuming the LLL condition \eqref{eq:main-thm-LLL-condition}, we have $\pprime>p$ by \eqref{eq:parameter-p-prime}.
Then we have
\[
\mathbb{P}[\neg c\mid X^0]=\mathbb{P}\left[\neg c\mid \hollowstar^V\right]=\mathbb{P}[\neg c] \le p<\pprime.
\]
Suppose $i^*$ to be the smallest $0\le i\le n$ such that $\mathbb{P}[\neg c\mid X^{i}]\geq \pprime$. 
By $\mathbb{P}[\neg c\mid X^0]<\pprime$,
we have $i^*\geq 1$.
Then $\mathbb{P}[\neg c \mid  X^{i^*-1}] \leq \pprime$, which means $X^{i^*}\neq X^{i^*-1}$.
Combining with \Cref{def-pas-main} and \Cref{Alg:main}, 
we have $X^{i^*}$ extends $X^{i^*-1}$ by assigning $v_{i^*}$ some value in $Q_{v_{i^*}}$. Therefore,
\begin{align}\label{eq:invariant-q-amplification}
\mathbb{P}[\neg c \mid X^{i^*}]\leq \frac{\mathbb{P}[\neg c \mid X^{i^*-1}]}{\min\limits_{x\in Q_{v_{i^*}}}\mathbb{P}[v_{i^*} = x\mid  X^{i^*-1}]}< \pprime\abs{Q_{v_{i^*}}}\le \pprime q.
\end{align}
Thus, $\pprime\leq \mathbb{P}[\neg c\mid X^{i^*}]<\pprime q $.
Combining with \Cref{definition:frozen-fixed}, 
we have all variables in $\vbl(c)$ are $X^{i^*}$-fixed,
and thus the variables in $\vbl(c)$ will stay unchanged in \Cref{Alg:main}.
Therefore, if $i^*<n$, we have 
$X^{i^*+1}(v) = X^{i^*}(v)$ for each $v\in \vbl(c)$ 
and then $\pprime \leq \mathbb{P}[\neg c\mid X^{i^*+1}] = \mathbb{P}[\neg c\mid X^{i^*}]<\pprime q$.
Iteratively, one can show that 
$$\forall i^*\le i\le n, \quad
\pprime \leq \mathbb{P}[\neg c \mid X^{i}] <\pprime q.$$
In addition, by $i^*$ is the smallest $0\le i\le n$ such that $\mathbb{P}[\neg c\mid X^{i}]\geq \pprime$,
we have
$$\forall 0\leq i<i^*,\quad \mathbb{P}[\neg c \mid X^{i}] < \alpha < \pprime q.$$
Then \eqref{eq:invariant-p-prime-q-bound} is proved.
\end{proof}

The invariant of \Cref{inputcondition-magin} for \newsample{} stated in
\Cref{lemma:invariant-marginsample}-(\ref{item:invariant-marginsample})
follows easily from \Cref{lemma:invariant-p-prime-q-bound}. 
To prove 
the invariant of \Cref{inputcondition-recur} for \recsample{},
we show the following.
\begin{lemma}\label{lemma:general-invariant}
Assume \Cref{inputcondition-recur} for $(\Phi,\sigma,v)$.
For any $u\in V$, if $u$ is not $\sigma$-fixed, then $(\Phi,\Mod{\sigma}{u}{a},v)$
and $(\Phi,\Mod{\sigma}{u}{\star},u)$
satisfy \Cref{inputcondition-recur} for any $a\in Q_u\cup\{\star\}$.
\end{lemma}
\begin{proof}
%
The proofs for $(\Phi,\Mod{\sigma}{u}{\star},u)$ and $(\Phi,\Mod{\sigma}{u}{a},v)$ are similar.
Thus,
it suffices to that prove 
$(\Phi,\Mod{\sigma}{u}{a},v)$ satisfies \Cref{inputcondition-recur} for any $a\in Q_u\cup\{\star\}$.

The feasibility of $\Mod{\sigma}{u}{a}$ follows from \Cref{localuniformitycor}. 
Meanwhile, we also have $\Mod{\sigma}{u}{\star}(v) = \star$.
In the next, we prove that $\mathbb{P}[\neg c\mid \Mod{\sigma}{u}{\star}]\leq \pprime q$ for all $c\in \mathcal{C}$.
If $a=\star$, $\mathbb{P}[\neg c\mid \Mod{\sigma}{u}{\star}]\leq \pprime q$ holds trivially, because as a not $\sigma$-fixed variable, $u$ must have $\sigma(u)=\hollowstar$, and hence changing $\sigma$ to $\Mod{\sigma}{u}{\star}$ doe not change the  probability of any event conditioning on $\sigma$.
And for the case that $a\in Q_u$: 
if $c$ is $\sigma$-frozen then $u\not\in\vbl(c)$ since $u$ is not $\sigma$-fixed, and hence we have $\mathbb{P}[\neg c\mid \Mod{\sigma}{u}{a}] = \mathbb{P}[\neg c\mid \sigma]\leq \pprime q$, where the inequality is by that $(\Phi,\sigma,v)$ satisfies \Cref{inputcondition-recur}; and if otherwise $c$ is not $\sigma$-frozen, which means $\mathbb{P}[\neg c\mid \sigma]\le \pprime$, and hence as calculated in \eqref{eq:invariant-q-amplification}, we have
\[
\mathbb{P}[\neg c \mid \Mod{\sigma}{u}{a}] \le \abs{Q_{u}}\mathbb{P}[\neg c \mid {\sigma}]\le \pprime q,
\] 
which finishes the proof.
\end{proof}

The invariant of \Cref{inputcondition-recur} for \recsample{} stated in
\Cref{lemma:invariant-marginsample}-(\ref{item:invariant-marginoverflow})
follows from
Lemmas
\ref{lemma:invariant-marginsample}-(\ref{item:invariant-marginsample}) and \ref{lemma:general-invariant}.
Because by $(\Phi,\sigma,v)$ satisfies \Cref{inputcondition-magin}, we have 
$(\Phi,\sigma_{v\leftarrow \star},v)$ satisfies \Cref{inputcondition-recur}. In addition, during the execution of \recsample{}$(\Phi,\tau,v)$ where $\tau = \sigma_{v\leftarrow \star}$,  
the algorithm will only change an input partial assignment $\tau$ to $\Mod{\tau}{u}{a}$ for those vertices $u$ that are not $\tau$-fixed $u$ and for $a\in Q_{u}\cup\{\star\}$.
\Cref{lemma:invariant-marginsample} is proved.

%

Combining \Cref{lemma:invariant-marginsample} and \Cref{sampcor}, we prove the correctness of $\newsample{}$ (\Cref{Alg:eq}), assuming the LLL condition in \eqref{eq:main-thm-LLL-condition} for the input CSP in the main algorithm (\Cref{Alg:main}).

The correctness of $\rejsamp{}$ (\Cref{Alg:rej}) has already been established in \Cref{rejcorrect} (which is standard).

The correctness of the main sampling algorithm (\Cref{Alg:main}) then follows from the correctness of these two main subroutines.
Note that this is not trivial because in \Cref{Alg:main},  
the variables are chosen to draw from their marginal distributions adaptive to randomness.
We then formally prove that such being adaptive to randomness does not affect the correctness of sampling.

%


\begin{proof}[Proof of \Cref{maincor}]
By Lemmas \ref{rejcorrect} and \Cref{bercorrect}, \Cref{Alg:main} terminates with probability $1$. 

Let $X^0,X^1,\dots,X^n$ be the sequence of partial assignments defined in \Cref{def-pas-main}.
Let $X^*$ denote the output of \Cref{Alg:main}. We then show that, for every $\sigma\in \Omega$, $\Pr{X^*=\sigma}=\mu(\sigma)$.

Fix an arbitrary satisfying assignment $\sigma\in \Omega$. We further define a sequence of partial assignments $\sigma^0,\sigma^1,\dots,\sigma^n$ as follows.
Let $\sigma^0 = \hollowstar^{V}$.
For each $1\le i\le n$, if $v_i$ is $\sigma^{i-1}$-fixed, let 
$\sigma^{i} = \sigma^{i-1}$;
otherwise, let $\sigma^{i} = \sigma^{i-1}_{v_i\gets\sigma(v_i)} $. 
We claim that for every $0 \leq i\leq n$,
\begin{align}\label{eq-yi-zi}
\prod\limits_{j=1}^{i}\Pr{X^j=\sigma^j\mid X^{j-1}=\sigma^{j-1}} = \mu_{\Lambda(\sigma^{i})}(\sigma_{\Lambda(\sigma^{i})}),
\end{align}
with convention that both sides equal to 1 for $i=0$.

We then prove this claim by an induction on $i$. The basis with $i=0$ holds by the convention.

For the induction step,  we consider $i\geq 1$. 
By the consistency of the oracle in \Cref{assumption:frozen-oracle},  if $X^{i-1}=\sigma^{i-1}$, then $v_i$ can only be simultaneously fixed or non-fixed in both $X^{i-1}$ and $\sigma^{i-1}$.
\begin{itemize}
\item If $v_i$ is $\sigma^{i-1}$-fixed, then $\sigma^i=\sigma^{i-1}$, and by Lines \ref{line-main-if}-\ref{line-main-sample} of \Cref{Alg:main}, we have 
\[
\Pr{X^{i}=\sigma^{i}\mid X^{i-1}=\sigma^{i-1}} = 1. 
\]
Thus,  
\begin{align*}
\prod\limits_{j=1}^{i}\Pr{X^j=\sigma^j\mid X^{j-1}=\sigma^{j-1}}  
= &\prod\limits_{j=1}^{i-1}\Pr{X^j=\sigma^j\mid X^{j-1}=\sigma^{j-1}}\\
\text{(by I.H.)}\quad = &\mu_{\Lambda(\sigma^{i-1})}(\sigma_{\Lambda(\sigma^{i-1})})\\
\text{(since $\sigma^i=\sigma^{i-1}$)}\quad = &\mu_{\Lambda(\sigma^{i})}(\sigma_{\Lambda(\sigma^{i})}).
\end{align*}

\item If $v_i$ is not $\sigma^{i-1}$-fixed, then by the correctness of \newsample{} (guaranteed by \Cref{lemma:invariant-marginsample} and \Cref{sampcor}), we have 
$$\Pr{X^{i}=\sigma^{i}\mid X^{i-1}=\sigma^{i-1}} = \mu^{\sigma^{i-1}}_{v_i}(\sigma(v_i)).$$ Thus, we have

\begin{align*}
\prod\limits_{j=1}^{i}\Pr{X^j=\sigma^j\mid X^{j-1}=\sigma^{j-1}}
&= \mu^{\sigma^{i-1}}_{v_i}(\sigma(v_i))\cdot \prod\limits_{j=1}^{i-1}\Pr{X^j=\sigma^j\mid X^{j-1}=\sigma^{j-1}}\\
\text{(by I.H.)}\quad&= \mu^{\sigma^{i-1}}_{v_i}(\sigma(v_i))\cdot \mu_{\Lambda(\sigma^{i-1})}(\sigma_{\Lambda(\sigma^{i-1})})\\
\text{(chain rule)}\quad&=\mu_{\Lambda(\sigma^{i})}(\sigma_{\Lambda(\sigma^{i})}).
\end{align*}
\end{itemize}
This finishes the induction. The claim in \eqref{eq-yi-zi} is proved.

Observe that the sequence $X_0,X_1,\ldots,X_n,X^*$ is a Markov chain, where the last step $X^*$ is constructed from $X^n$ by \rejsamp{}.
Suppose that event $X^* = \sigma$. 
By \Cref{fact-on-alg-main} we have $X^i(v_i) = \sigma(v_i)$ if $v_i$ is not $X^{i-1}$-fixed.
Therefore according to \Cref{def-pas-main}, we have that 
if $v_i$ is $X^{i-1}$-fixed, then $X^i = X^{i-1}$;
and if otherwise $X^i = X^{i-1}_{v_i\gets \sigma(v_i)}$.
Thus, given that $X^* = \sigma$ occurs, 
by  the consistency of the oracle in \Cref{assumption:frozen-oracle}, one can verify that $X^i = \sigma^i$ for all $0\leq i\leq n$.
Thus, 
\begin{align*}
\Pr{X^* = \sigma}
&= \Pr{(X^* = \sigma)\land \left(\bigwedge\nolimits_{1 \leq i\leq n}(X^i = \sigma^i)\right)}\\
\text{(chain rule)}\quad
&= \Pr{X^* = \sigma \mid \forall 1\le  i\leq n,X^{i}=\sigma^{i}} \cdot \prod\limits_{i=1}^{n}\Pr{X^i=\sigma^i\mid \forall 0\leq j<i,X^{j}=\sigma^{j}}\\
\text{(Markov property)}\quad
&= \Pr{X^* = \sigma \mid X^{n}=\sigma^{n}} \cdot \prod\limits_{i=1}^{n}\Pr{X^i=\sigma^i\mid  X^{i-1}=\sigma^{i-1}}\\
\text{(by \eqref{eq-yi-zi})}\quad
&= \Pr{X^* = \sigma \mid X^n = \sigma^n} \cdot \mu_{\Lambda(\sigma^{n})}(\sigma_{\Lambda(\sigma^{n})})\\ 
\text{(by \Cref{rejcorrect})}\quad
&= \mu^{\sigma^{n}}_{V\setminus \Lambda(\sigma^n)}(\sigma_{V\setminus \Lambda(\sigma^n)}) \cdot\mu_{\Lambda(\sigma^{n})}(\sigma_{\Lambda(\sigma^{n})})\\
\text{(chain rule)}\quad
&= \mu(\sigma).\qedhere
\end{align*}
\end{proof}


\section{Efficiency of Sampling}\label{section-samp-ef}

In this section, we show the efficiency of \Cref{Alg:main} under the LLL condition in \eqref{eq:main-thm-LLL-condition}. 

\Cref{Alg:main} assumes accesses to the following oracles for a class of constraints $\+{C}$, 
both of which receive as input a constraint $c\in\+{C}$ and a partial assignment $\sigma\in\+{Q}^*$ upon queries:
\begin{itemize}
    \item 
    $\eval(c,\sigma)$: the evaluation oracle in \Cref{definition:evaluation-oracle}, 
    which 
    decides whether ${\mathbb{P}[c\mid \sigma]}=1$, that is, whether $c$ is already satisfied by $\sigma$;
    \item 
    $\checkf(c,\sigma)$: 
    the oracle for frozen decision in \Cref{assumption:frozen-oracle},
    which 
    distinguishes  between the two cases $\mathbb{P}[\neg c \mid \sigma]>\pprime$ and $\mathbb{P}[\neg c \mid \sigma]<0.99\pprime$,
    where $\pprime$ is the threshold defined in \eqref{eq:parameter-p-prime},
    and answers arbitrarily and consistently if otherwise.
\end{itemize}

The complexity of our sampling algorithm is measured in terms of the 
queries to the two oracles $\eval(\cdot)$ and $\checkf(\cdot)$, and the computation costs. 
We prove the following theorem.

\begin{theorem}\label{mainef}
Given as input a CSP formula $\Phi=(V,\+{Q},\+{C})$ satisfying \eqref{eq:main-thm-LLL-condition},
\Cref{Alg:main} in expectation costs $O(q^2k^2\Delta^{9}n)$ queries to $\eval{}(\cdot)$, $O(k\Delta^6n)$ queries to $\checkf{}(\cdot)$, and $O(q^3k^3\Delta^{9}n)$ in computation.
\end{theorem}
Together with the correctness of \Cref{Alg:main} stated in \Cref{maincor}, this proves the main theorem for perfect sampling (\Cref{thm:main-sampling-perfect}),
since any query to the oracle $\checkf(\cdot)$ can be resolved in $\poly(q,k)$ time
assuming the FPTAS for violation probability in the condition of \Cref{thm:main-sampling-perfect}.
%

\begin{remark}[\textbf{Monte Carlo realization of frozen decision}]
The oracle $\checkf{}(\cdot)$ can be realized probabilistically through the Monte Carlo method.
Upon each query on a constraint $c$ and a partial assignment $\sigma$,
the two extreme cases $\mathbb{P}[\neg c \mid \sigma]>\pprime$ and $\mathbb{P}[\neg c \mid \sigma]<0.99\pprime$ can be distinguished with high probability $(1-\delta)$
by independently testing for $O(\frac{1}{\pprime}\log\frac{1}{\delta})$ times
whether the constraint $c$ is satisfied by a randomly generated assignment over $\vbl(c)$ consistent with $\sigma$.
We further apply a memoization to guarantee the consistency of the oracle as required in \Cref{assumption:frozen-oracle}. 
The resulting algorithm is called \Cref{Alg:main}' and is formally described in \Cref{sec:efficiency-main-alg}.
\end{remark}

This Monte Carlo realization of the $\checkf{}(\cdot)$ oracle introduces a bounded bias to the result of sampling and turns the perfect sampler in \Cref{mainef} to an approximate sampler \Cref{Alg:main}', which no longer assumes any nontrivial machinery beyond evaluating constraints.
\begin{theorem}\label{mainef-approx}
Given as input an $\varepsilon\in(0,1)$ and a CSP formula $\Phi$ satisfying \eqref{eq:main-thm-LLL-condition}, \Cref{Alg:main}'  
in expectation costs $O\left(q^2k^2\Delta^{9}n\log\left(\frac{\Delta n}{\varepsilon}\right)\right)$ queries to $\eval{}(\cdot)$
and $O\left(q^3k^3\Delta^{9}n\log\left(\frac{\Delta n}{\varepsilon}\right)\right)$ in computation,
and outputs within $\varepsilon$ total variation distance from the output of  \Cref{Alg:main} on input $\Phi$.
\end{theorem}
Together with the correctness of \Cref{Alg:main} stated in \Cref{maincor}, this proves the main theorem (\Cref{thm:main-sampling}) of the paper.

\medskip

\noindent{\textbf{A notation for complexity bound}:\,\,}
Throughout the section, we adopt the following abstract notation for any complexity bound.
A complexity bound is expressed as a  formal bi-variate linear  function:
\begin{align}\label{eq:abstract-complexity-bound}
t(\Formal{x},\Formal{y})=\alpha\cdot \Formal{x}+\beta\cdot \Formal{y}+c,
\end{align}
where $\alpha$ represents the number of queries to $\eval{}(\cdot)$, 
$\beta$ represents the number of queries to $\checkf{}(\cdot)$, 
and $\gamma$ represents the computation costs.

%
For examples, The complexity bounds in  Theorems \ref{mainef} and  \ref{mainef-approx} are  thus expressed respectively as:
\[
O\left((q^2k^2\Delta^{9}n)\cdot \Formal{x}+(k\Delta^6n)\cdot \Formal{y}+q^3k^3\Delta^{9}n\right) 
\text{ and }
O\left(\left(q^2k^2\Delta^{9}n\log\left(\frac{\Delta n}{\varepsilon}\right)\right)\cdot \Formal{x}+q^3k^3\Delta^{9}n\log\left(\frac{\Delta n}{\varepsilon}\right)\right).
\]

We remark that such expression is only for notational convenience, because we want to handle three different complexity measures simultaneously in the analyses.
Throughout our analyses, only linear calculations will be applied to such functions $t(\Formal{x},\Formal{y})$. 
We further express: 
\[
\alpha\cdot \Formal{x}+\beta\cdot \Formal{y}+\gamma\le \alpha'\cdot \Formal{x}+\beta'\cdot \Formal{y}+\gamma'
\quad\iff\quad
\alpha\le \alpha'\land\beta\le \beta'\land\gamma\le \gamma'.
\]
%
And we write $t(0,0)$ for the constant term $\gamma$ in \eqref{eq:abstract-complexity-bound}, which stands for the computation cost.

\subsection{Input model and data structure}\label{subsection-inputmodel}
Besides being accessed through the two oracles $\eval(\cdot)$ and $\checkf(\cdot)$,
the input CSP formula  $\Phi=(V,\+{Q},\+{C})$ is presented to the algorithm as follows:
\begin{itemize}
\item
The variables in $V=\{v_1,v_2,\ldots,v_n\}$ and constraints $\+{C}=\{c_1,c_2,\ldots,c_m\}$
can be randomly accessed by their indices $i\in[n]$ and $j\in[m]$.
\item
Given any $c\in\+{C}$, the $\vbl(c)$ can be retrieved in time $O(k)$;
given any $v\in V$, the set of constraints $c$ with $v\in\vbl(c)$ can be retrieved in time $O(\Delta)$;
given any $c\in\+{C}$, the set of dependent $c'\in\+{C}$ with $\vbl(c')$ intersecting $\vbl(c)$ can be retrieved within time $O(\Delta)$.
\end{itemize}
These requirements can be met by representing the CSP $\Phi=(V,\+{Q},\+{C})$ in its bipartite incidence graph and also the dependency graph, both using the adjacency linked list data structures.

The partial assignment $\sigma\in\+{Q}^*$ is maintained by the algorithm 
in such a way that passing $\sigma$ to function as its argument takes $O(1)$ time.
This can be resolved by storing $\sigma$ globally as an array of $|V|$ stacks 
and passing a pointer to this array when $\sigma$ is passed as a function argument,
such that whenever a value $x$ is assigned to $\sigma(v)$, $x$ is pushed into the stack associated to $v$; and when a function returns it pops the stacks associated to those variables that it has updated in the current level of recursion.

The partial assignment $\sigma\in\+{Q}^*$ also keeps a linked list  
of the variables currently set as $\star$.

%

\subsection{The recursive cost tree (RCT)} 
A crucial step for proving \Cref{mainef} is the analysis of the \newsample{} (\Cref{Alg:eq}), which calls to the  recursive subroutine \recsample{} (\Cref{Alg:Recur}).

Consider an input $(\Phi,\sigma,v)$ satisfying \Cref{inputcondition-magin} such that $\newsample{}{(\Phi,\sigma,v)}$ is well-defined.
Our goal is to upper bound the following complexity.
\begin{definition}\label{def-tsamp}
Let $\tsamp{}(\Phi,\sigma,v)$ denote the expected cost of $\newsample{}{(\Phi,\sigma,v)}$ (\Cref{Alg:eq}). 
\end{definition}
%

There are two nontrivial tasks involved in computing the $\newsample{}{(\Phi,\sigma,v)}$: computing of the $\nextvar{\sigma}$ and the Bernoulli factory, both of which are in the \recsample{} (\Cref{Alg:Recur}).
\begin{definition}\label{def:tvar-tbfup}
Let $\tvar{}(\sigma)$ denote the cost for deterministically computing $\nextvar{\sigma}$ defined in \eqref{eq:definition-var}. 
%
Let $\tbfup(\sigma)$ be the expected cost for the Bernoulli factory in \Cref{Line-bfs-rec} of \Cref{Alg:Recur} in the worst case of $v$ such that \Cref{inputcondition-recur} is satisfied, if there exists such a $v$; and let $\tbfup(\sigma)=0$, if no such $v$ exists.
\end{definition}
\noindent 
The $\tbfup(\sigma)$ upper bounds the expected cost for the Bernoulli factory on well-defined input $(\Phi,\sigma,v)$. 

The above complexity bounds $\tsamp{}(\Phi,\sigma,v)$, $\tvar{}(\sigma)$, and $\tbfup(\sigma)$ are all expressed in the form of \eqref{eq:abstract-complexity-bound}.
%
The concrete bounds for $\tvar{}(\sigma)$ and $\tbfup(\sigma)$ are proved respectively {in Sections \ref{subsection-var} and \ref{sec:bernoulli-factory-cost}}. 

%
 We first introduce a combinatorial structure that relates $\tsamp{}(\Phi,\sigma,v)$ to $\tvar{}(\sigma)$ and $\tbfup(\sigma)$. For each $v\in V$, we further define $\qus{v}\defeq Q_v\cup \set{\star}$ as an extended domain for accessment.  
%

%
%

\begin{definition}[recursive cost tree]\label{RCTdef}
For any {$\sigma\in \qs$}, 
let $\+{T}_{\sigma}=(T_{\sigma},\rho_{\sigma})$, where $T_{\sigma}$ is a rooted tree with nodes $V(T_\sigma)\subseteq \qs$ and $\rho_\sigma :V(T_\sigma)\rightarrow [0,1]$ is a labeling of nodes in $T_\sigma$, be constructed as:
\begin{enumerate}
    \item The root of $T_{\sigma}$ is $\sigma$, with $\rho_{\sigma}(\sigma)=1$ and depth of $\sigma$ being 0; \label{RCT-1}
    \item for $i=0,1,\ldots$: for all nodes $\tau\in V(T_\sigma)$ of depth $i$ in the current $T_{\sigma}$,\label{RCT-2}
    \begin{enumerate}
        \item if $\nextvar{\tau}=\perp$, then leave $\tau$  as a leaf node in $T_\sigma$; \label{RCT-2-a}
        \item otherwise, supposed $u=\nextvar{\tau}$, append $\{\tau_{u\gets x}\mid x\in Q_u\cup\{\star\}\}$ as the $q_u+1$ children to the node $\tau$  in $T_\sigma$, and label them as: \label{RCT-2-b}
        \begin{align*}
        \forall x\in Q_u\cup\{\star\},\quad
        \rho_\sigma(\tau_{u\gets x})=
        \begin{cases}
        (1-{q_u\cdot \theta_u}) \rho_\sigma(\tau) &  \text{if }x=\star,\\
        \mu^\sigma_u(x)\cdot \rho_\sigma(\tau) & \text{if }x\in Q_u.
        \end{cases}
        \end{align*}
    \end{enumerate}
\end{enumerate}
The resulting $\+{T}_{\sigma}=(T_{\sigma},\rho_{\sigma})$ is called the \emph{recursive cost tree (RCT) rooted at $\sigma$}.
\end{definition}

Define the following function $\lambda(\cdot)$ on RCTs $\+{T}_{\sigma}=(T_{\sigma},\rho_{\sigma})$:
\begin{align}\label{eq-lambdadef}
    \lambda(\+{T}_{\sigma}) \defeq \sum_{\tau\in V(T_\sigma)}\left(\rho_\sigma(\tau)\cdot \tvar(\tau)\right)+\sum_{\text{leaves }\tau\text{ in } T_\sigma}\left(\rho_\sigma(\tau)\cdot \tbfup(\tau)\right).
\end{align}
Note that $\lambda(\+{T}_{\sigma})$ is expressed in the form of \eqref{eq:abstract-complexity-bound} as the $\tvar$ and $\tbfup$.
%
%

The expected complexity of \newsample{} is bounded through this function $\lambda(\+{T}_{\sigma})$.
\begin{lemma}\label{RCTtimecor}
For any input $(\Phi,\sigma,v)$ satisfying \Cref{inputcondition-magin}, 
it holds for  $\sigma^*=\Mod{\sigma}{v}{\star}$ that
\[
\tsamp{}(\Phi,\sigma,v)\le (1-q_v\cdot \theta_v)(\lambda(\+{T}_{\sigma^*})+\lambda(\+{T}_{\sigma^*})(0,0))+O(1),
\]
where $\lambda(\+{T}_{\sigma^*})(0,0)$ is the constant term in $\lambda(\+{T}_{\sigma^*})$ (standing for the computation cost as in \eqref{eq:abstract-complexity-bound}).
\end{lemma}

In the following, we prove \Cref{RCTtimecor}. 
The following recursive relation for RCT is easy to verify.

\begin{proposition}\label{fact-RCT}
Let $\sigma\in \qs$ and $u=\nextvar{\sigma}$.
If $u\neq\perp$, then
$$\lambda(\+{T}_{\sigma}) = \tvar(\sigma)+(1-q_u\cdot \theta_u)\lambda\left(\+{T}_{\Mod{\sigma}{u}{\star}}\right)
             +\sum_{x\in Q_u}\left(\mu^\sigma_u(x)\cdot \lambda\left(\+{T}_{\Mod{\sigma}{u}{x}}\right)\right).$$
\end{proposition}


Then we show that the complexity upper bound in \Cref{RCTtimecor} holds for the \recsample{}.

\begin{lemma}\label{RCTtime}
Let  $(\Phi,\sigma,v)$ be the input to \recsample{} (\Cref{Alg:Recur}) satisfying \Cref{inputcondition-recur},
and let $\trecor{}(\Phi,\sigma,v)$ denote the expected cost of $\recsample{}{(\Phi,\sigma,v)}$. 
It holds that 
\[
\trecor{}(\Phi,\sigma,v)\le \lambda(\+{T}_\sigma)+O(\lambda(\+{T}_\sigma)(0,0)),
\]
where $\lambda(\+{T}_\sigma)(0,0)$ is the constant term in $\lambda(\+{T}_\sigma)$.

\end{lemma}

\begin{proof}
For simplicity, for any partial assignment $\sigma\in \qs$, we let $\gamma(\sigma)=\lambda(\+{T}_\sigma)(0,0)$ denote the constant term in $\lambda(\+{T}_\sigma)$. Let $C>0$ denote the constant computation cost that dominates 
the costs for argument passing and all computations in Lines \ref{Line-if-1-rec}-\ref{Line-else-rec} of $\recsample(\Phi,\sigma,v)$.
It suffices to show that
\[
\trecor{}(\Phi,\sigma,v)\le \lambda(\+{T}_\sigma)+C\cdot \gamma(\sigma).
\]
We prove this by an induction on the structure of RCT. 

The base case is when $T_\sigma$ is just a single root,
in which case $\nextvar{\sigma}= \perp$, and  by \Cref{RCTdef},
\[
\lambda(\+{T}_{\sigma})=\rho_{\sigma}(\sigma)\cdot \tvar(\sigma)+\rho_{\sigma}(\sigma)\cdot \tbfup(\sigma)=\tvar(\sigma)+\tbfup(\sigma).
\]
Also if $\nextvar{\sigma}=\perp$, the condition in \Cref{Line-if-1-rec} of $\recsample{}(\Phi,\sigma,v)$ is unsatisfied and 
$$\trecor{(\Phi,\sigma,v)} \leq \tvar(\sigma)+\mathbb{E}[\tbf(\Phi,\sigma,v)]
\leq \tvar(\sigma)+\tbfup(\sigma)+C,$$
where $\tbf(\Phi,\sigma,v)$ represents the cost of the Bernoulli factory  in \Cref{Line-bfs-rec} of \Cref{Alg:Recur},
and by \Cref{def:tvar-tbfup}, it holds that  $\tbfup(\sigma)\ge \mathbb{E}[\tbf(\Phi,\sigma,v)]$ for all such $v$ that $(\Phi,\sigma,v)$ satisfies \Cref{inputcondition-recur}.
The base case is proved. 

For the induction step, we assume that $T_\sigma$ is a tree of depth $>0$. 
Thus by \Cref{RCTdef}, $\nextvar{\sigma}=u\neq \perp$ for some $u\in V$. 
According to Lines \ref{Line-r-rec}-\ref{line-rec-1-rec} of $\recsample(\Phi,\sigma,v)$, one can verify that for every $x\in Q_u$,
the probability that $\sigma(u) = x$ upon \Cref{line-rec-2-rec} is 
\begin{align*} 
\Pr{r<q_u\cdot \theta_u}\cdot \frac{1}{q_u}
 +  \Pr{r\geq q_u\cdot \theta_u}\cdot\Pr{\recsample(\Phi,\sigma_{u\gets \star},u) = x}
 = \mu^\sigma_u(x),
\end{align*}    
where the first equality is due to the correctness of \recsample{} guaranteed in \Cref{sampcor}.

By the law of total expectation, 
\begin{align}\label{eq-t-rec}
\trecor{(\Phi,\sigma,v)}
= & \tvar{}(\sigma) + (1 - q_u\cdot \theta_u)\cdot
 \trecor{(\Phi,\sigma_{u\gets \star},u)}
 + \sum_{x\in Q_u}\left(\mu^\sigma_u(x)\cdot \trecor{(\Phi,\sigma_{u\gets x},v)}\right)+C.
\end{align}    
Note that by \Cref{RCT-2-b} in \Cref{RCTdef}, for each $x\in Q_u\cup \{\star\}$, the subtree in $T_{\sigma}$ rooted by $\sigma_{u\gets x}$ is precisely the $T_{\tau}$ in the RCT  $\+{T}_{\tau}=(T_{\tau},\rho_{\tau})$ rooted at $\tau=\sigma_{u\gets x}$. 
By \Cref{lemma:general-invariant}, \Cref{inputcondition-recur} is still satisfied by $(\Phi,\sigma_{u\gets x},v)$.
Thus, by induction hypothesis, 
\begin{align*}
    \trecor{(\Phi,\sigma_{u\gets x},u)}\leq \lambda(\+{T}_{\sigma_{u\gets x}})+C\cdot\gamma(\sigma_{u\gets x}),
\end{align*}
where $\gamma(\sigma_{u\gets x})=\lambda(\+{T}_{\gamma(\sigma_{u\gets x})})(0,0)$ represents the constant term in $\lambda(\+{T}_{\gamma(\sigma_{u\gets x})})$.
Combined with \eqref{eq-t-rec},
\begin{align*}
\trecor{(\Phi,\sigma,v)}\leq & \tvar{}(\sigma) + (1 - q_u\cdot \theta_u)\left( \lambda(\+{T}_{\sigma_{u\gets \star}})+C\cdot\gamma(\sigma_{u\gets \star})\right)\\
&+ \sum_{x\in Q_u}\left(\mu^\sigma_u(x)\left(\lambda(\+{T}_{\sigma_{u\gets x}})+C\cdot\gamma(\sigma_{u\gets x})\right)\right)+C
\\
=&\lambda(\+{T}_{\sigma})+C\cdot \gamma(\sigma)-C(\gamma(\tvar(\sigma)))+C\\
\leq &\lambda(\+{T}_{\sigma})+C\cdot \gamma(\sigma),
\end{align*}    
where the equation is by \Cref{fact-RCT}, and $\gamma(\tvar(\sigma))=\tvar(\sigma)(0,0)\ge 1$ is the constant term in $\tvar(\sigma)$ that represents the computation cost for $\nextvar{\sigma}$.
\end{proof}

For $(\Phi,\sigma,v)$ satisfying \Cref{inputcondition-magin}, 
$(\Phi,\sigma_{v\gets \star},v)$ satisfies \Cref{inputcondition-recur},  hence 
\[
\tsamp{(\Phi,\sigma,v)} = (1-q_{v} \cdot\theta_{v})\trecor{(\Phi,\sigma_{v\gets \star},v)}+O(1)
\leq (1-q_v\cdot \theta_v)(\lambda(\+{T}_{\sigma_{v\gets \star}})+O(\gamma(\sigma_{v\gets \star})))+O(1),
\]
where the inequality holds by \Cref{RCTtime}.
This proves \Cref{RCTtimecor}.

\subsection{A random path simulating RCT}
Given a partial assignment $\sigma\in \qs$ and a variable $v\in V\setminus \Lambda(\sigma)$, define 
\begin{equation}
\begin{aligned}\label{eq-def-induceddist}
\induceddist{\sigma}{v}(\star)
&=\frac{1-q_v \cdot\theta_v}{2-q_v\cdot \theta_v}, \\
\forall x\in Q_v,\qquad
\induceddist{\sigma}{v}(x)
&=\frac{\mu_{v}^{\sigma}(x)}{2-q_v\cdot \theta_v}.
\end{aligned}
\end{equation}
Obviously, $\induceddist{\sigma}{v}(\cdot)$ is a well-defined probability distribution over $\qus{v}$.
The recursive cost tree defined in \Cref{RCTdef} inspires the following random process of partial assignments.

\begin{definition}[the $\pth(\sigma)$ process]\label{pathdef}
For any {$\sigma\in \qs$}, 
$\pth(\sigma)=(\sigma_0,\sigma_1,\ldots,\sigma_\ell)$ is a random sequence of partial assignments generated from the initial $\sigma_0 = \sigma$ as follows: for $i=0,1,\ldots$,
\begin{enumerate}
    \item 
    if $\nextvar{\sigma_i}=\perp$, the sequence stops at $\sigma_i$;\label{rp-1}
    \item
    otherwise $u=\nextvar{\sigma_i}\in V$, 
    the partial assignment $\sigma_{i+1}\in\qs$ is generated from $\sigma_{i}$ by randomly giving $\sigma(u)$ a value $x\in \qus{u}$, \label{rp-2}
    such that
    \begin{align*}
    \forall x\in \qus{u},\qquad
    \Pr{\sigma_{i+1}=\Mod{(\sigma_i)}{u}{x}}=\induceddist{\sigma}{u}(x).
    \end{align*}
\end{enumerate}
\end{definition}
The length $\ell(\sigma)$ of $\pth(\sigma) = \left(\sigma_0,\sigma_1,\dots,\sigma_{\ell(\sigma)}\right)$ is a random variable whose distribution is determined by $\sigma$.
We simply write $\ell=\ell(\sigma)$ and $\pth(\sigma) = \left(\sigma_0,\sigma_1,\dots,\sigma_{\ell}\right)$ if $\sigma$ is clear from the context.

It is quite obvious that $\pth(\sigma)$ satisfies the Markov property. 
In fact, $\pth(\sigma)$ can be seen as a Markov chain on space $\+{Q}^*$ such that any $\sigma$ with  $\nextvar{\sigma}=\perp$ has a self-loop with probability 1.

For any two partial assignments $\tau_1,\tau_2\in \qs$, define

\begin{align}\label{eq-chidef}%
\chi(\tau_1,\tau_2) \triangleq \prod\limits_{v\in \Lambda^{+}(\tau_1)\setminus \Lambda^{+}(\tau_2)}\left(2-q_v\cdot \theta_v\right),
\end{align}
where $\chi(\tau_1,\tau_2)=1$ for the case when $\Lambda^{+}(\tau_1)\setminus \Lambda^{+}(\tau_2)=\emptyset$ by convention. With a bit abuse of notation, given any partial assignment $\sigma\in \qs$, let $\pth(\sigma) = (\sigma_0,\sigma_1,\dots,\sigma_{\ell})$, we use $\chi(\sigma)$ to denote $\chi(\sigma_{\ell},\sigma)$.

The significance of the random process $\pth(\sigma)$ and the function $\chi(\cdot,\cdot)$ is that they are related to the complexity of \newsample{} through the following function: for any sequence  $P=(\sigma_0,\sigma_1,\dots,\sigma_{\ell})\in(\qs)^{\ell+1}$ with $\ell\ge 0$,
\begin{align}\label{eq-pathweightdef}
H(P)\defeq\sum\limits_{i=0}^{\ell}\left(\chi(\sigma_i,\sigma_0)\cdot \tvar(\sigma_i)\right)+\chi(\sigma_{\ell},\sigma_0)\cdot \tbfup{(\sigma_{\ell})},
\end{align}
where $\tvar(\cdot)$ and $\tbfup{(\cdot)}$ are defined in \Cref{def:tvar-tbfup}, expressed in form of \eqref{eq:abstract-complexity-bound}.

Recall the $\lambda(\+{T}_{\sigma})$ defined in \eqref{eq-lambdadef}. We have the following lemma.
\begin{lemma}\label{pathprop}
Let $\sigma\in \qs$ be a partial assignment. Then
\begin{align*}
   \E{H(\pth(\sigma))}\geq {\lambda(\+{T}_{\sigma})}.
\end{align*}
\end{lemma}

The following corollary follows immediately by combining \Cref{RCTtimecor} and \Cref{pathprop}.
\begin{corollary}\label{cor-t-msamp}
For any $(\Phi,\sigma,v)$ satisfying \Cref{inputcondition-magin}, 
let  $\sigma^*=\Mod{\sigma}{v}{\star}$,
it holds that 
\[
\tsamp{}(\Phi,\sigma,v)\le (1-q_v\cdot \theta_v)(\E{H(\pth(\sigma^*))} + O(\E{H(\pth(\sigma^*))}(0,0))+O(1).
\]
where $\E{H(\pth(\sigma^*))}(0,0)$ is the constant term in $\E{H(\pth(\sigma^*))}$ in the form of \eqref{eq:abstract-complexity-bound}.
\end{corollary}

\begin{proof}[Proof of \Cref{pathprop}]
We prove this lemma by an induction on the structure of RCT.
The base case is when $T_\sigma$ is just a single root,
in which case $\nextvar{\sigma}= \perp$, and  by \Cref{RCTdef},
\[
\lambda(\+{T}_{\sigma})=\rho_{\sigma}(\sigma)\cdot \tvar(\sigma)+\rho_{\sigma}(\sigma)\cdot \tbfup(\sigma)=\tvar(\sigma)+\tbfup(\sigma).
\]
Also, by $\nextvar{\sigma}= \perp$ and \Cref{pathdef} we have $\pth(\sigma)=(\sigma)$, and $\chi(\sigma_{\ell},\sigma_0)=1$.
Hence by \eqref{eq-pathweightdef}, 
\begin{align*}
    \E{H(\pth(\sigma))}=\tvar(\sigma)+\tbfup(\sigma)=\lambda(\+{T}_\sigma),
\end{align*}
The base case is proved.

For the induction step, we assume that $T_\sigma$ is a tree of depth $>0$.
Thus by \Cref{RCTdef}, $\nextvar{\sigma}=u\neq \perp$ for some $u\in V$ and $\ell(\sigma)\geq 1$. 
According to \Cref{rp-2} of \Cref{pathdef},
we have 
\begin{align}
     \forall x\in \qus{u},\quad \Pr{\sigma_{1} = \sigma_{u\gets x}}&=\induceddist{\sigma}{u}(x).\label{eq-x1-yi}
\end{align}
Moreover, by the Markov property, given $\sigma_1 = \sigma_{u\gets x}$ for each $x\in \qus{u}$, the subsequence  $(\sigma_1,\sigma_2,\dots,\sigma_{\ell(\sigma)})$  is identically  distributed as $\pth(\sigma_{u\gets x})$. In addition, it can be verified that for any  sequence of partial assignments $P=(\tau_0,\tau_1,\dots,\tau_{\ell})$ with $\ell\geq 1$ satisfying $\Pr{\pth(\sigma)= P}> 0$, 
\begin{align}
H(P)=(2-q_u\cdot \theta_u)H((\tau_1,\dots,\tau_{\ell})).\label{eq:weightfact}
\end{align}
Therefore, conditioning on $\sigma_1 = \sigma_{u\gets x}$ for each $x\in \qus{u}$, we have
\begin{equation}\label{eq-x1-yi-2-qtheta}
\begin{aligned}
 \E{H(\pth(\sigma))\mid \sigma_1=\sigma_{u\gets x}}&=\tvar(\sigma)+(2-q_u\cdot \theta_u)\E{H((\sigma_1,\ldots,\sigma_{\ell}))\mid \sigma_1=\sigma_{u\gets x}}\\
 &=\tvar(\sigma)+(2-q_u\cdot \theta_u)\E{H(\pth(\sigma_{u\gets x}))}.
 \end{aligned}
\end{equation}
Therefore by the law of total expectation, we have
\begin{equation}
\begin{aligned}\label{eq-e-sum-j}
\E{H(\pth(\sigma))}
=&\Pr{\sigma_1=\sigma_{u\gets \star}}\cdot \E{H(\pth(\sigma))\mid \sigma_1=\sigma_{u\gets \star}}\\
&+\sum\limits_{x\in Q_u}\left(\Pr{\sigma_1=\sigma_{u\gets x}}\cdot \E{H(\pth(\sigma))\mid \sigma_1=\sigma_{u\gets x}}\right)\\
=&\frac{1-q_u\cdot\theta_u}{2-q_u\cdot\theta_u}\cdot \left(\tvar(\sigma)+(2-q_u\cdot \theta_u)\E{H(\pth(\sigma_{u\gets \star}))}\right)\\
&+\sum\limits_{x\in Q_u}\left(\frac{\mu_{u}^{\sigma}(x)}{2-q_u\cdot\theta_u}\cdot \left(\tvar(\sigma)+(2-q_u\cdot \theta_u)\E{H(\pth(\sigma_{u\gets x}))}\right)\right)\\
=&\tvar(\sigma)+(1-q_u\cdot \theta_u)\E{H(\pth(\sigma_{u\gets \star}))}+\sum\limits_{x\in Q_u}\left(\mu_{u}^{\sigma}(x)\cdot \E{H(\pth(\sigma_{u\gets x}))}\right)
\end{aligned}
\end{equation}
where the second equality is by (\ref{eq-x1-yi}) and (\ref{eq-x1-yi-2-qtheta}), and the last equality is by
$$\frac{1-q_u\cdot\theta_u}{2-q_u\cdot \theta_u} +  \sum_{x\in Q_u}\frac{\mu^\sigma_u(x)}{2-q_u\cdot \theta_u} = 1.$$
Note that by \Cref{RCT-2-b} in \Cref{RCTdef}, for each $x\in \qus{u}$, the subtree in $T_{\sigma}$ rooted by $\sigma_{u\gets x}$ is precisely the $T_{\tau}$ in the RCT  $\+{T}_{\tau}=(T_{\tau},\rho_{\tau})$ rooted at $\tau=\sigma_{u\gets x}$. 
By the induction hypothesis,
\begin{align*}
\E{H(\pth(\sigma_{u\gets x}))} \geq \lambda(\+{T}_{\sigma_{u\gets x}}).
\end{align*}
Combining with (\ref{eq-e-sum-j}), we have
\begin{equation*}
\begin{aligned}
\E{H(\pth(\sigma))}\geq \tvar(\sigma_0)+(1-q_u\cdot \theta_u)\cdot\lambda(\+{T}_{\sigma_{u\gets \star}}) +\sum\limits_{x\in Q_u}\left(\mu_{u}^{\sigma}(x)\cdot \lambda(\+{T}_{\sigma_{u\gets x}})\right)=\lambda(\+{T}_\sigma).
\end{aligned}
\end{equation*}
where the equality is by \Cref{fact-RCT}.
\end{proof}

\subsection{Refutation of bad path}\label{sec:witness-tree}
The partial assignment maintained in \Cref{Alg:main} evolves as a random sequence $X^0,X^1,\dots,X^n$ which was formally defined in \Cref{def-pas-main}.
The efficiency of the \newsample{} (\Cref{Alg:eq}) called within in \Cref{Alg:main} crucially relies on that its input partial assignments are generated as this random sequence.

We define a procedure $\simulator(\cdot)$ such that $\simulator(t)$ generates the prefix $({X}^0,{X}^1,\dots,{X}^t)$ of the random partial assignments  $X^0,X^1,\dots,X^n$ maintained in \Cref{Alg:main} defined in \Cref{def-pas-main}.
%
This is explicitly described in \Cref{Alg:main-eq}, which is defined just to facilitate the analysis.

\begin{algorithm}
  \caption{\simulator($1\le t\le n$)} \label{Alg:main-eq}
 $X^0\gets \hollowstar^V$\;
 \For{$i =1$ to $t$\label{line-main-eq-for}}
 {
    \eIf{$v_i$ is not $X^{i-1}$-fixed\label{line-main-eq-if}}
    {
    choose $r\in [0,1)$ uniformly at random\label{line-main-eq-r}\;
        identify the unique $b\in {Q_{v_i}}$ satisfying
        $\sum_{a<b}\mu^{X^{i-1}}_{v_i}(a) \leq r < \sum_{a\leq b}\mu^{X^{i-1}}_{v_i}(a)$\;
        $X^i\gets X^{i-1}_{{v_i}\gets{b}}$\;\label{line-main-eq-vi}
    }
    { $X^i\gets X^{i-1}$\;}
}
\textbf{return} $(X^0,X^1,\dots,X^{t})$\;
\end{algorithm} 


We may consider $\simulator(t-1)=({X}^0,{X}^1,\dots,{X}^{t-1})$, and ${X}_0^t$ constructed as ${X}_0^t=X^{t-1}_{v_t\gets\star}$ with probability $1-q_{v_{t}}\theta_{v_t}$ if $v_t$ is not $X^{t-1}$-fixed, which simulates what is passed to \recsample{}
and generates $\pth{}({X}_0^t)=(X^t_0,X^t_1,\ldots,X^t_{\ell})$. 

%

Formally, for $1\le t\le n$, we define the following random process:
\begin{equation}\label{eq:simulate-path}
\begin{aligned}
(X^0,X^1,\ldots,X^{t-1})
&\gets \simulator(t-1),\\
X^t_0
&\gets
\begin{cases}
X^{t-1}_{{v_{t}}\gets{\star}} & \text{if $v_t$ is not $X^{t-1}$-fixed and } r_t = 1,\\
X^{t-1} & \text{otherwise},
\end{cases}\\
(X^t_0,X^t_1,\ldots,X^t_{\ell})
&\gets
\pth(X^t_0).
\end{aligned}
\end{equation}
where $r_t$ is sampled from $\textrm{Bern}(1-q_v\cdot \theta_v)$ independently.

To bound the expected cost of $\newsample{}$, 
we will bound $\E{H(\pth(X^t_0))}$ where $H(\cdot)$ is defined in \eqref{eq-pathweightdef}.
But first, we give a witness for a certain kind of ``bad'' paths.

Recall the $\cfrozen{\sigma}$ defined in \Cref{definition:frozen-fixed} and the $\ccon{\sigma}$ in \Cref{definition:boundary-variables}.
Given any $U\subseteq V$ and $E\subseteq \+C$, we 
use $U \uplus E$ to denote the disjoint union $U\cup E$.

\begin{definition}[$\sigma$-bad variables, constraints and events]\label{def:cbad}
Let $t\in [n]$,
$U\subseteq V$, $E\subseteq \+C$,
$\sigma\in \+{Q}^*$ be a partial assignment,
and $T=U\uplus E$ be a subset of variables and constraints.
\begin{itemize}
\item
Define $\csfrozen{\sigma} \triangleq \cfrozen{\sigma}\cap \ccon{\sigma}$.
\item
Define $\vstar{\sigma}\triangleq \{v\in V\mid  \sigma(v) =\star \}$ to be the set of variables assigned as $\star$.
\item Let $\+{E}^{\sigma}_T$ be the event $\left( U=\vst{\sigma_{\ell}}\right)\land \left(E\subseteq \csfrozen{\sigma_\ell}\right)$ where $\sigma_\ell$ is the last partial assignment in the sequence $\pth(\sigma)$.
\item Let $\+{E}^{t}_T$ be the event $\left(U=\vst{X^t_{\ell}}\right)\land \left(E\subseteq \csfrozen{X^t_\ell}\right)$ where $(X^0,X^1,\ldots,X^{t-1},X^t_0,X^t_1,\ldots,X^t_{\ell})$ is constructed as in \eqref{eq:simulate-path}.
\end{itemize}
\end{definition}

Intuitively, $\vstar{\sigma}$, $\csfrozen{\sigma}$, $\+{E}^{\sigma}_T$ and $\+{E}^{t}_T$ provide witnesses for the deep recursion of \recsample{},
such that any ``bad'' \pth{} that  causes the inefficiency of \recsample{} also creates many ``bad'' variables in $\vstar{\sigma}$ and constraints in $\csfrozen{\sigma}$,
thus the events $\+{E}^{\sigma}_T$ and $\+{E}^{t}_T$  happen for some large enough $T$. 

We first present some basic properties along the $\pth(\sigma)$, including the monotonicity of several variable/constraint attributes, and the relation between the length of $\pth(\sigma)$ and the sizes of $\+C^{\sigma_{\ell}}_v$, $\ccon{\sigma_{\ell}}$, $ \vstar{\sigma_{\ell}}$ and $\csfrozen{\sigma_{\ell}}$. 
Recall that for each $v\in V^{\sigma}$,
$H_v^{\sigma}=(V_v^\sigma,\+{C}_v^{\sigma})$ defined in \Cref{sec:rejection-sampling}
denotes the connected component in $H^{\sigma}$ which contains the vertex/variable $v$. 

\begin{lemma}\label{cor-mono}
Let $\sigma\in \qs$ and $\pth(\sigma) = (\sigma_0,\sigma_1,\dots,\sigma_{\ell})$.
For every $0\leq i\leq j \leq \ell$,
it holds that 
\noindent
\begin{flalign*}
&\text{(monotonicity property)}\hspace{111pt}
\vst{\sigma_i}\subseteq \vst{\sigma_j}, \quad
\+{C}^{\sigma_i}_{\+{P}}\subseteq \+{C}^{\sigma_j}_{\+{P}},&&
\end{flalign*}
\noindent
where $\+{P}$ can be any attribute $\+{P}\in\{\, \mathsf{frozen},\,\, \star\text{-}\mathsf{con},\,\, \star\text{-}\mathsf{frozen}\,\}$.

Moreover, if there is exactly one variable $v\in V$ having $\sigma(v) = \star$, 
it holds that
\begin{flalign*}
&\text{(upper bound on }\abs{\+C^{\sigma_{\ell}}_v}\text{ and }\abs{\ccon{\sigma_{\ell}}}\text{)}\hspace{53pt}
\abs{\+C^{\sigma_{\ell}}_v} \leq  \abs{\ccon{\sigma_{\ell}}}\leq \Delta \cdot \left(\abs{\vstar{\sigma_{\ell}}}+\abs{\csfrozen{\sigma_{\ell}}} \right),&&
\end{flalign*}
\begin{flalign*}
&\text{(upper bound on length of }\pth(\sigma)\text{)}\hspace{56pt}
\ell\le k\Delta\cdot \left(\abs{ \vst{\sigma_{\ell}}}+ \abs{\csfrozen{\sigma_{\ell}}}\right).&&
\end{flalign*}
\end{lemma}
\noindent
The proof of \Cref{cor-mono} is through a careful verification of definitions, and is deferred to \Cref{sec:monotonicity-proofs}.

We now state two technical lemmas that are crucial for the analysis of the efficiency of \newsample{} and \rejsamp{}.
These two lemmas essentially provide 
 tail bounds for the witnesses of ``bad" paths, which state that 
a large witness of the ``bad" path generated as in~\eqref{eq:simulate-path} occurs with exponentially small probability. 
We first state the tail bound used in the analysis of \newsample{}.

\sloppy
\begin{lemma}\label{thm-expect-depth}
Assume $8\mathrm{e}p\Delta^3\leq 0.99\pprime$, where $\pprime$ is defined as in \eqref{eq:parameter-p-prime}.
Let $1\le t\le n$. Let $(X^0,X^1,\ldots,X^{t-1},X^t_0,X^t_1,\ldots,X^t_{\ell})$ be generated as in \eqref{eq:simulate-path}.
For any integer $i\geq 0$,
\[\Pr{\abs{\vstar{X^t_{\ell}}}+\abs{\csfrozen{X^t_{\ell}}} \geq i\Delta} \cdot \mathbb{E}\left[\chi\left(X^t_{0}\right)\mid \abs{\vstar{X^t_{\ell}}}+ \abs{\csfrozen{X^t_{\ell}}} \geq i\Delta\right]\leq 2^{-i}.
\]
\end{lemma}

The following tail bound is used in the analysis of \rejsamp{}.

\begin{lemma}\label{thm-expect-depth2}
Assume $8\mathrm{e}p\Delta^3\leq 0.99\pprime$. Let $(X^0,X^1,\dots,X^n)=\simulator(n)$.
For any integer $i\geq 0$,
\[\Pr{\abs{\+C^{X^n}_v}\geq 2i\Delta^2} \leq 8\mathrm{e}k\cdot 4^{-i}.
\]
\end{lemma}
\noindent
Lemmas \ref{thm-expect-depth} and \ref{thm-expect-depth2} are proved through similar arguments, which consist of the following two steps:
\begin{enumerate}
\item showing an exponential tail bound over the occurrences of ``bad" variables and \emph{disjoint} ``bad" constraints, which is done by careful analyses of \simulator{} and \pth{},
\item boosting the above basic tail bound to the form required as in \Cref{thm-expect-depth}, which is done by using a newly invented combinatorial structure named generalized $\{2,3\}$-tree.
\end{enumerate}
The formal proofs of \Cref{thm-expect-depth} and \Cref{thm-expect-depth2} are technically involved and are deferred to \Cref{sec:g23-tree}.

\subsection{Efficiency of \newsample{}}\label{sec:efficiency-marginsample}
We now prove the following upper bound on the expected running time of \newsample{} (\Cref{Alg:eq}),
which is expressed in the form of \eqref{eq:abstract-complexity-bound}.

Let $\rtsamp{}(\Phi,\sigma,v)$ be the random variable that represents the complexity of $\newsample{}(\Phi, \sigma, v)$ when \Cref{inputcondition-magin} is satisfied by $(\Phi, \sigma, v)$. 
Note that $\tsamp{}(\Phi,\sigma,v)=\mathbb{E}[\rtsamp{}(\Phi,\sigma,v)]$ by \Cref{def-tsamp}.
%

\begin{theorem}\label{thm-time-MSamp}
Assume $8\mathrm{e}p\Delta^3\leq 0.99\pprime$. 
Let $X^0,X^1,\ldots,X^n$ be the random sequence in \Cref{def-pas-main}.
Assume the convention that $\rtsamp{(\Phi,X^{t-1},v_t)}=0$ when $v_t$ is $X^{t-1}$-fixed.
For any $1\le t\le n$, 
\begin{align*}
\E{\rtsamp{(\Phi,X^{t-1},v_t)}}
   \leq 
   O\left(q^2k^2\Delta^{9}\right)\cdot \Formal{x}+24k\Delta^7 \cdot \Formal{y}+O\left(q^3k^3\Delta^{9}\right),
\end{align*}
where expectation is taken over both $X^{t-1}$ and the randomness of \newsample{} algorithm.
\end{theorem}
\noindent
The convention is safe to apply since $\newsample{}(\Phi, \sigma, v)$ is never called when $v$ is not $\sigma$-fixed.

To prove this theorem, one need to give concrete bounds on $\tvar(\sigma)$ and $\tbfup(\sigma)$
(\Cref{def:tvar-tbfup}), respectively for the two nontrivial steps in the algorithm:
the $\nextvar{\sigma}$ called at \Cref{Line-u-rec} and the Bernoulli factory called at \Cref{Line-bfs-rec}, both in \Cref{Alg:Recur}.

%

\subsubsection{Cost of $\nextvar{\sigma}$}\label{subsection-var}
We give an explicit bound on the complexity $\tvar(\sigma)$ (\Cref{def:tvar-tbfup}) for computing the $\nextvar{\sigma}$ in \Cref{Line-u-rec} of \Cref{Alg:Recur}, in terms of the size of ${\ccon{\sigma}}$ (\Cref{definition:boundary-variables}).

Assume the input model and data structures in \Cref{subsection-inputmodel}.
We have the  following result. 
\begin{proposition}\label{lemma:efficiency-var}
For any $\sigma \in \+{Q}^\ast$, 
$\nextvar{\sigma}$ can be computed using at most $\abs{\ccon{\sigma}}$ queries to $\eval{}(\cdot)$, $\Delta\abs{\ccon{\sigma}}$ queries to $\checkf{}(\cdot)$, and $O\left(k^2\Delta\abs{\ccon{\sigma}}^2+1\right)$ computation cost, 
that is
\[
\tvar(\sigma)\le 
\abs{\ccon{\sigma}}\cdot \Formal{x}+\Delta\abs{\ccon{\sigma}}\cdot \Formal{y}+O\left(k^2\Delta\abs{\ccon{\sigma}}^2+1\right).
\]
\end{proposition}

As assumed in \Cref{subsection-inputmodel}, 
a linked list is kept alongside with the partial assignment $\sigma$ for storing the variables set as $\star$ in $\sigma$. 
Recall that in \Cref{def:cbad} we use $\vstar{\sigma}$ to denote such set of variables.
Recall the simplification $\Phi^{\sigma}=(V^\sigma,\+{Q}^\sigma,\+{C}^{\sigma})$ of $\Phi$ under $\sigma$ and its corresponding hypergraph representation $H^{\sigma}=H_{\Phi^{\sigma}}=(V^\sigma,\+{C}^{\sigma})$,
which is formally defined in \Cref{sec:rejection-sampling} and used in the definition of $\nextvar{\sigma}$.
%

The procedure for computing $\nextvar{\sigma}$ is straightforward on the hypergraph $H^{\sigma}$:
\begin{itemize}
\item 
Perform a depth-first search starting from $\vstar{\sigma}$ 
on the sub-hypergraph of $H^{\sigma}$ induced by $V^{\sigma}\cap \vfix{\sigma}$ 
to find the connected components $\vcon{\sigma}\supseteq \vstar{\sigma}$. 
\item
Construct the vertex boundary of $\vcon{\sigma}$ in $H^{\sigma}$.
Return  the first boundary vertex $v_i$ with smallest $i$ if such vertex exists, 
and return $\perp$ if otherwise.
\end{itemize}
We then verify that this procedure can be implemented within the complexity in \Cref{lemma:efficiency-var}.

First, observe that the complexity for testing whether a $c\in \+{C}$
belongs to the sub-hypergraph of $H^{\sigma}$ induced by $V^{\sigma}\cap \vfix{\sigma}$, 
is bounded by $\Formal{x}+\Delta\cdot \Formal{y}+O(k\Delta)$,
i.e.~one query to $\eval(\cdot)$, $\Delta$ queries to $\checkf(\cdot)$, and $O(k\Delta)$ computation cost.
This is because it is equivalent to check whether $c$ is satisfied by $\sigma$ and $\var{c}\subseteq \vfix{\sigma}$:
the former takes one query to $\eval(\cdot)$; 
and the latter can be resolved by enumerating all $c'\in\+{C}$ such that $\vbl(c)\cap\vbl(c')\neq\emptyset$ and retrieving $\vbl(c')$
(which costs $O(k\Delta)$ in computation), and checking whether $c'$ is $\sigma$-frozen for all such $c'$ (which takes $\le \Delta$ queries to $\checkf(\cdot)$ in total). 

It is not difficult to verify that in above depth-first search,
the set of constraints that need to be checked whether belong to the induced sub-hypergraph is in fact just the set $\ccon{\sigma}$. 
Then the complexity contributed by testing the membership of constraints in the induced sub-hypergraph is bounded by $\abs{\ccon{\sigma}}\cdot\Formal{x}+\Delta\abs{\ccon{\sigma}}\cdot \Formal{y}+O(k\Delta\abs{\ccon{\sigma}})$.
And this is the only part that may access the oracles,
so we have the respective bounds on the queries to the oracles $\eval(\cdot)$ and $\checkf(\cdot)$.

For other computation costs,
in the depth-first search, 
the sets of variables and constraints that have been visited 
can be straightforwardly  stored using two dynamic arrays, one for variables and the other for constraints. 
Querying if some variable/constraint has been visited or updating their status, 
is done by iterating over the entire array, 
which takes linear time in the current size of the dynamic array each time a query or an update is conducted.
Note that the number of visited constraints is at most $\abs{\ccon{\sigma}}$
and hence the number of visited variables is at most $k\abs{\ccon{\sigma}}$. 
Therefore this part costs $O(k^2\abs{\ccon{\sigma}}^2)$ in computation in total. 
Overall, the computation cost is easily dominated by $O(k^2\Delta\abs{\ccon{\sigma}}^2+1)$, where the additional $O(1)$ is meant to deal with the degenerate case of $\abs{\ccon{\sigma}}=0$.

\subsubsection{Cost of Bernoulli factory}\label{sec:bernoulli-factory-cost}
Here we state a complexity bound for the Bernoulli factory in \Cref{Line-bfs-rec} of \Cref{Alg:Recur}
that is useful in our analysis of \newsample.

We have the following bound on the $\tbfup(\sigma)$ (\Cref{def:tvar-tbfup}) for the Bernoulli factory.

\begin{proposition}\label{general-bfs-bound-component-invariant}
There exist constants $C_0,C_1>0$ such that the following holds.
Let $1\le t\le n$ and let $\pth(X_0^t)=(X_0^t,X_1^t,\ldots,X^{t}_{\ell})$ be generated as in \eqref{eq:simulate-path}, then 
$$\tbfup{}(X^t_{\ell})\leq C_1q^2k^2\Delta^6\left(\abs{\ccon{X^t_{\ell}}}+1\right) (1-\mathrm{e}\pprime q)^{-\abs{\ccon{X^t_{\ell}}}}\cdot \Formal{x}+
    C_0q^3k^3\Delta^6\left(\abs{\ccon{X^t_{\ell}}}+1\right) (1-\mathrm{e}\pprime q)^{-\abs{\ccon{X^t_{\ell}}}}.
$$
\end{proposition}
Note that $X^t_{\ell}$ is a random variable and the bound in \Cref{general-bfs-bound-component-invariant} holds for any possible $X^t_{\ell}$.

Next, we prove \Cref{general-bfs-bound-component-invariant}.
In \Cref{beref},
the following complexity bound is proved for all such $\sigma\in\qs$ 
where \Cref{inputcondition-recur} is satisfied by $(\Phi,\sigma, v)$ for some $v\in V$:
\begin{align}\label{eq:bfs-bound-component-invariant}
    \tbfup{}(\sigma)= 
    O\left(q^2k^2\Delta^6(\abs{\+{C}_v^{\sigma}}+1) (1-\mathrm{e}\pprime q)^{-\abs{\+{C}_v^{\sigma}}}\cdot \Formal{x}+
    q^3k^3\Delta^6(\abs{\+{C}_v^{\sigma}}+1) (1-\mathrm{e}\pprime q)^{-\abs{\+{C}_v^{\sigma}}}\right),
\end{align}
where  $\+{C}_v^{\sigma}$ denotes the set of constraints in the connected component that contain $v$ in $\Phi^{\sigma}$, the simplification of $\Phi$ under $\sigma$.

\begin{proof}[Proof of \Cref{general-bfs-bound-component-invariant}]
By the construction in \eqref{eq:simulate-path}, $X^t_0$ satisfies one of the two cases:
\begin{enumerate}
    \item $(\Phi,X^t_0,v_t)$ satisfies \Cref{inputcondition-recur} and $v_t$ is the only variable with $X^t_0(v_t)=\star$;
    \item $X^t_0(u)\in Q_u\cup \set{\hollowstar}$ for all $u\in V$.
\end{enumerate}
For case (1): 
By \Cref{cor-mono}, we have $\abs{\+C^{X^t_{\ell}}_v} \leq \abs{\ccon{X^t_{\ell}}}$  because  $X^t_\ell$ is generated by $\pth(X^t_0)=(X^t_0,X^t_1,\ldots,X^t_{\ell})$ obeying the Markov property.
Moreover, by \Cref{lemma:general-invariant} and the construction of \pth{}, 
it holds for sure  that 
\Cref{inputcondition-recur} is satisfied by $(\Phi,X^t_{\ell}, u)$ for some $u\in V$.
Therefore, the complexity bound in \eqref{eq:bfs-bound-component-invariant} always holds for $X^t_{\ell}$, which combined with the relation $\abs{\+C^{X^t_{\ell}}_v} \leq \abs{\ccon{X^t_{\ell}}}$ that we have just established, proves the case.

For case (2):
In this case by the definition of \pth{}, $\pth(X^t_0)=(X^t_0)$ and \Cref{inputcondition-recur} is no longer satisfied by $(\Phi,X_\ell^t,v)=(\Phi,X_0^t,v)$ for any $v$. Thus, $\tbfup{}(X^t_{\ell})=0$ due to \Cref{def:tvar-tbfup}.
%
%
%
%
%
\end{proof}

\subsubsection{Complexity bound for \newsample{}}
It remains to bound the complexity of \newsample{}. We give an upper bound on the weighted function $H$ defined in \eqref{eq-pathweightdef} for a random path.

\begin{lemma}\label{pathtimebound}
Assume $8\mathrm{e}p\Delta^3\leq 0.99\pprime$, where $\pprime$ is fixed as in \eqref{eq:parameter-p-prime}.
Let $1\le t\le n$ and let $\pth(X_0^t)=(X_0^t,X_1^t,\ldots,X^{t}_{\ell})$ be generated as in \eqref{eq:simulate-path}.
There exist constants $C,C'>0$ such that
\begin{align*}
  \mathbb{E}\left[H(\pth(X_0^t))\right]
  \leq Cq^2k^2\Delta^{9}\cdot \Formal{x}+24k\Delta^6 \cdot \Formal{y} + C'q^3k^3\Delta^{9}.
\end{align*}
\end{lemma}

\begin{proof}
According to \eqref{eq-pathweightdef}, it is sufficient to prove  
\begin{equation}\label{sum-i12-qtheta}
\begin{aligned}
\E{\sum\limits_{i=0}^{\ell}\left(\chi(X^t_i,X^t_0)\cdot \tvar(X^t_i)\right)}
\leq 24k\Delta^5 \cdot \Formal{x}+24k\Delta^6 \cdot \Formal{y} + 120C_1 k^3\Delta^8
\end{aligned}
\end{equation}
and 
\begin{equation}\label{eq-2-qtheta-ell-2}
\begin{aligned}
\E{\chi(X^t_{\ell},X^t_0)\cdot \tbfup{(X^t_{\ell})}}
\leq 20C_2 k^2q^2\Delta^{9}\cdot\Formal{x} + 20C_3k^3q^3\Delta^{9}.
\end{aligned}
\end{equation}
Therefore, by (\ref{eq-pathweightdef}) we have 
\begin{equation*}
\begin{aligned}
\mathbb{E}\left[H(\pth(X^t_0)))\right]&=\E{\sum\limits_{i=0}^{\ell}\left(\chi(X^t_{i},X^t_0)\cdot \tvar(X^t_i)\right)+\chi(X^t_{\ell},X^t_0)\cdot \tbfup{(X^t_{\ell})}}\\
&= \E{\sum\limits_{i=0}^{\ell}\left(\chi(X^t_{i},X^t_0)\cdot \tvar(X^t_i)\right)} + \E{\chi(X^t_{\ell},X^t_0)\cdot \tbfup{(X^t_{\ell})}}\\
&\leq 24k\Delta^5 \cdot \Formal{x}+24k\Delta^6 \cdot \Formal{y} +120C_1 k^3\Delta^8+20C_2k^2q^2\Delta^{9}\cdot\Formal{x}+ 20C_3k^3q^3\Delta^{9}  \\
&\leq (20C_2+24)k^2q^2\Delta^{9}\cdot \Formal{x}+24k\Delta^6 \cdot \Formal{y} + (120C_1 + 20C_3)k^3q^3\Delta^{9},
\end{aligned}
\end{equation*}
where in the second to the last inequality we apply (\ref{sum-i12-qtheta}) and (\ref{eq-2-qtheta-ell-2}).
The lemma follows by taking $C=20C_2+24$ and $C'=120C_1 + 20C_3$.

It is then sufficient to prove (\ref{sum-i12-qtheta}) and (\ref{eq-2-qtheta-ell-2}). We first prove (\ref{sum-i12-qtheta}). We claim that
\begin{equation}\label{eq-ccon-xtell-bound}
\abs{\ccon{X^{t}_{\ell}}} \leq \Delta \cdot  \left(\abs{\vstar{X^{t}_{\ell}}}+\abs{\csfrozen{X^{t}_{\ell}}}\right).
\end{equation}

By construction of $X^t_0$ in \eqref{eq:simulate-path}, 
either  $(\Phi,X^t_0,v_t)$ satisfies \Cref{inputcondition-recur} and $v_t$ is the only variable with $X^t_0(v_t)=\star$, in which case \eqref{eq-ccon-xtell-bound} follows from \Cref{cor-mono};
or $X^t_0$ satisfies for all $u\in V$, $X^t_0(u)\in Q_u\cup \set{\hollowstar}$, in which case $\pth(X^t_0)=(X^t_0)$, $\abs{\ccon{X^{t}_{\ell}}}=\abs{\csfrozen{X^{t}_{\ell}}}=\abs{\vstar{X^{t}_{\ell}}}=0$, and \eqref{eq-ccon-xtell-bound} holds. This proves the claim.

Let $L \defeq  \abs{\vstar{X^t_{\ell}}}+\abs{\csfrozen{X^t_{\ell}}}$. By the monotonicity property in \Cref{cor-mono}, $\ccon{X^t_i}\subseteq \ccon{X^t_{\ell}}$ for every $0\leq i\leq \ell$.
Combining with \eqref{eq-ccon-xtell-bound} we have 
\begin{align}\label{eq-abs-con-xti}
\abs{\ccon{X^t_i}}\leq \abs{\ccon{X^t_{\ell}}} \leq \Delta\left( \abs{\vstar{X^t_{\ell}}}+\abs{\csfrozen{X^t_{\ell}}}\right)\leq \Delta L.
\end{align}
Then by \Cref{lemma:efficiency-var}, there exists a constant $C_1>0$ such that for every $0\leq i\leq \ell$, 
\begin{align*}
\tvar{}(X^t_{i}) &\leq \abs{\ccon{X^t_{i}}}\cdot \Formal{x}+\Delta\abs{\ccon{X^t_{i}}}\cdot \Formal{y}+C_1\cdot k^2\Delta\abs{\ccon{X^t_{i}}}^2+C_1\\
&\le \Delta L\cdot \Formal{x}+\Delta^2 L\cdot \Formal{y}+C_1\cdot k^2\Delta^3L^2+C_1.
\end{align*}
Note that for each $0\leq i\leq \ell$ we have $\chi(X^t_i,X^t_0)\leq \chi(X^t_{\ell},X^t_0)=\chi(X^t_0)$. Hence,  
\begin{equation*}
\begin{aligned}
\E{\sum\limits_{i=0}^{\ell}\left(\chi(X^t_{i},X^t_0)\cdot \tvar(X^t_i)\right)}
&\leq \E{(\ell+1)\cdot \chi(X^t_0)\cdot\left(\Delta L\cdot \Formal{x}+\Delta^2 L\cdot \Formal{y}+C_1\cdot k^2\Delta^3L^2+C_1\right)}.
\end{aligned}
\end{equation*}
By the upper bound on length of $\pth(\sigma)$ in \Cref{cor-mono} and the special case that $\ell=0$ when $X^t_0$ satisfies that $X^t_0(u)\in Q_u\cup \set{\hollowstar}$ for all $u\in V$, we have $\ell \leq kL\Delta$. 
Thus, the following can be verified in separate cases 
$\ell>0$, $(\ell=0)\bigwedge (L=0)$, and $(\ell=0)\bigwedge (L>0)$:
\begin{align*}
 &(\ell+1)\cdot \chi(X^t_0)\cdot\left(\Delta L\cdot \Formal{x}+\Delta^2 L\cdot \Formal{y}+C_1\cdot k^2\Delta^3L^2+C_1\right)\\ 
 \leq  
 &2kL\Delta\cdot \chi(X^t_0)\cdot\left(\Delta L\cdot \Formal{x}+\Delta^2 L\cdot \Formal{y}+C_1\cdot k^2\Delta^3L^2+C_1\right).
\end{align*}
Therefore, we have 
\begin{equation}\label{eq-e-sum-i-ell-2-qtheta-2}
\begin{aligned}
&\E{\sum\limits_{i=0}^{\ell}\left(\chi(X^t_{i},X^t_0)\cdot \tvar(X^t_i)\right)}\\
\leq 
&\E{2kL\Delta\cdot \chi(X^t_0)\cdot\left(\Delta L\cdot \Formal{x}+\Delta^2 L\cdot \Formal{y}+C_1\cdot k^2\Delta^3L^2+C_1\right)}\\
\leq 
&\E{\chi(X^t_0)\cdot \left(L^2\cdot 2k\Delta^2 \cdot \Formal{x}+L^2\cdot 2k\Delta^3 \cdot \Formal{y}+L^3 \cdot 2C_1 k^3\Delta^4+L\cdot 2C_1k\Delta\right)}.
\end{aligned}
\end{equation}
Let $j \defeq \lfloor\frac{i}{\Delta}\rfloor$.
Let $\alpha\in\{1,2, 3\}$. We have
\begin{equation}\label{eq-pathtimebound-4}
\begin{aligned}
\E{\chi(X^t_0)\cdot L^{\alpha}}
&=\sum\limits_{i\geq 0}\left(\Pr{L=i}\cdot i^{\alpha}\cdot \E{\chi(X^t_0)\mid L=i}\right)\\
&\leq \Delta^{\alpha}\sum\limits_{i\geq 0}\left((j+1)^{\alpha}\cdot\Pr{L=i} \cdot \E{\chi(X^t_0)\mid L=i}\right)\\
&\leq 
\Delta^{\alpha}\sum\limits_{i\geq 0}\left((j+1)^{\alpha}\cdot \Pr{L\geq j\Delta}\cdot \E{\chi(X^t_0)\mid L\geq j\Delta}\right)\\
&= 
\Delta^{\alpha+1}\sum\limits_{j\geq 0}\left((j+1)^{\alpha}\cdot \Pr{L\geq j\Delta}\cdot \E{\chi(X^t_0)\mid L\geq j\Delta}\right),
\end{aligned}
\end{equation}
where the inequalities are by the non-negativity of $\chi(\cdot)$.
By \Cref{thm-expect-depth}, 
we have for every $i\geq 0$,
\begin{align}\label{eq-pr-llargeridelta-expectation}
\Pr{L\geq i \Delta }\cdot \E{\chi(X^t_0)\mid L\geq i\Delta}  
\leq 2^{-i}.
\end{align}
Combining with (\ref{eq-pathtimebound-4}), we have
\begin{align*}
\E{\chi(X^t_0)\cdot L^{\alpha}} \leq \Delta^{\alpha+1}\sum_{i\geq 0}\left((i+1)^{\alpha}\cdot 2^{-i}\right).
\end{align*}
Note that
$\sum_{i\geq 0}\left((i+1)\cdot 2^{-i}\right)\leq 4$, $\sum_{i\geq 0}\left((i+1)^{2}\cdot 2^{-i}\right)\leq 12$ and $\sum_{i\geq 0}\left((i+1)^{3}\cdot 2^{-i}\right)\leq 52$.
Thus, we have 
\begin{align}
\E{\chi(X^t_0)\cdot L^{\alpha}} \leq 
\begin{cases}
4\Delta^2\quad &\text{ if $\alpha = 1$,}\\
12\Delta^3\quad &\text{ if $\alpha = 2$,}\\
52\Delta^4\quad &\text{ if $\alpha = 3$.}\\
\end{cases}
\end{align}
Combining with~(\ref{eq-e-sum-i-ell-2-qtheta-2}), 
\eqref{sum-i12-qtheta} is immediate.

Next, we prove (\ref{eq-2-qtheta-ell-2}).
By \Cref{general-bfs-bound-component-invariant} and (\ref{eq-abs-con-xti}),  there exist some constants $C_2,C_3>0$ such that 
\begin{equation*}
\begin{aligned}
\tbfup{(X^t_{\ell})}
&\leq k^2q^2\Delta^6 \left(\abs{\ccon{X^t_{\ell}}}+1\right)(1-\mathrm{e}\pprime q)^{-\abs{\ccon{X^t_{\ell}}}}(C_2\cdot \Formal{x}+C_3kq )\\
&\leq k^2q^2\Delta^6\cdot(\Delta L+1) (1-\mathrm{e}\pprime q)^{-\Delta L}\cdot (C_2 \cdot \Formal{x}+C_3kq  ).
\end{aligned}
\end{equation*}
We have
\begin{equation}\label{eq-e-sum-i-ell-2-qtheta-3}
\begin{aligned}
\E{\chi(X^t_{\ell},X^t_0)\cdot \tbfup{(X^t_{\ell})}}
&\leq \E{\chi(X^t_0)k^2q^2\Delta^6\cdot(\Delta L+1) (1-\mathrm{e}\pprime q)^{-\Delta L}\cdot (C_2 \cdot \Formal{x}+C_3kq )}\\
&= k^2q^2\Delta^7\cdot\E{\chi(X^t_0)(L+1) (1-\mathrm{e}\pprime q)^{-\Delta L}}\cdot (C_2 \cdot \Formal{x}+C_3kq  ).
\end{aligned}
\end{equation}
By \eqref{eq:parameter-p-prime}
we have 
and $\mathrm{e}\pprime q\leq (4\Delta^2)^{-1}$.
Thus $(1-\mathrm{e}\pprime q)^{-\Delta^2}\leq 1.3$ for each $\Delta\geq 2$.
Let $j \defeq \lfloor\frac{i}{\Delta}\rfloor$.
We have
\begin{align*}
\E{\chi(X^t_0)(L+1) (1-\mathrm{e}\pprime q)^{-\Delta L}}
= &\sum_{i\geq 0}\left(\Pr{L = i}\cdot \E{\chi(X^t_0)\mid L=i}\cdot (i+1) \cdot (1-\mathrm{e}\pprime q)^{-i\Delta}\right)\\
\leq & 
\sum_{i\geq 0}\left((i+1)\cdot 1.3^{i/\Delta}\cdot\Pr{L = i}\cdot \E{\chi(X^t_0)\mid L= i}\right)\\
\leq & 
\sum_{i\geq 0}\left((i+1)\cdot 1.3^{i/\Delta}\cdot\Pr{L \geq i}\cdot \E{\chi(X^t_0)\mid L\geq i}\right)\\
\leq & 
2\Delta\sum_{i\geq 0}\left((j+1)\cdot 1.3^{j}\cdot \Pr{L \geq j\Delta}\cdot \E{\chi(X^t_0)\mid L\geq j\Delta}\right)\\
\leq & 2\Delta^2 \sum_{j\geq 0}\left((j+1)\cdot1.3^{j}\cdot\Pr{L \geq j\Delta}\cdot\E{\chi(X^t_0)\mid L\geq j\Delta}\right),
\end{align*}
where the equality is by the law of total expectation and the inequalities are by the non-negativity of $\chi(\cdot)$.
Combining with \eqref{eq-pr-llargeridelta-expectation}, we have
\begin{align*}
\E{\chi(X^t_0)(L+1) (1-\mathrm{e}\pprime q)^{-\Delta L}} \leq 2\Delta^2\sum_{i\geq 0}\left((i+1)\cdot 0.65^{i}\right)\leq 20\Delta^2.
\end{align*}
Combining with \eqref{eq-e-sum-i-ell-2-qtheta-3},
we have 
\begin{equation*}
\begin{aligned}
\E{\chi(X^t_{\ell},X^t_0)\cdot \tbfup{(X^t_{\ell})}}
&\leq k^2q^2\Delta^7\cdot 20\Delta^2 \cdot (C_2\cdot \Formal{x}+C_3kq )\\
&\leq 20C_2 k^2q^2\Delta^{9}\cdot\Formal{x} + 20C_3k^3q^3\Delta^{9}.
\end{aligned}
\end{equation*}
Then (\ref{eq-2-qtheta-ell-2}) holds, which finishes the proof of the lemma.
\end{proof}

Now we can prove \Cref{thm-time-MSamp}, the main theorem of \Cref{sec:efficiency-marginsample}.
\begin{proof}[Proof of \Cref{thm-time-MSamp}]
Let 
\begin{align*}
U &\defeq \left\{\sigma\in \qs \mid \Pr{X^{t-1}=\sigma}>0 \right\},\\
S &\defeq \left\{\sigma\in \qs \mid \Pr{X^{t-1}=\sigma}>0 \land  v_t\not\in \vfix{\sigma}\right\}.    
\end{align*}
For any $\sigma\in S$, we have $\Pr{X^{t-1}=\sigma}>0$ and $v_t\not\in \vfix{\sigma}$, 
therefore by the construction of $X^0,X^1,\ldots,X^n$ in \Cref{def-pas-main},
there is a positive probability that $\newsample{}(\Phi, \sigma, v_t)$ is called within \Cref{Alg:main}.
Hence by \Cref{lemma:invariant-marginsample}, \Cref{inputcondition-magin} is satisfied by $(\Phi, \sigma, v_t)$ for every $\sigma \in S$.

For any $\sigma\in \qs$, let $\gamma(\sigma)=\mathbb{E}\left[H(\pth(\sigma))\right](0,0)$ denote the constant term in $\mathbb{E}\left[H(\pth(\sigma))\right]$. 
By \Cref{cor-t-msamp}, there exist constants $C_1,C_2>0$ such that for every $\sigma\in S$, 
\begin{align}\label{eq-tsamp-phi-r}
\tsamp{}(\Phi,\sigma,v_t)\le (1-q_{v_t}\cdot \theta_{v_t})(\E{H(\pth(\sigma_{v_t\gets \star}))}+C_1\cdot \gamma(\sigma_{v_t\gets \star}))+C_2.
\end{align}
And for every $\sigma\in U\setminus S$, $v_t$ is $\sigma$-fixed, and hence $\tsamp{}(\Phi,\sigma,v_t)=\E{\rtsamp{}(\Phi,\sigma,v_t)}=0$ by convention. On the other hand, $\E{H(\pth(\sigma))}$ is always nonnegative. Thus the following  holds trivially:
\begin{align}\label{eq-tsamp-phi-r-trivial}
\tsamp{}(\Phi,\sigma,v_t)\le \E{H(\pth(\sigma))}+C_1\cdot \gamma(\sigma)+C_2.
\end{align}
In addition, given $\sigma\in U$ and $X^{t-1} = \sigma$, by \eqref{eq:simulate-path} we have if $v_t$ is $\sigma$-fixed, then $X^t_0 = \sigma$;
otherwise, 
\begin{align}\label{eq-xt0-case-vt-not-fix}
    X^t_0 = \begin{cases}\sigma_{v_t\gets \star}& \text{with probability $1-q_{v_t}\cdot \theta_{v_t}$}\\ \sigma& \text{with probability $q_{v_t}\cdot \theta_{v_t}$}  \end{cases}.
\end{align}
Thus, if $v_t$ is $\sigma$-fixed, by \eqref{eq-tsamp-phi-r-trivial} and $X^t_0= X^{t-1} = \sigma$,
we have 
\begin{align}\label{eq-tsamp-for-two-cases}
\tsamp{}(\Phi,X^{t-1},v_t)\le \E{H(\pth(X_0^t))}+C_1\cdot \gamma(X_0^t)+C_2.
\end{align}
If $v_t$ is not $\sigma$-fixed, 
by \eqref{eq-xt0-case-vt-not-fix} and the law of total expectation,
we have
$$(1-q_{v_t}\cdot \theta_{v_t})(\E{H(\pth(\sigma_{v_t\gets \star}))}+C_1\cdot \gamma(\sigma_{v_t\gets \star}))+C_2 \leq \E{H(\pth(X_0^t))}+C_1\cdot \gamma(X_0^t)+C_2.$$
Combining with $\eqref{eq-tsamp-phi-r}$ and $X^{t-1} = \sigma$, 
we also have \eqref{eq-tsamp-for-two-cases}.
In summary, conditioning on any possible $X^{t-1}=\sigma\in U$, \eqref{eq-tsamp-for-two-cases} always holds.
Hence by the law of total expectation, we have 
\begin{align*}
\E{\rtsamp{(\Phi,X^{t-1},v_t)}}=\E{\tsamp{}(\Phi,X^{t-1},v_t)} \leq \E{H(\pth(X^t_0))}+C_1\cdot \E{\gamma(X^t_0)}+C_2.
\end{align*}
By \Cref{pathtimebound}, for such $X_0^t$ constructed as above,
there exist constants $C_3,C_4>0$ such that
\begin{align*}
  \mathbb{E}\left[H(\pth(X_0^t)))\right]\leq C_3q^2k^2\Delta^{9}\cdot \Formal{x}+24k\Delta^6 \cdot \Formal{y} + C_4q^3k^3\Delta^{9},
\end{align*}
which means that
\[
\E{\rtsamp{(\Phi,X^{t-1},v_t)}}\leq C_3q^2k^2\Delta^{9}\cdot \Formal{x}+24k\Delta^6 \cdot \Formal{y} + (C_1+1)C_4\cdot q^3k^3\Delta^{9}+C_2. \qedhere
\]
\end{proof}

\subsection{Efficiency of \rejsamp{}}
Recall the random sequence of partial assignments:
\[
X^0,X^1,\dots,X^n,
\]
maintained in \Cref{Alg:main}, as formally defined in \Cref{def-pas-main}.

Here, we focus on $X=X^{n}$, the partial assignment obtained after all $n$ iterations of the \textbf{for} loop in \Cref{line-main-for} of \Cref{Alg:main}, and passed to the \rejsamp{} (\Cref{Alg:rej}) as input.
%


Recall the following definitions in \Cref{sec:rejection-sampling}.
For each $\sigma\in\qs$ and $v\in V^{\sigma}$, 
$H_v^{\sigma}=(V_v^\sigma,\+{C}_v^{\sigma})$ denotes the connected component in $H^{\sigma}$ that contains the vertex/variable $v$, 
where $H^{\sigma}$ is the hypergraph representation for the CSP formula $\Phi^\sigma$ obtained from the simplification of $\Phi$ under $\sigma$.
We further stipulate that 
$H_v^{\sigma}=(V_v^\sigma,\+{C}_v^{\sigma})=(\emptyset,\emptyset)$ is the empty hypergraph when $v\in \Lambda(\sigma)$ is assigned in $\sigma$.

For partial assignment $\sigma\in\qs$, 
we let $\trej{(\sigma)}$ be the random variable that represents the complexity of $\rejsamp{}(\Phi,\sigma,V\setminus \Lambda(\sigma))$, expressed in form of \eqref{eq:abstract-complexity-bound}.
%
We prove the following bound on the expectation of $\trej{(X)}$ on the random partial assignment $X=X^n$.

\begin{theorem}\label{rejsampef}
Assume $8\mathrm{e}p\Delta^3\leq 0.99\pprime$,where $\pprime$ is fixed as in \eqref{eq:parameter-p-prime}. 
\[
\E{\trej{(X)}}\leq 128\mathrm{e}k \Delta^4 n\cdot \Formal{x}+ O(qk\Delta^4 n),
\]
where expectation is taken over both $X$ and the randomness of \rejsamp{} algorithm.
\end{theorem}

\begin{proof}[Proof of \Cref{rejsampef}]
Let $\{H_i^X= (V_i^X,\+{C}_i^X)\}\mid 1\leq i\leq K\}$ be the connected components in $H^X$ constructed in \Cref{line-rs-find} \Cref{Alg:rej}.
For each $1\le i\le K$, 
let $T\left(H^X_i\right)$ denote the cost contributed by $H_i^X= (V_i^X,\+{C}_i^X)$ during the \text{repeat} loop  in \Cref{Alg:rej}.
Then the total cost is given by
\begin{align}\label{eq-trej-phix}
   \trej{(X)}  \leq  \Delta n\cdot \Formal{x} + O(\Delta n)+ \sum_{i\in [K]}T\left(H^X_i\right).
\end{align}
This is because:
\begin{itemize}
    \item 
    the cost of \Cref{line-rs-find} is at most $n\Delta\cdot \Formal{x}+O(n\Delta)$, 
    because it uses at most $\abs{\+{C}}\le \Delta n$ queries to $\eval{}(\cdot)$, 
    which contributes the $\Delta n\cdot \Formal{x}$ term, 
    and a depth-first search that visits all variables and constraints to compute $H_1^X,H_2^X,\dots,H_K^X$, which costs $O(\Delta n)$ in computation.
    \item 
    the cost of Lines \ref{line-rs-for}-\ref{line-rs-until} for each component $H_i^X$ is $T\left(H^X_i\right)$.
\end{itemize}
Alternatively, for each $v\in V$, we use $T\left(H^X_v\right)$ to denote the $T\left(H^X_i\right)$ for the $1\le i\le K$ with $v\in V_i^X$; and let $T\left(H^X_v\right)=0$ if there is no such $1\le i\le K$, which occurs when $v\in\Lambda(X)$ is assigned in $X$.
%
%
Clearly, 
$$\sum_{i\in [K]}T\left(H^X_i\right) \leq \sum_{v\in V}T\left(H^X_v\right).$$
Combining with (\ref{eq-trej-phix}), we have
\begin{equation}\label{eq-rej-timebound}
    \E{\trej{(X)}} \leq n\Delta\cdot \Formal{x} + O(n\Delta) + \sum_{v\in V}\E{T\left(H^X_v\right)}.
\end{equation}


We claim that for each $v\in V$,
\begin{equation}
\begin{aligned}\label{eq-et-phiyi}
\E{T(H^X_v)\mid X} \leq \abs{ \+C^X_v}\cdot (1-\mathrm{e}\pprime q)^{-\abs{ \+C^X_v}} \cdot \Formal{x} + O\left(kq\cdot \left(\abs{ \+C^X_v}+1\right)\cdot (1-\mathrm{e}\pprime q)^{-\abs{ \+C^X_v}}\right).
\end{aligned}
\end{equation}
The degenerate case with $H_v^{X}=(\emptyset,\emptyset)$ is trivial. 
We then assume $H_v^{X}\neq(\emptyset,\emptyset)$.
It is well known that the expected number of trials 
(the iterations of the \textbf{repeat} loop in \Cref{line-rs-until}) 
taken by the rejection sampling until success 
is given by $\mathbb{P}_{H^{X}_v}[\Omega_{H^{X}_v}]^{-1}$, 
where $\Omega_{H^{X}_v}$ denotes the set of satisfying assignments of the CSP of $H^{X}_v$ and hence  $\mathbb{P}_{H^{X}_v}[\Omega_{H^{X}_v}]$ gives the probability that a uniform random assignment $\sigma\in\+{Q}_{V^X_v}$ is satisfying for $\Phi^{X}_v$, the CSP formula correspond to the component $H^{X}_v$.

By \Cref{locallemma} and \Cref{lemma:invariant-p-prime-q-bound}, we have
\[
\mathbb{P}_{H^X_v}\left[\Omega_{H^X_v}\right]
\ge 
(1-\mathrm{e}\pprime q)^{\abs{\+C^{X}_v}}.
\]
Therefore, it takes 
$(1-\mathrm{e}\pprime q)^{-\abs{\+C^{X}_v}}$ iterations in expectation to successfully sample the assignment on $V_v^X$.
And within each iteration, 
it is easy to verify that at most $\abs{\+{C}_v^X}$ queries to $\eval{}(\cdot)$ and
$O\left(k\abs{ \+{C}_v^X}+q\abs{ V_v^X}\right)=O\left(qk(\abs{ \+{C}_v^X}+1)\right)$ computation cost are spent.
This proves the claim \eqref{eq-et-phiyi}.

We then bound the expectation of \eqref{eq-et-phiyi} over $X$. Let $\ell \defeq \lfloor t/(2\Delta^2)\rfloor$.
By \eqref{eq:parameter-p-prime},
we have  $\mathrm{e}\pprime q\leq (2\Delta^2)^{-1}$
and then
$$(1-\mathrm{e}\pprime q)^{-(\ell+1)\Delta^2}\leq 2^{\ell+1}.$$
We have
\begin{align*}
\sum_{t\geq 0 } \left((t+1)(1-\mathrm{e}\pprime q)^{-t}\Pr{\abs{\+C^X_v}=t}\right) &\leq 2\Delta^2\sum_{t\geq 0}\left(\Pr{\abs{\+C^X_v} \geq 2\ell\Delta^2}\cdot (\ell+1)\cdot(1-\mathrm{e}\pprime q)^{-(\ell+1)\Delta^2}\right)\\
&\leq 4\Delta^2 \sum_{t\geq 0}\left((\ell+1)\cdot 2^{\ell}\cdot \Pr{\abs{\+C^X_v} \geq 2\ell\Delta^2} \right)\\
&=4\Delta^4 \sum_{\ell\geq 0}\left((\ell+1)\cdot 2^{\ell}\cdot\Pr{\abs{\+C^X_v} \geq 2\ell \Delta^2}\right)\\
(\text{by \Cref{thm-expect-depth2}})\quad
&\le
32\mathrm{e}k\Delta^4 \sum_{\ell\geq 0} \left((\ell+1)\cdot2^{-\ell}\right)\\
&\le 128\mathrm{e}k\Delta^4.
\end{align*}
Therefore,
\begin{align*}
\E{\left(\abs{ \+C^X_v}+1\right)\cdot (1-\mathrm{e}\pprime q)^{-\abs{ \+C^X_v}}}
=
\sum_{t\geq 0 } \left((t+1)(1-\mathrm{e}\pprime q)^{-t}\Pr{\abs{\+C^X_v}=t}\right) &\leq 128\mathrm{e}k\Delta^4.
\end{align*}
Thus, there exists a constant $C>0$  such that the expectation of \eqref{eq-et-phiyi} is bounded by
\begin{align*}
\E{T(H^X_v)}
\leq 128\mathrm{e}k\Delta^4\cdot \Formal{x}+128\mathrm{e}Ck^2q\Delta^4.
\end{align*}
The theorem follows by \eqref{eq-rej-timebound}.
\end{proof}

\subsection{Efficiency of the main sampling algorithm}\label{sec:efficiency-main-alg}
We now prove \Cref{mainef} and \Cref{mainef-approx}. 

\begin{proof}[Proof of \Cref{mainef}]
Assuming the LLL condition in \Cref{eq:main-thm-LLL-condition}, 
we have $8\mathrm{e}p\Delta^3\leq 0.99\pprime$ for $\pprime$ set as in \eqref{eq:parameter-p-prime},
which is the regime of parameters assumed by \Cref{thm-time-MSamp} and \Cref{rejsampef}.

The followings are the nontrivial costs in \Cref{Alg:main}:
\begin{itemize}
    \item 
    The initialization in \Cref{line-main-init} and the testing of being $X^{t-1}$-fixed for $v_t$ in \Cref{line-main-if} for every $1\le t\le n$, which altogether cost:
    \[
    \Delta n\cdot\Formal{y} + O(\Delta n),
    \]
    because each $v_t$ queries $\checkf{}(\cdot)$ for at most $\abs{c\in\+{C}\mid v_t\in\vbl(c)}\le\Delta$ times and also costs $O(\Delta)$ in computation to retrieve all such $c\in\+{C}$ that $v_t\in\vbl(c)$.
    \item 
    The calls to $\newsample{}(\Phi,X^{t-1},v_t)$ at \Cref{line-main-sample}, which according to \Cref{thm-time-MSamp}, costs in total:
    \[
    \sum_{t=1}^n\E{\rtsamp{(\Phi,X^{t-1},v_t)}}\leq 
    O\left(q^2k^2\Delta^{9} n\right)\cdot \Formal{x}+24k\Delta^6 n \cdot \Formal{y}+O\left(q^3k^3\Delta^{9} n\right).
    \]
    \item 
    The final call to $\rejsamp(\Phi,X^n,V\setminus\Lambda(X^n))$ at \Cref{line-main-partialass}, which by \Cref{rejsampef}, costs:
    \[
    \E{\trej{(X^n)}}\leq 128\mathrm{e}k\Delta^4 n\cdot \Formal{x}+ O\left( q k\Delta^4 n\right).
    \]
\end{itemize}
The overall complexity of \Cref{Alg:main} in expectation is bounded by 
\begin{align}\label{eq:complexity-bound-main-alg}
    O\left(q^2k^2\Delta^{9} n\right)\cdot \Formal{x}+24k\Delta^6 n \cdot \Formal{y}+O\left(q^3k^3\Delta^{9} n\right).  
\end{align}
The theorem is proved.
\end{proof}

Next, we construct the \Cref{Alg:main}' for approximate sampling and prove \Cref{mainef-approx}. 

Given as input a CSP formula $\Phi=(V,\+{Q},\+{C})$ with $n=|V|$ variables 
and an error bound $\varepsilon\in(0,1)$, 
\Cref{Alg:main}' does the followings. 
Set the parameters as:
\begin{align}\label{eq:monte-carlo-frozen-parameters}
    \delta=0.005 \quad\text{ and }\quad
    N=\left\lceil\frac{\ln\left({200k\Delta^6n}{\varepsilon^{-2}}\right)}{0.33\pprime\delta^2}\right\rceil.
\end{align}
\Cref{Alg:main}' simply executes \Cref{Alg:main} on input $\Phi$, 
with the oracle $\checkf(\cdot)$ replaced by the following explicitly implemented Monte Carlo subroutine:
\begin{itemize}
    \item 
    Given as input any constraint $c\in\+{C}$ and any partial assignment $\sigma\in\qs$, 
    repeat for $N$ times:
    \begin{itemize}
        \item     
        generate an assignment $Y\in\+{Q}_{\vbl(c)}$ on $\vbl(c)$ uniformly at random consistent with $\sigma$; 
        \item
        check whether $c(Y)=\True$ by querying $\eval(c,Y)$;
    \end{itemize}
    \item 
    let $Z$ be the number of times within $N$ trials that $c(Y)=\False$, and return $I[Z/N>0.995\pprime]$.
\end{itemize}
We further apply a standard memoization trick to guarantee the consistency of the oracle $\checkf(\cdot)$ as required in \Cref{assumption:frozen-oracle}.
Each  constraint $c\in\+{C}$ is associated with a deterministic dynamic dictionary $\mathrm{Dic}_c$ 
which stores key-value pairs in the form of $(\tau,ans)$ with $\tau\in\bigotimes_{v\in\vbl(c)}(Q_v\cup\{\star,\hollowstar\})$ and $ans\in\{0,1\}$.
Upon each query to $\checkf(c,\sigma)$, 
we first lookup in $\mathrm{Dic}_c$ for the key $\sigma_{\vbl(c)}$.
If an $ans\in\{0,1\}$ is retrieved, the query to $\checkf(c,\sigma)$ is answered with $ans$;
and if otherwise, 
an $ans=I[Z/N>0.995\pprime]\in\{0,1\}$ is computed using the above Monte Carlo subroutine, 
the key-value pair $(\sigma_{\vbl(c)},ans)$ is inserted into $\mathrm{Dic}_c$ 
and the query to $\checkf(c,\sigma)$ is answered with $ans$.
Using a \emph{Trie} data structure for the deterministic dynamic dictionary $\mathrm{Dic}_c$ incurs an $O(k)$ computation cost for each query and update.
Overall, the resulting algorithm is \Cref{Alg:main}'.

\begin{proof}[Proof of \Cref{mainef-approx}]
In \Cref{Alg:main}', each query to the oracle $\checkf(\cdot)$ made in \Cref{Alg:main} is implemented using $N$ queries to the evaluation oracle $\eval(\cdot)$ along with $O(kN)$ computation cost, where $N$ is set as in \eqref{eq:monte-carlo-frozen-parameters}.
By \Cref{mainef}, the complexity of \Cref{Alg:main} in expectation is bounded as \eqref{eq:complexity-bound-main-alg}. By replacing 
\[
\Formal{y}\gets N\cdot \Formal{x} + O(k N),
\]
the expected complexity of \Cref{Alg:main}' is bounded by
\[
O\left(q^2k^2\Delta^{9}n\log{\left(\frac{\Delta n}{\varepsilon}\right)}\cdot \Formal{x}+q^3k^3\Delta^{9}n\log{\left(\frac{\Delta n}{\varepsilon}\right)}\right).
\]

By Chernoff bound, one can verify that for each query to $\checkf(c,\sigma)$, 
the two cases $\mathbb{P}[\neg c\mid\sigma]>\pprime$ and $\mathbb{P}[\neg c\mid\sigma]\leq 0.99\pprime$ are distinguished correctly except for an error probability  bounded by:
\[
2\exp\left(-\frac{\delta^2}{3}0.99\pprime N\right)\le\frac{\varepsilon}{2M}\quad\text{ where }M\triangleq 50k\Delta^6n\varepsilon^{-1}.
\]
Due to \eqref{eq:complexity-bound-main-alg}, the expected number of queries to $\checkf(\cdot)$ made in \Cref{Alg:main} is at most $24k\Delta^6n\leq \frac{\varepsilon M}{2}$.

By Markov's inequality, the probability that the number of queries  to $\checkf(\cdot)$ made in \Cref{Alg:main} exceeds  $M$ is  at most $\frac{\varepsilon}{2}$.
By a union bound, with probability at least $1-\varepsilon$, no error occurs in any query to $\checkf(\cdot)$.
Then by a coupling between \Cref{Alg:main} and \Cref{Alg:main}', the total variation distance between the outputs of the two algorithms on the same input CSP $\Phi$ can be bounded within $\varepsilon$. 
\end{proof}

The following is a formal restatement of \Cref{thm:main-margin-sampling} and \Cref{thm:main-inference}.
Here the $\tilde{O}(\cdot)$ hides polylogarithmic factors. 
The algorithm is just the Monte Carlo method that uses $\newsample{}(\Phi,\hollowstar^V,v)$ (\Cref{Alg:eq}) as a subroutine for sampling from the marginal distribution $\mu_v$.

\begin{theorem}\label{thm:margin-sample-inference}
There are algorithms for marginal sampling and probabilistic inference
such that given as input $\varepsilon,\delta\in(0,1)$, 
CSP formula $\Phi=(V,\+{Q},\+{C})$ satisfying \eqref{eq:main-thm-LLL-condition}, 
and $v\in V$, the algorithms perform as follows:
\begin{itemize}
    \item 
The algorithm for marginal sampling returns 
a random value $x\in Q_v$ distributed approximately as $\mu_v$ within total variation distance $\varepsilon$, in expectation using at most $O\left(q^2k^2\Delta^{9}\log(k\Delta/\varepsilon)\right)$ queries to $\eval{}(\cdot)$ and  $O\left(q^3k^3\Delta^{9}\log(k\Delta/\varepsilon)\right)$ computation cost.
    \item
The algorithm for inference returns a $\hat{\mu}_v\in[0,1]^{Q_v}$ using at most $\tilde{O}\left(q^3k^2\Delta^{9}\varepsilon^{-2}\log(1/\delta)\right)$ queries to $\eval{}(\cdot)$ 
and  $\tilde{O}\left(q^4k^3\Delta^{9}\varepsilon^{-2}\log(1/\delta)\right)$ computation cost
such that 
\begin{align*}
\Pr{\forall x\in Q_v:\, (1-\varepsilon)\mu_v(x)\le\hat{\mu}_v(x)\le (1+\varepsilon)\mu_v(x)}\ge 1-\delta.
\end{align*}
\end{itemize}
\end{theorem}

\begin{proof}
When \eqref{eq:main-thm-LLL-condition} is satisfied,  
we have $8\mathrm{e}p\Delta^3\leq 0.99\pprime$ for the $\pprime$ set as in \eqref{eq:parameter-p-prime}.
Then \Cref{inputcondition-magin} is satisfied by $(\Phi,\hollowstar^V,v)$.
According to \Cref{sampcor} and \Cref{thm-time-MSamp}, $\newsample{}(\Phi,\hollowstar^V,v)$ returns a perfect sample from $\mu_v$ with expected cost:
\[
 O\left(q^2k^2\Delta^{9}\right)\cdot \Formal{x}+24k\Delta^6 \cdot \Formal{y}+O\left(q^3k^3\Delta^{9}\right).
\]
Using the same Monte Carlo simulation of $\checkf{}(\cdot)$ as in the proof of \Cref{mainef-approx}, 
but this time with sufficiently large parameter  $N=O\left(\frac{1}{\pprime}\log(k\Delta/\varepsilon)\right)=O\left(q^2k\Delta^2\log(k\Delta/\varepsilon)\right)$, with the substitution $\Formal{y}\gets N\cdot \Formal{x} + O(k N)$, the $\newsample{}(\Phi,\hollowstar^V,v)$ is transformed into an approximate sampler for $\mu_v$ with bias at most $\varepsilon/2$ in total variation distance, with the following cost:
\[
 O\left(q^2k^2\Delta^{9}\log(k\Delta/\varepsilon)\right)\cdot \Formal{x} + O\left(q^3k^3\Delta^{9}\log(k\Delta/\varepsilon)\right).
\]
This proves the marginal sampler part of the algorithm.

By the local uniformity property stated in \Cref{generaluniformity}, 
assuming \eqref{eq:main-thm-LLL-condition},
we have $\mu_v(x)>\frac{1}{2q}$.
Thus, by Chernoff bound, each $\mu_v(x)$ for $x\in Q_v$ can be estimated within $(1\pm\varepsilon)$-multiplicative precision with probability $>0.9$ using $O(q/\varepsilon^{2})$ approximate samples, which in expectation costs in total:
\begin{align}\label{eq:inference-complexity-bound}
 \tilde{O}\left(q^3k^2\Delta^{9}\varepsilon^{-2}\right)\cdot \Formal{x} + \tilde{O}\left(q^4k^3\Delta^{9}\varepsilon^{-2}\right).
\end{align}
By Markov's inequality, the total cost exceeds this bound with probability $<0.1$. 
By truncation, this gives us an algorithm with fixed cost bounded as \eqref{eq:inference-complexity-bound} so for each $x\in Q_v$, with probability $>0.8$ the algorithm estimates the value of $\mu_v(x)$ within $(1\pm\varepsilon)$-multiplicative precision.
The algorithm asserted by the theorem is then given by repeating this for $O(\log\frac{q}{\delta})$ times and applying the median trick.
\end{proof}

\section{The Generalized \texorpdfstring{$\{2,3\}$}{{2,3}}-Tree }\label{sec:g23-tree}

In this section, we prove the tail bounds stated in Lemmas 
\ref{thm-expect-depth} and \ref{thm-expect-depth2} in \Cref{section-samp-ef}.
These two tail bounds basically state that the large witness of long running time occurs with an exponentially small probability,
and are both proved with a new 
combinatorial structure called ``generalized $\{2,3\}$-tree''.
The generalized $\{2,3\}$-tree is a refinement of $\{2,3\}$-tree, which is previously used in  ~\cite{guo2019counting,Vishesh21towards}. 
Rather than treated similarly in $\{2,3\}$-tree,
the ``bad" variables and constraints are treated separately in generalized $\{2,3\}$-tree, because 
the densities of these two types of bad events are different.

Given a hypergraph $H=(V,\+E)$, let $\Lin{H}$ be the line graph of $H$ whose vertex set is the hyperedges in $\+E$ and two hyperedges in $\+E$ are adjacent in $\Lin{H}$ if and only if they share some vertex in $H$. 
Let $\text{dist}_{\Lin{H}}(\cdot,\cdot)$ be the shortest path distance in $\Lin{H}$. 

\begin{definition}{(generalized $\{2,3\}$-tree)}\label{WTdef}
Given a hypergraph $H=(V,\+E)$, a set $T = U\cup E$ where $U\subseteq V$ and $E\subseteq  \+E$, 
let $G(T)$ be a directed graph constructed on the set $T$ such that, 
for any $u, v\in T$ there is an arc from $u$ to $v$ if and only if at least one of the following conditions is satisfied: 
    \begin{itemize}
        \item $u,v\in E$ and $\text{dist}_{\Lin{H}}(u,v)=2\text{ or }3$;
        \item $u\in U, v\in E$ and  there exists $e\in \+E$ such that $u\in e\land \text{dist}_{\Lin{H}}(v,e)=1$;
        \item $u\in E,v\in U$ and there exists $e\in \+{E}$ such that $v\in e\land \text{dist}_{\Lin{H}}(u,e)=1\text{ or }2$;
        \item $u,v\in U$ and there exists $e\in \+{E}$ such that $u,v\in e$.
    \end{itemize}
Then $T$ is a \emph{generalized $\{2,3\}$-tree} of $H$ if the followings hold:
\begin{enumerate}
    \item for all distinct $u, v\in E$, $\text{dist}_{\Lin{H}}(u,v)\geq 2$;  \label{WT-1}
    \item $G(T)$ has a rooted directed spanning tree.
\end{enumerate}
Furthermore, each $v\in T$ is called a \emph{root} of $T$ if it is a root of some rooted directed spanning tree of $G(T)$.
For any $v\in V$ and intergers $r,\ell,t\geq 0$, define 
$$\MSC{T}_{v}^{r,\ell}(H) \defeq \left\{\text{generalized $\{2,3\}$-tree $T$ of $H$} \mid (v\text{ is a root of $T$})\land (\abs{T\cap V} =r)\land (\abs{T\cap \+E}) = \ell) \right\},$$
$$\MSC{T}_{v}^{t}(H) \defeq \left\{\text{generalized $\{2,3\}$-tree $T$ of $H$} \mid (v\text{ is a root of $T$})\land (\abs{T\cap V} =r + \Delta\cdot \abs{T\cap \+E} = t) \right\}.$$
We will use $\MSC{T}_{v}^t$ and $\MSC{T}_{v}^{r,\ell}$ to  stand for $\MSC{T}_{v}^t(H)$ and $\MSC{T}_{v}^{r,\ell}(H)$ respectively if $H$ is clear from the context.
\end{definition}

The generalized $\{2,3\}$-tree in \Cref{WTdef} is inspired by the the notion of $\{2,3\}$-tree defined for the line graph $\Lin{H}$~\cite{alon1991parallel}.
We extend this notion to the original hypergraph $H$ to simultaneously depict the distances between vertices and hyperedges in $H$. 
We further allow including vertices in $H$ to some generalized $\{2,3\}$-tree $T$ granted that all vertices and hyperedges included are close enough to each other.
One can verify that every $\{2,3\}$-tree in $\Lin{H}$ is also a generalized $\{2,3\}$-tree in $H$. 

Specifically, when the underlying hypergraph in \Cref{WTdef} is the hypergraph representation $H_{\Phi}=(V,\+C)$ of some CSP $\Phi$, a generalized $\{2,3\}$-tree $T\subseteq V\cup \+C$ in $H_{\Phi}$ becomes a subset of variables and constraints.

\subsection{Boosting of tail bounds using generalized $\{2,3\}$-trees}\label{sec:boost}

In this section we prove Lemmas 
\ref{thm-expect-depth} and \ref{thm-expect-depth2}. As stated in \Cref{sec:witness-tree}, 
the tail bounds in these two lemmas are proved by boosting two ``basic'' tail bounds over the occurrences of ``bad" variables and \emph{disjoint} ``bad" constraints using generalized $\{2,3\}$-tree. 

Recall the bad event $\+{E}^{t}_T$ in \Cref{def:cbad}. 
The following lemma is used in the proof of \Cref{thm-expect-depth}, which provides a tail bound for the bad event $\+{E}^{t}_T$ when the constraints in $T$ are disjoint. 
It states that within an LLL regime, the bad event $\+{E}^{t}_T$ that may lead to the inefficiency of \newsample{} rarely occurs.

\sloppy
\begin{lemma}\label{badtree}
Assume $8\mathrm{e}p\Delta^3\leq 0.99\pprime$, where $\pprime$ is defined as in \eqref{eq:parameter-p-prime}.
Let $1\le t\le n$ and $(X^0,X^1,\ldots,X^{t-1},X^t_0,X^t_1,\ldots,X^t_{\ell})$ be generated as in \eqref{eq:simulate-path}.
For any subset of variables and constraints $T=U\uplus E$ such that the constraints in $E$ are disjoint, we have
\begin{align}\label{eq-lemma-badtree}
\Pr{ \+{E}^{t}_T}\cdot\E{\chi(X^t_{0})\mid \+{E}^{t}_T}\leq \left(8\mathrm{e}k\Delta\right)^{-\abs{ U }}\cdot \left(4\mathrm{e}\Delta^3\right)^{-\abs{ E }}.
\end{align}
\end{lemma}

The following lemma is used in the proof of  \Cref{thm-expect-depth2}. 
It provides another tail bound for a bad event over a set of disjoint constraints that may lead to the inefficiency of \rejsamp{}.

\begin{lemma}\label{badtree2}
Assume $8\mathrm{e}p\Delta^3\leq 0.99\pprime$. Let $(X^0,X^1,\ldots, X^n)=\simulator(n)$.
For any set of disjoint constraints $T\subseteq \+{C}$,
 $$\Pr{T\subseteq \cfrozen{X^n}} \leq \left(4\mathrm{e}\Delta^3\right)^{-\abs{ T }}.$$
\end{lemma}
\noindent
These are the two basic tail bounds. They will be proved later in \Cref{sec:raw-tail-bound}. 
For the rest of \Cref{sec:boost}, we prove \Cref{thm-expect-depth} and \Cref{thm-expect-depth2} assuming \Cref{badtree} and \Cref{badtree2}.

We need the following lemma to bound the sum of weights of the generalized $\{2,3\}$-trees in $\MSC{T}_{v}^t$.

\begin{lemma}\label{lem-prob-g23tree}
Let $H = (V, \+E)$ be a hypergraph such that each hyperedge contains at most $k$ vertices and shares vertices with at most $\Delta$ hyperedges.
For any $v\in V$ and $t> 0$, we have
\begin{equation*}
    \sum\limits_{\substack{T\in \MSC{T}_{v}^t}}\left((8\mathrm{e}k\Delta)^{-\abs{T\cap V}}(4\mathrm{e}\Delta^3)^{-\abs{T\cap \+E}}\right)\leq 2^{-\left\lfloor t/\Delta\right\rfloor-1}\cdot \Delta^{-1}.
\end{equation*}
\end{lemma}

From now on, we denote $w_1\triangleq(8\mathrm{e}k\Delta)^{-1}$ and $w_2\triangleq(4\mathrm{e}\Delta^{3})^{-1}$.  
The following definition is used in the proof of \Cref{lem-prob-g23tree}, which is inspired by the fact that each hyperedge of $H$ contains at most $k$ vertices and shares vertices with at most $\Delta$ hyperedges.

\begin{definition}[$\{a,b\}$-labelled tree]\label{GWprocess2}
An \emph{$\{a,b\}$-labelled tree} is a rooted directed tree where
\begin{itemize}
    \item each node is labelled $a$ or $b$;
    \item each node labelled $a$ has weight $w_1x$ and each node labelled $b$ has weight $w_2x^{\Delta}$.
\end{itemize}
For each 
$\{a,b\}$-labelled tree $T$, 
let the weight of $T$ be the product of the weights of all the nodes in $T$.
With a bit abuse of notation, let $a(T)$ and $b(T)$ be the numbers of nodes with label $a$ and $b$ in $T$, respectively. 
Given $z\in \{a,b\}$, define a tree $\+T_{z}$ with infinite nodes
as follows:
\begin{itemize}
    \item $\+T_{z}$ is a \emph{$\{a,b\}$-labelled tree} where the root is labelled $z$;
    \item each node labelled $a$ has $k\Delta$ children labelled $a$ and $\Delta^2$ children labelled $b$;
    \item each node labelled $b$ has $k\Delta^2$ children labelled $a$ and $\Delta^3$ children labelled $b$.
\end{itemize}
A \emph{proper subtree} $T$ of $\+T_{z}$ is a rooted directed subtree of $\+T_{z}$ where
the root of $\+T_{z}$ is also the root of
$T$.
Obviously, $T$ is also an $\{a,b\}$-labelled tree.
Define $\+T^{t}_z\triangleq \{\text{proper subtree $T$ of $\+T_{z}$}\mid a(T)+\Delta \cdot b(T) = t\}$. 

\end{definition}

Now we can prove \Cref{lem-prob-g23tree}.

\begin{proof}[Proof of \Cref{lem-prob-g23tree}]
For each $v\in V$,
we claim that there exists an injection from the generalized $\{2,3\}$-trees in $\MSC{T}_{v}^t$ to the $\{a,b\}$-labelled trees in $\+T_a^t$ such that each $T\in \MSC{T}_{v}^t$ is mapped to some $T'\in \+T_a^t$ satisfying $\abs{T\cap V}=a(T')$ and $\abs{T\cap \+E}=b(T')$.  
Recall that $w_1=(8\mathrm{e}k\Delta)^{-1}$ and $w_2=(4\mathrm{e}\Delta^{3})^{-1}$.
We have
\begin{equation*}
    \sum\limits_{\substack{T\in \MSC{T}_{v}^t}}\left((8\mathrm{e}k\Delta)^{-\abs{T\cap V}}(4\mathrm{e}\Delta^3)^{-\abs{T\cap \+E}}\right)\leq \sum_{T\in \+T^t_{a}}w_1^{a(T)}w_2^{b(T)}.
\end{equation*}
Thus, to prove the lemma, it is sufficient to prove 
\begin{equation}\label{eq-sumo-tta}
    \sum_{T\in \+T^t_{a}}w_1^{a(T)}w_2^{b(T)}\leq 2^{-\left\lfloor t/\Delta\right\rfloor-1}\cdot \Delta^{-1}.
\end{equation}
Let $f_1$ be the sum of the weights over all proper subtrees of $\+T_{a}$.
Similarly, Let $f_2$ be the sum of the weights over all proper subtrees of $\+T_{b}$.
By \Cref{GWprocess2}, one can verify that
\begin{align*}
    f_1 &= w_1x(1+f_1)^{k\Delta}\cdot (1+f_2)^{\Delta^2},\\
    f_2 &= w_2x^{\Delta}(1+f_1)^{k\Delta^2}\cdot (1+f_2)^{\Delta^3}.
\end{align*}
In addition, let $[x^m]f_1(x)$ be the coefficient of $x^m$ in $f_1$.
By \Cref{GWprocess2}, we also have 
\begin{align}\label{eq-def-f1}
\forall m\geq 0, \quad [x^m]f_1(x) = 
    \sum_{T\in \+T^m_{a}}w_1^{a(T)}w_2^{b(T)}.
\end{align}
Define
$$h\triangleq x(1+f_1)^{k\Delta}\cdot (1+f_2)^{\Delta^2}.$$
We have 
$f_2 = w_2h^{\Delta}$, $f_1 = w_1h$ and then
\begin{align*}
    h &= x(1+w_1h)^{k\Delta}\cdot \left(1+w_2h^{\Delta}\right)^{\Delta^2}.
\end{align*}
By applying the Lagrange inversion theorem, we have for each $m \geq 1$,
\begin{equation}
\begin{aligned}\label{eq-efficiency-generationfuncition}
   \relax [x^{m}]h(x)&=\frac{1}{m}[u^{m-1}]\left((1+w_1u)^{k\Delta}\cdot (1+w_2u^{\Delta})^{\Delta^2}\right)^{m}\\
   &= \frac{1}{m}\sum\limits_{i=0}^{\lfloor m/\Delta\rfloor}\left([u^{m-1-\Delta i}]\left(1+w_1u\right)^{k\Delta m}\cdot [u^{\Delta i}]\left(1+w_2u^{\Delta}\right)^{\Delta^2 m} \right)
\end{aligned}
\end{equation}

Assume $m\geq 3$ and $0\leq i\leq \lfloor\frac{m}{2\Delta}\rfloor$.
Then we have $m\leq 4(m-1-\Delta i)$.
In addition, 
for each $0< \gamma\leq \beta$ where $\gamma,\beta$ are integers,
\begin{equation}\label{stirling}
    \binom{\beta}{\gamma}\leq \left(\frac{\mathrm{e}\beta}{\gamma}\right)^{\gamma}.
\end{equation}
Thus, we have
\begin{equation}\label{eq-generatingfunction-part1}
\begin{aligned}
    &\quad [u^{m-1-\Delta i}]\left(1+w_1u\right)^{k\Delta m}\\
    &= \binom{k\Delta m}{m-1-\Delta i}w_1^{m-1-\Delta i}\\
(\text{by \eqref{stirling}})\quad   &\leq \left(\frac{\mathrm{e}k\Delta m}{m-1-\Delta i}\right)^{m-1-\Delta i}w_1^{m-1-\Delta i}\\
(\text{by $m\leq 4(m-1-\Delta i)$})\quad    &\leq (4\mathrm{e}k\Delta w_1)^{m-1-\Delta i}\\
\left(\text{by $w_1=(8\mathrm{e}k\Delta)^{-1}$}\right)\quad &= 2^{1+\Delta i-m}.
\end{aligned}
\end{equation}

Similarly, assume $m\geq 3$ and $\lfloor\frac{m}{2\Delta}\rfloor< i\leq m$.
We have $m\leq 2\Delta i$.
Thus, we have
\begin{equation}\label{eq-generatingfunction-part2}
\begin{aligned}
    &\quad[u^{\Delta i}]\left(1+w_2u^\Delta\right)^{\Delta^2 m}\\
    &=\binom{\Delta^2 m}{ i}w_2^{i}\\
(\text{by \eqref{stirling}})\quad  &\leq \left(\frac{\mathrm{e}\Delta^2 m}{i}\right)^{i}w_2^{i}\\
(\text{by $m\leq 2\Delta i$})\quad  &\leq (2\mathrm{e}\Delta^3 w_2)^{i}\\
    \left(\text{by $w_2=(4\mathrm{e}\Delta^3)^{-1}$}\right)\quad &= 2^{-i}.
\end{aligned}
\end{equation}
Combining \eqref{eq-efficiency-generationfuncition} with \eqref{eq-generatingfunction-part1} and \eqref{eq-generatingfunction-part2}, we have if $m\geq 3$, 
\begin{align*}
    \relax [x^{m}]h(x)
    \leq  m^{-1}\left(\sum\limits_{i=0}^{\lfloor m/\Delta\rfloor} 2^{1-m+\Delta i-i} \right)
    \leq  2^{1-\lfloor m/\Delta\rfloor}.
\end{align*}

For the case when $m=1$ and $m=2$, one can also verify $[x^{m}]h(x)\leq 2^{1-\lfloor m/\Delta\rfloor}$ directly from \eqref{eq-efficiency-generationfuncition}.
Therefore, by $f_1 = w_1h$ and $w_1=(8\mathrm{e}k\Delta)^{-1}$, we have
\[
[x^m]f_1(x)=w_1[x^m]h(x)\leq 2^{-\lfloor m/\Delta\rfloor-1}\cdot \Delta^{-1}.
\]
Combining with \eqref{eq-def-f1}, \eqref{eq-sumo-tta} is immediate.

In the following, we present the injection from the generalized $\{2,3\}$-trees in $\MSC{T}_{v}^t$ to the $\{a,b\}$-labelled trees in $\+T_a^t$ as claimed.
Then the lemma is proved.
Given integers $r,\ell\geq 0$ and $z\in \{a,b\}$,
define $\+T^{r,\ell}_z\triangleq \{\text{proper subtree $T$ of $\+T_{z}$}\mid (a(T)=r)\land  (b(T) = \ell)\}$. 
To construct an injection from $\MSC{T}_{v}^t$ to $\+T_a^t$ such that each $T\in \MSC{T}_{v}^t$ is mapped to some $T'\in \+T_a^t$ satisfying $\abs{T\cap V}=a(T')$ and $\abs{T\cap \+E}=b(T')$,
it is sufficient to construct an injection from $\MSC{T}_{v}^{r,\ell}$ to $\+T^{r,\ell}_a$ for each integers $r,\ell\geq 0$.
In the following, we present the injection from $\MSC{T}_{v}^{r,\ell}$ to $\+T^{r,\ell}_a$.
Given a generalized $\{2,3\}$-trees $T\in\MSC{T}_{v}^{r,\ell}$,
by \Cref{WTdef} we have $G(T)$ has a rooted directed spanning tree with root $v$.
Choose such a spanning tree $\phi(T)$ arbitrarily. 
Obviously, $\phi(T)\neq \phi(T')$ for different $T,T'\in \MSC{T}_{v}^t$.
Thus, $\phi$ is an injection from the set $\MSC{T}_{v}^{r,\ell}$ to the set $\{\phi(T)\mid T\in \MSC{T}_{v}^{r,\ell}\}$.
For each node $u$ in $\phi(T)$,
we have either $u\in V$ or $u\in \+{E}$.
In addition,
by $T\in\MSC{T}_{v}^{r,\ell}$,
we have $\abs{T\cap V} =r$, $\abs{T\cap\+{E}} = \ell$.
Combining with $\phi(T)$ is a spanning tree of $G(T)$,
we have
the node set of $\phi(T)$ is $T$ and then 
$\#\{\text{nodes of $\phi(T)$ in $V$} \} = r$, $\#\{\text{nodes of $\phi(T)$ in $\+E$} \} = \ell$.

Given $H = (V, \+E)$ and $u\in V$,
let $V(u)\subseteq V$ be the set of vertices $u'$ such that there exists $e\in \+{E}$ satisfying $u,u'\in e$, and
$\+E(u)\subseteq \+{E}$ be the set of hyperedges $e$
such that there exists $e'\in \+E$ such that $u\in e'\land \text{dist}_{\Lin{H}}(e',e)=1$.
Similarly, 
given $e\in \+E$, 
let $V(e)\subseteq V$ be the set of vertices $u$ such that there exists $e'\in \+{E}$ satisfying $u\in e'\land \text{dist}_{\Lin{H}}(e,e')=1\text{ or }2$,
and $\+E(e)\subseteq \+{E}$ be the set of hyperedges $e'$ 
such that $\text{dist}_{\Lin{H}}(e,e')=2\text{ or }3$.
Moreover, recall that each hyperedge in $H$ contains at most $k$ vertices and shares vertices with at most $\Delta$ hyperedges.
We have 
for each $u\in V$ and $e\in \+{E}$, $\abs{V(u)}\leq k\Delta$,
$\abs{\+{E}(u)}\leq \Delta^2$,
$\abs{V(e)}\leq k\Delta^2$,
$\abs{\+{E}(e)}\leq \Delta^3$.
Let $\mathsf{T}_v$ be a rooted directed tree with infinite nodes such that
\begin{itemize}
    \item each node is labelled with some $u\in V\cup \+{E}$ and the root is labelled $v$;
    \item for each $u\in V\cup \+{E}$ and $u'\in V(u)\cup \+E(u)$, each node labelled $u$ has exactly one child labelled $u'$.
\end{itemize}
Given integers $r,\ell\geq 0$, let $\mathsf{T}^{r,\ell}_v$ be the set of subtrees $T$ of $\mathsf{T}_v$ such that the root of $T$ is the root of $\mathsf{T}_v$ and 
$\#\{\text{nodes in $T$ with labels from $V$} \} = r$, $\#\{\text{nodes in $T$ with labels from $\+E$} \} = \ell$.
In addition, by the definition of $\mathsf{T}_v$, one can verify that 
\begin{itemize}
    \item the root of $\mathsf{T}_v$ has a label from $V$;
    \item each node labelled some $u
    \in V$ has $\abs{V(u)}\leq k\Delta$ children with labels from $V$ and $\abs{\+E(u)}\leq \Delta^2$ children with labels from $\+E$;
    \item each node labelled some $u
    \in \+E$ has $\abs{V(u)}\leq k\Delta^2$ children with labels from $V$ and $\abs{\+E(u)}\leq \Delta^3$ children with labels from $\+E$.
\end{itemize}
Combining with the definitions of $\+T_{a}$, $\+T_{a}^{r,\ell}$ and $\mathsf{T}^{r,\ell}_v$, one can verify that 
$\abs{\mathsf{T}^{r,\ell}_v} \leq\abs{\+T_{a}^{r,\ell}}$.
Thus, 
there exists an injection $\psi$ from $\mathsf{T}^{r,\ell}_v$ to $\+T_{a}^{r,\ell}$. 

In addition, by \Cref{WTdef}, for each $T\in\MSC{T}_{v}^{r,\ell}$ and each arc from $u$ to $u'$ in $G(T)$,
we have $u\in T \subseteq V\cup \+{E}$ and $u'\in V(u)\cup \+E(u)$.
Combining with that $\phi(T)$ is a rooted directed spanning tree of $G(T)$,
we have for each node $u$ in $\phi(T)$,
the children of $u$ are from $V(u)\cup \+E(u)$.
Thus, $\phi(T)$ is a subtree of $\mathsf{T}_v$.
Combining with that $v$ is the root of $\phi(T)$ and $\#\{\text{nodes of $\phi(T)$ in $V$} \}=r$, $\#\{\text{nodes of $\phi(T)$ in $E$} \} = \ell$,
we have $\phi(T) \in \mathsf{T}^{r,\ell}_v$.
Formally, $\{\phi(T)\mid T\in \MSC{T}_{v}^{r,\ell}\}\subseteq \mathsf{T}^{r,\ell}_v$.
Recall that $\phi$ is an injection from the set $\MSC{T}_{v}^{r,\ell}$ to the set $\{\phi(T)\mid T\in \MSC{T}_{v}^{r,\ell}\}$.
We have $\phi$ is an injection from  $\MSC{T}_{v}^{r,\ell}$ to $\mathsf{T}^{r,\ell}_v$.
Recall that $\psi$ is an injection from $\mathsf{T}^{r,\ell}_v$ to $\+T_{a}^{r,\ell}$.
We have $\psi\uplus\phi$ is an injection from $\MSC{T}_{v}^{r,\ell}$ to $\+T_{a}^{r,\ell}$. 
Thus the lemma follows.
\end{proof}

Let $\sigma\in\+{Q}^*$ be a partial assignment such that only one variable $v\in V$ has $\sigma(v) = \star$.
The following lemma bounds the length $\ell$ of $\pth(\sigma) = (\sigma_0,\sigma_1,\ldots,\sigma_{\ell})$ and further shows that if the set of bad variables/constraints given $\sigma_{\ell}$ becomes too large, 
then there exists a large generalized $\{2,3\}$-tree $T$ in $H_{\Phi}$ such that the event $\+{E}^{\sigma}_T$ as in \Cref{def:cbad} happens.

\begin{lemma}\label{big23tree}
Let $\sigma\in\+{Q}^*$ 
be a partial assignment with exactly one variable $v\in V$ having $\sigma(v) = \star$, 
and let $\pth(\sigma) =
(\sigma_0,\sigma_1,\ldots,\sigma_{\ell})$. Suppose $\abs{ \vst{\sigma_{\ell}}}+\abs{\csfrozen{\sigma_{\ell}}}\geq  L$ for some integer $L\geq 1$, then there exists a generalized $\{2,3\}$-tree $T=U\uplus E$ of $H_{\Phi}$ with root $v$ such that $\+{E}^{\sigma}_T$ happens
and $\abs{U} + \Delta\cdot \abs{E}\geq L$.
\end{lemma}

Before proving  \Cref{big23tree},  we show we can already prove \Cref{thm-expect-depth} using \Cref{big23tree}.
\begin{proof}[Proof of \Cref{thm-expect-depth}]
The case $i=0$ is trivial. In the following, we assume $i >0$.
Recall $V=\{v_1,v_2,\ldots,v_n\}$.
We claim that for each possible $\sigma$,
given $X^t_0 = \sigma$,
if $\abs{\vstar{X^t_{\ell}}}+\abs{\csfrozen{X^t_{\ell}}} \geq i\Delta$,
then $\mathcal{E}^t_T$ happens for some $L\geq i\Delta$ and $T\in \MSC{T}_{v_t}^{L}$. 
Formally,
\begin{equation}\label{eq-formal-pr-upper-ett}
\Pr{\left(\abs{\vstar{X^t_{\ell}}}+\abs{\csfrozen{X^t_{\ell}}} \!\geq i\Delta\right)\!\land\! \left(X^t_0 = \sigma\right)}
\leq \Pr{\left(\exists L\geq i\Delta, T\in \MSC{T}_{v_t}^{L} \text{ s.t. }\mathcal{E}^t_T \text{ happens}\right)\!\land \!\left( X^t_0 = \sigma\right)}.
\end{equation}
Combining with the union bound, we have
\begin{equation}
\begin{aligned}\label{eq-eventvstarcstar-eventetthappens}
\Pr{\left(\abs{\vstar{X^t_{\ell}}}+\abs{\csfrozen{X^t_{\ell}}} \geq i\Delta\right)\land \left(X^t_0 = \sigma\right)}
\leq \sum_{L\geq i\Delta}\sum\limits_{T\in \MSC{T}_{v_t}^{L}}\Pr{ \+{E}^{t}_T\land \left(X^t_0 = \sigma\right)}.
\end{aligned}
\end{equation}
Therefore, we have 
\begin{align}
 &\Pr{\abs{\vstar{X^t_{\ell}}}+ \abs{\csfrozen{X^t_{\ell}}} \geq i\Delta} \cdot \mathbb{E}\left[\chi(X^t_0)\;\Big\vert\; \abs{\vstar{X^t_{\ell}}}+ \abs{\csfrozen{X^t_{\ell}}} \geq i\Delta\right]\notag\\
=&\Pr{\abs{\vstar{X^t_{\ell}}}+ \abs{\csfrozen{X^t_{\ell}}} \geq i\Delta} \cdot \sum_{\sigma}\left(\chi(\sigma)\Pr{X^t_0 = \sigma\;\Big\vert\; \abs{\vstar{X^t_{\ell}}}+\abs{\csfrozen{X^t_{\ell}}} \geq i\Delta}\right)\notag\\
=& \sum_{\sigma}\left(\chi(\sigma)\Pr{\left(X^t_0 = \sigma\right)\land \left(\abs{\vstar{X^t_{\ell}}}+\abs{\csfrozen{X^t_{\ell}}} \geq i\Delta\right)}\right)\label{eq-boundprexpectation-1}\\
\leq & \sum_{\sigma}\chi(\sigma)\left(\sum\limits_{L\geq i\Delta}\sum\limits_{T\in \MSC{T}_{v_t}^{L}}\Pr{ \+{E}^{t}_T\land \left(X^t_0 = \sigma\right)}\right)= \sum\limits_{L\geq i\Delta}\sum\limits_{T\in \MSC{T}_{v_t}^{L}}\sum_{\sigma}\left(\chi(\sigma)\Pr{ \+{E}^{t}_T\land \left(X^t_0 = \sigma\right)}\right)\notag,
\end{align}
where the first equality is by the law of total expectation, the second equality is by the chain rule, and the inequality is by the 
non-negativity of $\chi(\cdot)$ and \eqref{eq-eventvstarcstar-eventetthappens}.
Moreover, by the law of total expectation, we also have
\begin{equation}\label{eq-sum-chisimgapr}
\begin{aligned}
\sum_{\sigma}\left(\chi(\sigma)\Pr{ \+{E}^{t}_T\land \left(X^t_0 = \sigma\right)}\right)
=   \sum_{\sigma}\left(\chi(\sigma)\Pr{ X^t_0 = \sigma\, \vert\,  \+{E}^{t}_T}\Pr{\+{E}^{t}_T}\right)
\leq \Pr{ \+{E}^{t}_T}\E{\chi(X^t_0)\, \vert\, \+{E}^{t}_T}.
\end{aligned}
\end{equation}
Thus, we have 
\begin{equation*}
\begin{aligned}
 &\Pr{\abs{\vstar{X^t_{\ell}}}+ \abs{\csfrozen{X^t_{\ell}}} \geq i\Delta} \cdot \mathbb{E}\left[\chi(X^t_0)\;\Big\vert\; \abs{\vstar{X^t_{\ell}}}+ \abs{\csfrozen{X^t_{\ell}}} \geq i\Delta\right]\\
 (\text{by \eqref{eq-boundprexpectation-1} and \eqref{eq-sum-chisimgapr}})\quad =& \sum\limits_{L\geq i\Delta}\sum\limits_{T\in \MSC{T}_{v_t}^{L}}\left(\Pr{ \+{E}^{t}_T}\cdot\E{\chi(X^t_0)\mid \+{E}^{t}_T}\right)\\
 = & \sum\limits_{j\geq i}\sum\limits_{r=0}^{\Delta-1}\sum\limits_{T\in \MSC{T}_{v_t}^{j\Delta+r}}\left(\Pr{ \+{E}^{t}_T}\cdot\E{\chi(X^t_0)\mid \+{E}^{t}_T}\right)\\
 (\text{by \Cref{badtree}})\quad \leq &\sum\limits_{j\geq i}\sum\limits_{r=0}^{\Delta-1}\sum\limits_{\substack{T\in \MSC{T}_{v_t}^{j\Delta+r}}}\left(\left(8\mathrm{e}k\Delta\right)^{-\abs{T\cap V }}\cdot \left(4\mathrm{e}\Delta^3\right)^{-\abs{T\cap \+C }}\right)\\
 (\text{by \Cref{lem-prob-g23tree}})\quad\leq &\sum\limits_{j\geq i}\sum\limits_{r=0}^{\Delta-1}(2^{-j-1}\cdot \Delta^{-1})\\
= & 2^{-i}.
\end{aligned}
\end{equation*}

In the following, we prove \eqref{eq-formal-pr-upper-ett}. Then the theorem is proved.
Given a possible assignment $\sigma$ of $X^t_0$,
suppose $X^t_0 = \sigma$ and $\abs{\vstar{X^t_{\ell}}}+\abs{\csfrozen{X^t_{\ell}}} \geq i\Delta$.
Then, by \eqref{eq:simulate-path} we have
$\abs{\vstar{\sigma_{\ell}}}+ \abs{\csfrozen{\sigma_{\ell}}} \geq i\Delta$.
We claim $\sigma = X^{t-1}_{{v_{t}}\gets{\star}}$.
In addition, we have $X^{t-1}(v)\neq \star$ for any $v\in V$, because by \eqref{eq:simulate-path},
$X^{t-1}$ is generated from $X^0 =\hollowstar^V$ with \Cref{Alg:main-eq}
and no vertex in \Cref{Alg:main-eq} is set as $\star$.
Combining with $\sigma = X^{t-1}_{{v_{t}}\gets{\star}}$ and \eqref{eq:simulate-path}, 
we have there exists only one variable $v_t \in V$ such that $\sigma(v_t) = \star$.
Combining with \Cref{big23tree} and $\abs{\vstar{\sigma_{\ell}}}+ \abs{\csfrozen{\sigma_{\ell}}} \geq i\Delta$,
we have $\mathcal{E}^{\sigma}_T$ happens for some $L\geq i\Delta$ and $T\in \MSC{T}_{v_t}^{L}$. 
In addition, by $X^t_0 = \sigma$ and the definitions of $\+{E}^{t}_T$ and $\+{E}^{\sigma}_T$ in \Cref{def:cbad},
we have $\mathcal{E}^{\sigma}_T$ is exact $\+{E}^{t}_T$.
Therefore, $\mathcal{E}^{t}_T$ happens for some $L\geq i\Delta$ and $T\in \MSC{T}_{v_t}^{L}$
and \eqref{eq-formal-pr-upper-ett} is immediate.

At last, we show the claim $\sigma = X^{t-1}_{{v_{t}}\gets{\star}}$ to finish the proofs of \eqref{eq-formal-pr-upper-ett} and the lemma.
By \eqref{eq:simulate-path} we have either $\sigma = X^t_0 = X^{t-1}$ or $\sigma = X^t_0 = X^{t-1}_{{v_{t}}\gets{\star}}$. 
Thus, it is sufficient to show $\sigma \neq X^{t-1}$.
Suppose $\sigma = X^{t-1}$ for contradiction.
Recall that $X^{t-1}(v)\neq \star$ for any $v\in V$.
We have $\sigma(v)=X^{t-1}(v)\neq \star$ for any $v\in V$.
Thus, by \Cref{pathdef} we have $\pth(\sigma)=\sigma$ and $\sigma_{\ell
} = \sigma$.
Thus, $\sigma_{\ell
}(v)\neq \star$ for any $v\in V$.
Combining with \Cref{def:cbad},
we have $\vst{\sigma_{\ell
}}=\emptyset$ and 
$\csfrozen{\sigma_{\ell
}}=\emptyset$,
which is contradictory with 
$\abs{\vstar{\sigma_{\ell}}}+ \abs{\csfrozen{\sigma_{\ell}}} \geq i\Delta$.
Thus, we have $\sigma \neq X^{t-1}$
and then $\sigma = X^{t-1}_{{v_{t}}\gets{\star}}$.
This finishes the proofs of \eqref{eq-formal-pr-upper-ett} and the lemma.
\end{proof}

To prove \Cref{big23tree}, we need to introduce the definition of $\gvc$, an undirected graph with a vertex set over all variables and constraints of the CSP formula.
\begin{definition}[Graph of variables and constraints]\label{def:graphvc}
Let $\Phi =(V,\+Q,\+C)$ be the CSP formula.
Define  $\gvc=(V\cup \+{C},E)$ as the graph where vertices are $V\cup \+{C}$ and there is an edge between two vertices $u,v$ if and only if one of the following holds:
\begin{enumerate}
    \item $u,v\in V$ and there exists some $c\in \+{C}$ such that $u,v\in \var{c}$.\label{graphvc-1}
    \item $u,v\in \+{C}$ and $\text{dist}_{\Lin{H_\Phi}}(u,v)=1\text{ or }2$.\label{graphvc-2}
    \item $u\in V,v\in \+{C}$ and  there exists some $c\in \+{C}$ such that $u\in \var{c}\land \text{dist}_{\Lin{H_\Phi}}(c,v)=1$.\label{graphvc-3}
\end{enumerate}
Furthermore, for any $S\subseteq V\cup \+C$, we let $\gvc(S)$ denote the subgraph of $\gvc$ induced by $S$.
\end{definition} 

Recall in \Cref{definition:boundary-variables}:   for any $\sigma\in \qs$, $\hfix{\sigma}$ denotes the sub-hypergraph of $H^\sigma$ induced by $V^{\sigma}\cap\vfix{\sigma}$. 
Recall that for each $v\in V^{\sigma}$,
$H_v^{\sigma}=(V_v^\sigma,\+{C}_v^{\sigma})$ denotes the connected component in $H^{\sigma}$ that contains the vertex/variable $v$. 
Also, for each $c\in \+C$, we denote the simplified constraint of $c$ under $\sigma$ as~$c^{\sigma}$.

The following lemma states a connectivity property on the graph $\gvc$.

\begin{lemma}\label{lem-connect2}
Assume the condition of \Cref{big23tree}.
Then $\gvc\left(\csfrozen{\sigma_i}\cup \vst{\sigma_i}\right)$ is connected for each $0 \leq i \leq \ell$.

\end{lemma}
\begin{proof}
We prove this lemma by induction on $i$. 
For simplicity, we say a variable or constraint $c$ is connected to a subset $S\subseteq V\cup \+{C}$ in  $\gvc$ if $c$ is connected to some $c'\in S$. The base case is when $i=0$.
By the condition of the lemma, $v$ is the only variable satisfying
$\sigma(v) = \star$.
Combining with $\sigma_0 = \sigma$, 
we have $v$ is the only variable satisfying $\sigma_0(v) = \star$. Therefore,
$\vst{\sigma_0} = \set{v}$.
In addition, 
we have the following claim: each $c\in \csfrozen{\sigma_0}$ is connected to $v$ in $\gvc(\csfrozen{\sigma_0}\cup \vst{\sigma_0})$.
Combining with the claim,
we have $\gvc(\csfrozen{\sigma_0}\cup \vst{\sigma_0})$ is connected.

Now we prove the claim, which completes the proof of the base case.
By $c\in \csfrozen{\sigma_0}$,
we have $c\in \ccon{\sigma_0}\cap \cfrozen{\sigma_0}$.
By $c\in \ccon{\sigma_0}$ and \Cref{definition:boundary-variables},
we have $\vcon{\sigma_0}\cap \var{c}\neq \emptyset$.
Combining with $v$ is the only variable satisfying $\sigma_0(v) = \star$ and the definition of $\vcon{\sigma_0}$,
we have there exists a connected path $c^{\sigma_{0}}_1,c^{\sigma_{0}}_2,\dots,c^{\sigma_{0}}_{t}=c^{\sigma_{0}}\in \+C^{\sigma_{0}}$ such that $\sigma_{0}(v)=\star$, $v \in \var{c^{\sigma_0}_1}$
and $\var{c^{\sigma_{0}}_j}\subseteq V^{\sigma_0}\cap \vfix{\sigma_{0}}$ for each $j <t$.
If $c = c_1$,  then $v \in \var{c}$ and the claim is immediate by the definition of $\gvc$.
In the following, we assume $c \neq c_1$. Let $w_j\in \left(\var{c^{\sigma_0}_{j}}\cap\var{c^{\sigma_0}_{j+1}}\right)$ for each $j<t$.
    Then $w_j\not \in \Lambda(\sigma_{0})$.
By $w_j\in \var{c^{\sigma_0}_j}$ and
$\var{c^{\sigma_0}_j}\subseteq \vfix{\sigma_{0}}$,
we have $w_j\in \vfix{\sigma_{0}}$.
Combining with $w_j\not \in \Lambda(\sigma_{0})$,
we have either $\sigma_{0}(w_j) = \star$, where we set $\widehat{c}_j=w_j$; or 
$w_j \in \var{\widehat{c}_j}$ for some $\widehat{c}_j\in \cfrozen{\sigma_{0}}$. Note that $\widehat{c}_j$ can be either a variable or a constraint. In the former case, we have $\widehat{c}_j\in \vst{\sigma_0}$. In the latter case,
 By $w_j$ is connected to $v$ in $\hfix{\sigma_0}$ through the path $c^{\sigma_0}_1,c^{\sigma_0}_2,\dots,c^{\sigma_0}_{j}$, 
we have 
$w_j\in \vcon{\sigma_0}$.
Thus, we have $\widehat{c}_j \in \ccon{\sigma_{0}}$ by \Cref{definition:boundary-variables}.
Combining with $\widehat{c}_j\in \cfrozen{\sigma_{0}}$,
we have $\widehat{c}_j\in \csfrozen{\sigma_{0}}$.
In summary, we always have $\widehat{c}_j\in \vst{\sigma_0}\cup \csfrozen{\sigma_{0}}$.
\sloppy Moreover, for each $j < t-1$, 
if $\widehat{c}_j\in \+C$, we have $w_j \in \var{c^{\sigma_0}_{j+1}}\cap \var{\widehat{c}^{\sigma_0}_j}$, otherwise we have $\widehat{c}_j=w_j$.
Thus by \Cref{def:graphvc}, it can be verified that $\widehat{c}_j$ and $\widehat{c}_{j+1}$ are adjacent in $\gvc$.
In addition, if $\widehat{c}_1\in \+C$, we have
$w_1 \in \var{c_1}\cap \var{\widehat{c}_1}$ , otherwise we have $\var{\widehat{c}_1}=w_1\in \var{c_1}$, hence $c_1$ and $\widehat{c}_{1}$ are adjacent in $\gvc$. Similarly, we have  $\widehat{c}_{t-1}$ and $c_{t}$ are adjacent in $\gvc$
Thus, we have 
$v,c_1,\widehat{c}_1,\widehat{c}_2,\dots,\widehat{c}_{t-1},c_{t}=c$ is a connected path in $\gvc$.
Combining with $v \in \vst{\sigma_{0}}$ and
$\widehat{c}_j \in \csfrozen{\sigma_{0}}$ for each $j< t$, the claim is immediate.

For the induction step, we prove this lemma for each $i>0$.
We claim that each $v\in \vst{\sigma_i}$ is connected to $\vst{\sigma_{i-1}}$ in $\gvc(\csfrozen{\sigma_i}\cup \vst{\sigma_i})$. In addition,
we can prove each  $c\in \csfrozen{\sigma_{i}}$ is connected to  $\vst{\sigma_{i-1}}$ in $\gvc(\csfrozen{\sigma_i}\cup \vst{\sigma_i})$ by a similar argument to the base case.
Moreover, by the induction hypothesis we have $\gvc(\csfrozen{\sigma{i-1}}\cup \vst{\sigma_{i-1}})$ is connected.
Combining with $\csfrozen{\sigma_{i-1}} \subseteq \csfrozen{\sigma_{i}}$ and $\vst{\sigma_{i-1}} \subseteq \vst{\sigma_{i}}$ by \Cref{cor-mono},
 we have $\gvc(\csfrozen{\sigma_{i}}\cup \vst{\sigma_{i}})$ is connected.

Now we prove the claim that each $v\in \vst{\sigma_i}$ is connected to $\vst{\sigma_{i-1}}$ in $\gvc$, which completes the proof of the lemma.
If $v\in \vst{\sigma_{i-1}}$, the claim is immediate by $\vst{\sigma_{i-1}}\subseteq \vst{\sigma_{i}}$.
In the following, we assume $v\in \vst{\sigma_i}\setminus \vst{\sigma_{i-1}}$, where by \Cref{pathdef}  we have $v=\nextvar{\sigma_{i-1}}$. 
By the definition of $\nextvar{\cdot}$, we have $v\in \vinf{\sigma_{i-1}}$ and then $v \in \var{\widehat{c}}$ for some constraint $\widehat{c} \in \ccon{\sigma_{i-1}}$.
In addition, by $\widehat{c}\in \ccon{\sigma_{i-1}}$ one can verify that
there exists a variable $w\neq v$ and a connected path $c^{\sigma_{i-1}}_1,c^{\sigma_{i-1}}_2,\dots,c^{\sigma_{i-1}}_{t} = \widehat{c}^{\sigma_{i-1}}\in\+C^{\sigma_{i-1}}$ such that $\sigma_{i-1}(w)=\star$, $w \in \var{c^{\sigma_{i-1}}_1}$ and $\var{c^{\sigma_{i-1}}_j}\subseteq V^{\sigma_{i-1}}\cap \vfix{\sigma_{i-1}}$ for each $j<t$.
Then there are two possibilities for $\widehat{c}$.
\begin{itemize}
\item If $\widehat{c} = c_1$, we have $v,w\in \var{c_1}$.
Therefore,
$v$ is connected to $w$ in $\gvc(\csfrozen{\sigma_{i-1}}\cup \vst{\sigma_{i-1}}\cup \{v\})$. Also by $\sigma_{i-1}(w)=\star$ we have $w\in \vst{\sigma_{i-1}}$. In addition, by \Cref{cor-mono} we have 
$$\csfrozen{\sigma_{i-1}}\cup \vst{\sigma_{i-1}}\cup \{v\}\subseteq \csfrozen{\sigma_{i-1}}\cup \vst{\sigma_{i}}\subseteq \csfrozen{\sigma_{i}}\cup \vst{\sigma_{i}}.$$
Thus the claim is immediate.
\item Otherwise, $\widehat{c} \neq c_1$.
Similarly to the base case, 
one can find a connected path $c_1,\widehat{c}_1,\widehat{c}_2,\dots,\widehat{c}_{t-1},\\c_{t}=\widehat{c}$ in $\gvc$,
where $w \in \var{c_1}$,
$\widehat{c}_j \in \vst{\sigma_{i-1}}\cup \csfrozen{\sigma_{i-1}}$ for each $j< t$, 
and there exists $w_{t-1} \in \var{c_{t}}\cap \var{\widehat{c}_{t-1}}$.
Recall that $v\in \var{\widehat{c}}$ and $w\in \var{c_1}$
Thus, $w,c_1,\widehat{c}_1,\widehat{c}_2,\dots,\widehat{c}_{t-1},v$ is also a connected path in $\gvc$.
Combining with $w\in \vst{\sigma_{i-1}}$,
$\widehat{c}_j \in \vst{\sigma_{i-1}}\cup \csfrozen{\sigma_{i-1}}$ for each $j< t$,
we have $v$ is connected to $\vst{\sigma_{i-1}}$ in $\gvc(\csfrozen{\sigma_i}\cup \vst{\sigma_i})$.
Thus the claim holds.
\end{itemize}
\end{proof}

The following crucial technical lemma states that when the set of ``bad" variables and constraints  is large, there exists a large generalized $\{2,3\}$-tree capturing the occurence of such event.

\begin{lemma}
\label{bigWT}

Let $\sigma\in\qs$ 
be a partial assignment. If $\gvc\left(\csfrozen{\sigma}\cup \vst{\sigma}\right)$ is connected, then there always exists a generalized $\{2,3\}$-tree $T=U\uplus E$ in $H_{\Phi}$ with root $v$ such that 
\[
U=\vst{\sigma_{\ell}},
\quad
E\subseteq \csfrozen{\sigma}, 
\quad \text{ and }\quad
\Delta\cdot \abs{E}\geq \abs{\csfrozen{\sigma}}
\]
\end{lemma}

\begin{proof}
We construct a generalized $\{2,3\}$-tree $T$
and a rooted directed tree $T^{\ast}$ as follows.
Let $B$ denote a subset of variables and constraints and $R$ denote a subset of constraints.
For simplicity,  let $\gvc$ denote $\gvc\left(\csfrozen{\sigma}\cup \vst{\sigma}\right)$ .  
\begin{itemize}
    \item Initially, let $T=\{v\}$, $B = \csfrozen{\sigma}\cup \vst{\sigma}\setminus \{v\}$, $R=\emptyset$, and let $T^{\ast}$ be a tree with only one vertex $v$;
    \item while $B\neq \emptyset$
    \begin{itemize}
    \item Choose $u\in T,w\in B$ according to the following:
    \begin{enumerate}[label=(\roman*)]
        \item If there exists $u\in T,w\in B$ such that $\text{dist}_{\gvc}(u,w)=1$, choose any such $(u,w)$; \label{choosea}
        \item Otherwise, if there exists $u\in T\cap \+C,w\in B$ such that there exists $c\in \+C$  satisfying $\var{u}\cap \var{c}\neq \emptyset$ and $\text{dist}_{\gvc}(c,w)=1$, choose any such $(u,w)$; \label{chooseb}
        \item Otherwise, choose any $u\in T,w\in B$. \label{choosec}
    \end{enumerate}
        \item add $w$ to $T$, and add a node $w$ and an arc from $u$ to $w$ to $T^{\ast}$.
        \item update $B,R$ as follows:
        \begin{enumerate}[label=(\alph*)]
        \item If $w\in V$, update $B\gets B\setminus \{w\},R\gets R$\;  \label{removea}
        \item If $w\in \+{C}$, update $B\gets B\setminus \Gamma(w),R\gets R\cup \Gamma(w)$, where $\Gamma(w)=\{c\in \+{C}\mid \var{w}\cap \var{c}\neq \emptyset\}$.\label{removeb}
    \end{enumerate}
        \end{itemize}
\end{itemize}
Let $U = T\cap V$ and $E = T\cap \+C$.
We claim that when the above construction process stops, $T=U\uplus E$ is a generalized $\{2,3\}$-tree in $H_{\Phi}$ with root $v$ satisfying $
U=\vst{\sigma},
E\subseteq \cfrozen{\sigma}\text{ and }
\Delta\cdot \abs{E}\geq \cfrozen{\sigma}. $

We first show that $T$ is a generalized $\{2,3\}$-tree in $H_{\Phi}$ with root $v$.  From the construction process, we know each $w\in \csfrozen{\sigma}\cup \vst{\sigma}$ is either added into $T$ or removed from $B$ in \Cref{removeb} when some $c\in \+C$ is added into $T$ and $w\in \Gamma(c)$. We will simply refer to the latter case as ``removed in \Cref{removeb}'' for the rest of the proof. If $w$ is removed in \Cref{removeb}, then $w\in \+{C}$ and there exists $c\in \+{C}$ such that $\var{c}\cap \var{w}\neq \emptyset$ and $c$ is added into $T$.
Therefore for all distinct $u, v\in T\cap C$, we have $\text{dist}_{\Lin{H_{\Phi}}}(u,v)\geq 2$. 
Thus, by \Cref{WTdef}, to show $T$ is a generalized $\{2,3\}$-tree in $H_{\Phi}$,
it is sufficient to show the claim that $T^{\ast}$ is a rooted spanning tree of $G(T)$, where $G(T)$ is defined in \Cref{WTdef}.
In the next, we prove this claim.
By the construction process, we have $T^{\ast}$ is a rooted connected tree with the node set $T$ immediately.
Thus, it is sufficient to show that
each arc of $T^*$ is in $G(T)$.
For each $u\in T\setminus\{v\}$, let $w$ be the only father of $u$ in $T^*$. In other words, there is an arc from $w$ to $u$.
Then when the pair $u,w$ is chosen in the construction process, we have the following cases:
\begin{itemize}
    \item  $\text{dist}_{\gvc}(u,w)=1$. This corresponds to the case of \Cref{choosea}. In this case, if either $u\in V$ or $w\in V$, by comparing \Cref{def:graphvc} with the definition of $G(T)$ in \Cref{WTdef},
    one can verify that the arc from $u$ to $w$ must be an arc of $G(T)$.
    Otherwise, we have $u,w\in \+C$.
    By \Cref{removeb} of the construction process, we have $\text{dist}_{\Lin{H_\Phi}}(u,w)\geq 2$. By \Cref{def:graphvc} and  $\text{dist}_{\gvc}(u,w)=1$, we have $\text{dist}_{\Lin{H_\Phi}}(w,u)=2$.
    Thus, one can also verify that the arc from $w$ to $u$ is an arc of $G(T)$ by \Cref{WTdef}.
    
    \item Otherwise, $\text{dist}_{\gvc}(u,w)>1$.  Then the condition in \Cref{choosea} is not satisfied, otherwise, some $u'\in T,w'\in B$ where $\text{dist}_{\gvc}(u',w')=1$ rather than $u,w$ will be chosen
    in the process.
    Thus, we have $\text{dist}_{\gvc}(T,B)\triangleq \min\limits_{a\in T,b\in B}\text{dist}_{\gvc}(a,b)>1$ when the pair $u,w$ is chosen in the construction process. 
    Moreover, it is straightforward from the construction process that the three sets $T,R$ and $B$ form a partition of the vertex set of $\gvc$. 
    In addition, we have $\gvc$ is connected by assumption. Therefore $\text{dist}_{\gvc}(T\cup R, B)=1$.
    Combining with $\text{dist}_{\gvc}(T,B)>1$,
    we have $\text{dist}_{\gvc}(R,B)=1$ and there exist $c'\in R, w'\in B$ where $\text{dist}_{\gvc}(c',w')=1$. 
    By $c'\in R$ and \Cref{removeb} of the construction process, there must be some $u'\in \+{C}$ such that  
    $c'\in \Gamma(u')$ and $c'$ is removed from $B$ in \Cref{removeb} when $u'\in \+C$ is added into $T$. 
    Thus, we have $\var{c'}\cap \var{u'}\neq \emptyset$
    and $u'$ has been added to $T$ when the pair $u,w$ is chosen in the construction process.
    Thus, $u',w'$ satisfy the condition in \Cref{chooseb}. 
    Therefore, we have $u,w$ also satisfy the condition in \Cref{chooseb},
    otherwise, $u',w'$ rather than $u,w$ will be chosen
    in the process.
    Thus, we have $u\in T\cap \+{C}$ and
    there exists $c\in \+C$  satisfying $\var{u}\cap \var{c}\neq \emptyset$ and $\text{dist}_{\gvc}(c,w)=1$.
    If $w\in \+{C}$, by $\text{dist}_{\gvc}(c,w)=1$ and \Cref{def:graphvc},
    we have  $\text{dist}_{\Lin{H_\Phi}}(c,w)=1\text{ or }2$.
    Combining with $\var{u}\cap \var{c}\neq \emptyset$,
    we have 
    $\text{dist}_{\Lin{H_\Phi}}(u,w)=1,2\text{ or }3$.
    Combing with $\text{dist}_{\gvc}(u,w)>1$,
    we have $\text{dist}_{\Lin{H_\Phi}}(u,w)=3$.
    Otherwise, $\text{dist}_{\Lin{H_\Phi}}(u,w)=1 \text{ or } 2$.
    By \Cref{def:graphvc}, we have $\text{dist}_{\gvc}(u,w)=1$,
    which is a contradiction.
    Combining $u,w\in \+C$, $\text{dist}_{\Lin{H_\Phi}}(u,w)=3$ with \Cref{WTdef},
    we have the arc from $u$ to $w$ is also an arc of $G(T)$.
    If $w\in V$, by $\text{dist}_{\gvc}(c,w)=1$ and \Cref{def:graphvc},
    there exists some $c'\in \+{C}$ such that $w\in \var{c'}\land \text{dist}_{\Lin{H_\Phi}}(c,c')=1$.
    Combining with $\var{u}\cap \var{c}\neq \emptyset$,
    we have 
    $w\in \var{c'}\land \text{dist}_{\Lin{H_\Phi}}(u,c')= 1 \text{ or } 2$.
    Combining with \Cref{WTdef},
    we also have the arc from $u$ to $w$ is an arc of $G(T)$.
\end{itemize}
This shows that $T^*$ is a directed spanning tree of $G(T)$.
In addition, it is easy to verify that $v$ is the root of $T^*$ .
Therefore, $T$ is a generalized $\{2,3\}$-tree in $H_{\Phi}$ with root $v$.

At last, we show
$U=\vst{\sigma}$ and 
$E\subseteq \cfrozen{\sigma}$.
By the construction process, 
one can verify that $U=\vst{\sigma}$.
Moreover, by \Cref{removeb} of the process, at most $\abs{\Gamma(w)}\leq \Delta - 1$ constraints are moved from $B$ when a constraint $w\in \+{C}\cap B$ is added to $T$.
Combining with $\csfrozen{\sigma}\subseteq B$ in the initialization, 
we have $\Delta\cdot \abs{E}\geq \abs{\csfrozen{\sigma}}$. This completes the proof.
\end{proof}

Now we are ready to prove \Cref{big23tree}.
\begin{proof}[Proof of \Cref{big23tree}]

  By $\sigma\in\qs$ 
is a partial assignment with exactly one variable $v\in V$ having $\sigma(v) = \star$ and  \Cref{lem-connect2}, we have $\gvc\left(\csfrozen{\sigma_{\ell}}\cup \vst{\sigma_{\ell}}\right)$ is connected. Combining with \Cref{bigWT}, there exists a generalized $\{2,3\}$-tree $T=U\uplus E$ in $H_{\Phi}$ such that $U=\vst{\sigma_{\ell}}$, $E\subseteq \csfrozen{\sigma_{\ell}}$ and $\Delta\cdot\abs{E}\geq \abs{\csfrozen{\sigma_{\ell}}}$. It is straightforward to verify that
$\Delta\cdot \abs{E}+\abs{U}\geq L$ by the assumption that $\abs{\vst{\sigma_{\ell}}}+\abs{\csfrozen{\sigma_{\ell}}}\geq L$.  
\end{proof}

We then prove \Cref{thm-expect-depth2}. Recall the definition of generalized $\{2,3\}$-tree in \Cref{WTdef} and the definition of $H_{\Phi} = (V,\+C)$ in \Cref{sec:rejection-sampling}.
We have the following lemma which is similar to \Cref{big23tree}.

\begin{lemma}\label{lem-big23tree-rs}
For every $v\in V$, 
there exists a generalized $\{2,3\}$-tree $T = \{v\}\uplus E$ in $H_{\Phi}$ with root $v$ where
$
E\subseteq \cfrozen{X^n}$
and 
$
\Delta^2\abs{E}\geq \abs{\+{C}^{X^n}_v}$.
\end{lemma}

Before proving \Cref{lem-big23tree-rs}, we show we can already prove \Cref{thm-expect-depth2} using \Cref{lem-big23tree-rs}.

\begin{proof}[Proof of \Cref{thm-expect-depth2}]
The case $i=0$ is trivial. In the following, we assume $i>0$.  Recall the definition of $H_{\Phi} = (V,\+C)$ in \Cref{sec:rejection-sampling}.
Given integers $t,r,\ell\geq 0$ and $v\in V$, 
recall $\MSC{T}_{v}^t$ and $\MSC{T}_{v}^{r,\ell}$ in \Cref{WTdef}.
For any $i\geq 1$ and $v\in V$, we have 
\begin{equation}\label{eq-rejsamp1}
\begin{aligned}
    & \Pr{\abs{\+C^X_v}\geq 2i\Delta^2} \\
(\text{by \Cref{lem-big23tree-rs}}) \quad  \leq & \sum\limits_{j\geq 2i}\sum_{\substack{T\in \MSC{T}_v^{1,j}}}\Pr{T\cap \+C \subseteq \cfrozen{X}}\\
  (\text{by \Cref{badtree2}})\quad  \leq & \sum\limits_{j\geq 2i}\sum_{\substack{T\in \MSC{T}_v^{1,j}}}(4\mathrm{e}\Delta^3)^{-\abs{T\cap \+C}}\\
  =&8\mathrm{e}{k}\Delta\sum\limits_{j\geq 2i}\sum_{\substack{T\in \MSC{T}_v^{1,j}}}(8\mathrm{e}k\Delta)^{-1}(4\mathrm{e}\Delta^3)^{-\abs{T\cap \+C}}\\
  \left(\text{by $\MSC{T}_v^{1,j}\subseteq \MSC{T}_v^{j\Delta+1}$}\right)\quad\leq & 8\mathrm{e}{k}\Delta\sum\limits_{j\geq 2i}\sum_{\substack{T\in \MSC{T}_v^{j\Delta+1}}}(8\mathrm{e}k\Delta)^{-\abs{T\cap V}}(4\mathrm{e}\Delta^3)^{-\abs{T\cap \+C}}\\
  (\text{by \Cref{lem-prob-g23tree}})\quad\leq & 8\mathrm{e}{k}\Delta\cdot \sum\limits_{j\geq 2i}\left(2^{-j-1}\cdot \Delta^{-1}\right)\\
  \leq & 8\mathrm{e}{k}\cdot 4^{-i}.
\end{aligned}
\end{equation}
\end{proof}

We then finish \Cref{sec:boost} by proving \Cref{lem-big23tree-rs}. Let $X =X^n$ where $X^0,X^1,\dots,X^n$ is
the partial assignment sequence of \Cref{Alg:main} in \Cref{def-pas-main}.
Then we have the following lemma.

\begin{lemma}\label{lem-subset-rs}
$ V\setminus \Lambda(X) \subseteq \var{\cfrozen{X}}$.
\end{lemma}
\begin{proof}
Given $v_i \in V\setminus \Lambda(X)$ where $i\in [n]$, by Lines \ref{line-main-for}-\ref{line-main-sample} of \Cref{Alg:main} and \Cref{sampcor},
we have $v_i\in \vfix{X^{i-1}}$, otherwise, $v_i$ will be assigned a value from $Q_{v_i}$ in \Cref{line-main-sample}.
Moreover, again by Lines \ref{line-main-for}-\ref{line-main-sample} of \Cref{Alg:main}, we also have $v_i \notin \Lambda^+(X^{i-1})$.
Combining with $v_i\in \vfix{X^{i-1}}$, we have $v_i \in \cfrozen{X_{i-1}}$.
In addition, similar to \Cref{lem-mono}, 
one can also prove that
$\cfrozen{X^{j-1}}\subseteq \cfrozen{X^{j}}$ for each $j\in [n]$.
Then we have $\cfrozen{X^{i-1}}\subseteq \cfrozen{X}$ by induction.
Combining with $v_i \in \cfrozen{X_{i-1}}$,
we have $v_i\in \var{\cfrozen{X}}$.
Then the lemma follows.
\end{proof}

Now we can prove \Cref{lem-big23tree-rs}.

\begin{proof}[Proof of \Cref{lem-big23tree-rs}]
Let $\{\Phi_i^X= (V_i^X,\+{C}_i^X)\}\mid 1\leq i\leq K\}$ be the decomposition of $\Phi^X$.
If $v\not \in V_i$ for each $i\in [k]$, we have 
$\+C^X_v = \emptyset$ and
the lemma is trivial.
In the following, we assume
\emph{w.l.o.g.} that 
$v\in V^X_i$
for some $i\in K$.
Then we have 
$\Phi^X_v = (V^X_i,\+C^X_i)$.
Let 
$$S \triangleq  \left\{c\in \cfrozen{X}\mid c^X \in \+C^X_{i}\right\}.$$

At first, we prove that there exists some $c_v\in S$ such that $v \in \var{c_v}$.
By $v\in V^X_i$,
we have $v \not\in \Lambda(X)$. 
Combining with \Cref{lem-subset-rs},
we have there exists some $c_v \in \cfrozen{X}$ such that $v \in \var{c_v}$.
In addition, by $v \in \var{c_v}$ and $c_v \in \cfrozen{X}$, we also have $c^X_v \in \+C^X_{i}$.
Combining with $c_v \in \cfrozen{X}$, we have $c_v \in S$.

Now we prove $\abs{\+C^X_{i}}\leq \Delta\abs{S}$.
For each $c^X\in \+C^X_i$,
we have there exists a connected path $c^X_1,c^X_2,\dots,c^X_{t} = c^X\in \+C^X_i$ such that $v \in \var{c^X_1}$.
Let $v' \in \var{c^X}$.
We have $v' \not\in \Lambda(X)$.
Combining with \Cref{lem-subset-rs},
we have $v' \in \var{\widehat{c}}$ for some $\widehat{c}\in \cfrozen{X}$.
Then we have $\widehat{c}^{X} \in \+C^X_{i}$ because
there exists a connected path $c^X_1,c^X_2,\dots,c^X_{t},\widehat{c}^{X}\in \+C^X_i$ where $v \in \var{c^X_1}$.
Combining $\widehat{c}\in \cfrozen{X}$ with $\widehat{c}^{X} \in \+C^X_{i}$,
we have $c\in S$.
In summary, for each each $c^X\in \+C^X_i$, there exists some $\widehat{c}\in S$ such that $\var{c^X}\cap \var{\widehat{c}^X}\neq \emptyset$.
Thus, we have  $\abs{\+C^X_{i}}\leq \Delta\abs{S}$.

In the next, we prove that $\gvc(S)$ is connected.
It is enough to prove that any two different constraints $c,\widehat{c}\in S$ are connected in $\gvc(S)$.
Given $c,\widehat{c}\in S$, 
we have $c^X,\widehat{c}^X$ are in $\+C^X_{i}$.
Therefore, 
we have there exists a connected path $c^X = c^X_1,c^X_2,\dots,c^X_t = \widehat{c}^X\in \+C^X_{i}$.
If $t\leq 3$, obviously $c$ and $\widehat{c}$ are connected in $G^2(S)$.
In the following, we assume that $t>3$.
Let $w_j \in \left(\var{c^X_{j}}\cap \var{c^X_{j+1}}\right)$ for each $j< t$.
Then we have 
$w_j\not \in \Lambda(X)$.
Combining with \Cref{lem-subset-rs},
we have $w_j \in \var{\widehat{c}^X_j}$ for some $\widehat{c}_j\in \cfrozen{X}$.
Moreover, we also have $\widehat{c}^X\in \+C^X_{i}$,
because $\widehat{c}^X_j$ is connected to $c^X$ through $c^X_2,\dots,c^X_j\in \+C^X_{i}$.
Thus, we have $\widehat{c}_j\in S$.
In addition, for each $\widehat{c}_{j},\widehat{c}_{j+1}$ where $j< t-1$,
we have $\widehat{c}_{j}$ and $\widehat{c}_{j+1}$  are connected in $G^2(\+C)$,
because $w_j\in \var{\widehat{c}_{j}}\cap \var{c_{j+1}}$
and $w_{j+1}\in \var{\widehat{c}_{j+1}}\cap \var{c_{j+1}}$.
Thus, the constraints $c = c_1,\widehat{c}_1,\widehat{c}_2,\dots,\widehat{c}_{t-1},c_t = \widehat{c}$ forms a connected path in $\gvc$.
Combining with $\widehat{c}_j\in S$ for each $j\leq t-1$ and $c,\widehat{c}\in S$, we have the constraints $c,\widehat{c}$ are connected in $\gvc(S)$.

In summary, we have $c_v\in S\subseteq \cfrozen{X}$, $\Delta\abs{S}\geq \abs{\+C^X_{i}}$ and $\gvc(S)$ is connected. Combining with $v\in \var{c_v}$  we have $\gvc(S\cup \{v\})$ is also connected.
By going through the process in the proof of \Cref{bigWT}, 
we have there exists a subset of constraints and variables $T\subseteq S\cup \{v\}$ such that $T=\set{v}\uplus E$ is a generalized $\{2,3\}$-tree in $H_{\Phi}$ with root $v$ and $$\abs{E}\geq \abs{S}/\Delta\geq \abs{\+C^X_{i}}/\Delta^2 = \abs{\+C^X_{v}}/\Delta^2.$$

In addition, if $\abs{\+{C}^{X}_v}\geq L\Delta$, then it is straightforward to verify that
$\Delta\cdot \abs{E}+1\geq L$ and the lemma follows.
\end{proof}

\subsection{Basic tail bounds for bad events}\label{sec:raw-tail-bound}

In this subsection, we prove \Cref{badtree} and \Cref{badtree2}. For any constraint $c\in \+{C}$ and partial assignment  $\sigma\in \qs$, let 
$Z(\sigma,c)\defeq\abs{\var{c}\setminus \Lambda(\sigma) }$ denote 
the number of unassigned variables in $\var{c}$ under $\sigma$.
For any subset of variables and constraints $T=U \uplus E$ and partial assignment $\sigma\in \qs$,
define
\begin{equation}\label{eq-definition-g}
    g(\sigma,T)\triangleq \prod\limits_{v\in U\setminus \vst{\sigma}}\left(1-q_v\cdot\theta_v\right)\prod\limits_{c\in E}\left((0.99\pprime)^{-1}\mathbb{P}[\neg c\mid \sigma](1+\eta)^{Z(\sigma,c)}\right).
\end{equation}
The following lemma is the core of \Cref{badtree}, which shows that $g(\hollowstar^V,T)$
is an upper bound of the left hand of \eqref{eq-lemma-badtree}.

\begin{lemma}\label{WTprob2}
Given $t$, $(X^0,X^1,\ldots,X^{t-1},X^t_0,X^t_1,\ldots,X^t_{\ell})$ and $T=U \uplus E$ as in \Cref{badtree}, for any partial assignment $\sigma\in \qs$ and any integer $0\leq i\leq t-1$ where $\Pr{X^i=\sigma}>0$,  
we have
\begin{align}\label{eq-lemma-WTprob1}
\Pr{ \+{E}^{t}_T\mid X^i=\sigma}\cdot\E{\chi(X^t_{0})\mid \+{E}^{t}_T\land X^i=\sigma}\leq g(\sigma,T).\end{align}
Specifically, 
\begin{align}\label{eq-lemma-WTprob2}
\Pr{ \+{E}^{t}_T}\cdot\E{\chi(X^t_{0})\mid \+{E}^{t}_T}\leq g(\hollowstar^V,T).\end{align}
\end{lemma}

The following lemma is a special case of \eqref{eq-lemma-WTprob1}, given $i=t-1$ and $r_t=1$ where $r_t$ is defined as in \eqref{eq:simulate-path}.

\begin{lemma}\label{WTprob}
Let $\sigma\in\qs$ 
be a partial assignment satisfying $\mathbb{P}[\neg c\mid \sigma]\leq \pprime q$ for all $c\in \mathcal{C}$. 
For any subset of variables and constraints $T=U\uplus E$ such that the constraints in $E$ are disjoint, 
\begin{align}\label{eq-lemma-WTprob}
\Pr{ \+{E}^{\sigma}_T}\cdot\E{\chi(\sigma)\mid \+{E}^{\sigma}_T}\leq g(\sigma,T).\end{align}
\end{lemma}

The following lemma provides a useful recursion of $g(\cdot,\cdot)$, which
is used in the proof of \Cref{WTprob}.
\begin{lemma}\label{gaddition}
Under the condition of \Cref{WTprob}, if $\nextvar{\sigma}=u\neq \perp$, then
\[
\sum\limits_{x\in Q_u}\left(\mu_{u}^{\sigma}(x)\cdot g(\sigma_{u\gets x},T)\right)\leq g(\sigma,T)
\]
\end{lemma}
\begin{proof}
To prove this lemma, it is sufficient to show that 
\begin{equation}\label{eq-core-gaddition}
\begin{aligned}
&\sum\limits_{x\in Q_u}\left(\mu^{\sigma}_{u}(x)\prod\limits_{v\in U\setminus \vst{\sigma_{u\gets x}}}\left(1-q_v\cdot \theta_v\right)\prod\limits_{c\in E}\left(\mathbb{P}[\neg c\mid \sigma_{u\gets x}](1+\eta)^{Z(\sigma_{u\gets x},c)}\right)\right)\\
\leq & \prod\limits_{v\in U\setminus \vst{\sigma}}\left(1-q_v\cdot \theta_v\right)\prod\limits_{c\in E}\left(\mathbb{P}[\neg c\mid \sigma](1+\eta)^{Z(\sigma,c)}\right)
\end{aligned}
\end{equation}
Combining with \eqref{eq-definition-g}, we have
\begin{equation*}
\begin{aligned}
&\sum\limits_{x\in Q_u}\left(\mu_{u}^{\sigma}(x)\cdot g(\sigma_{u\gets x},T)\right)\\
(\text{by \eqref{eq-definition-g}} )\quad=&\sum\limits_{x\in Q_u}\left(\mu^{\sigma}_{u}(x)\prod\limits_{v\in U\setminus \vst{\sigma_{u\gets x}}}\left(1-q_v\cdot \theta_v\right)\prod\limits_{c\in E}\left((0.99\pprime)^{-1}\mathbb{P}[\neg c\mid \sigma_{u\gets x}](1+\eta)^{Z(\sigma_{u\gets x},c)}\right)\right)\\
(\text{by \eqref{eq-core-gaddition}} )\quad\leq& \prod\limits_{v\in U\setminus \vst{\sigma}}\left(1-q_v\cdot \theta_v\right)\prod\limits_{c\in E}\left((0.99\pprime)^{-1}\mathbb{P}[\neg c\mid \sigma](1+\eta)^{Z(\sigma,c)}\right)\\
(\text{by \eqref{eq-definition-g}} )\quad=&g(\sigma,T).
\end{aligned}
\end{equation*}
The lemma is proved.
In the following, we prove \eqref{eq-core-gaddition}.

In addition, by the constraints in $E$ are disjoint, we have $\var{c}\cap \var{c'}\neq \emptyset$ for any different $c,c'\in E$.
Thus, there exists at most one unique constraint $c_0\in E$ such that $u\in \var{c_0}$.
Let $S = E\setminus \{c_0\}$ if $u\in \var{E}$ and $S = E$ otherwise.
Thus for each $c\in S$, we have $u\not\in \var{c}$.
Then
$\mathbb{P}[\neg c\mid \sigma_{u\gets x}]=\mathbb{P}[\neg c\mid \sigma]$ and $Z(\sigma_{u\gets x},c) = Z(\sigma,c)$ for each $x\in Q_u$.
Therefore,
\begin{equation}\label{eq-sum-musigmauxtimespivcons}
\begin{aligned}
&\sum\limits_{x\in Q_u}\left(\mu^{\sigma}_{u}(x)\prod\limits_{c\in E}\left(\mathbb{P}[\neg c\mid \sigma_{u\gets x}](1+\eta)^{Z(\sigma_{u\gets x},c)}\right)\right)\\
= & \prod\limits_{c\in S}\left(\mathbb{P}[\neg c\mid \sigma_{u\gets x}](1+\eta)^{Z(\sigma_{u\gets x},c)}\right)\sum\limits_{x\in Q_u}\left(\mu^{\sigma}_{u}(x)\prod\limits_{c\in E\setminus S}\left(\mathbb{P}[\neg c\mid \sigma_{u\gets x}](1+\eta)^{Z(\sigma_{u\gets x},c)}\right)\right)\\
= &\prod\limits_{c\in S}\left(\mathbb{P}[\neg c\mid \sigma](1+\eta)^{Z(\sigma,c)}\right)\sum\limits_{x\in Q_u}\left(\mu^{\sigma}_{u}(x)\prod\limits_{c\in E\setminus S}\left(\mathbb{P}[\neg c\mid \sigma_{u\gets x}](1+\eta)^{Z(\sigma_{u\gets x},c)}\right)\right).
\end{aligned}    
\end{equation}
In addition, by \Cref{generaluniformity} and the assumption that $\mathbb{P}[\neg c\mid \sigma]\leq \pprime q$ for all $c\in \mathcal{C}$, we have for each $x\in Q_u$,
$\mu^{\sigma}_u(x)\leq q_u^{-1}(1+\eta).$
Therefore,
\begin{align*}
\sum\limits_{x\in Q_u}\left(\mu^{\sigma}_{u}(x)\cdot \mathbb{P}[\neg c_0\mid \sigma_{u\gets x}]\right)\leq (1+\eta)\cdot q_u^{-1}\sum\limits_{x\in Q_u} \mathbb{P}[\neg c_0\mid \sigma_{u\gets x}]
= (1+\eta)\cdot \mathbb{P}[\neg c_0\mid \sigma].
\end{align*}
Thus, we have 
\begin{equation*}
\begin{aligned}
    \sum_{x\in Q_u}\left(\mu^{\sigma}_u(x) \mathbb{P}[\neg c_0\mid \sigma_{u\gets x}](1+\eta)^{Z(\sigma_{u\gets x},c_0)}\right)=&\sum_{x\in Q_u}\left(\mu^{\sigma}_u(x) \mathbb{P}[\neg c_0\mid \sigma_{u\gets x}](1+\eta)^{Z(\sigma,c_0)-1}\right)\\
    \leq & \mathbb{P}[\neg c_0\mid \sigma](1+\eta)^{Z(\sigma,c_0)}.
\end{aligned} 
\end{equation*}
Therefore, if $E\setminus S = \{c_0\}$,
we have 
\begin{align}\label{eq-sum-musigmauxtimespivconstsetminuss}
    \sum\limits_{x\in Q_u}\left(\mu^{\sigma}_{u}(x)\prod\limits_{c\in E\setminus S}\left(\mathbb{P}[\neg c\mid \sigma_{u\gets x}](1+\eta)^{Z(\sigma_{u\gets x},c)}\right)\right) \leq \prod\limits_{c\in E\setminus S}\left(\mathbb{P}[\neg c\mid \sigma](1+\eta)^{Z(\sigma,c)}\right).
\end{align}
If $E\setminus S = \emptyset$,
both sides of \eqref{eq-sum-musigmauxtimespivconstsetminuss} are equal to 1 and we also have \eqref{eq-sum-musigmauxtimespivconstsetminuss}, where we assume that a product over an empty set is 1.
Because $E\setminus S$ is either $\{c_0\}$ or $\emptyset$,
we always have  \eqref{eq-sum-musigmauxtimespivconstsetminuss}.
Combining with \eqref{eq-sum-musigmauxtimespivcons},
we have 
\begin{equation}\label{eq-sum-musigmauxtimespivconst-simplify}
\begin{aligned}
\sum\limits_{x\in Q_u}\left(\mu^{\sigma}_{u}(x)\prod\limits_{c\in E}\left(\mathbb{P}[\neg c\mid \sigma_{u\gets x}](1+\eta)^{Z(\sigma_{u\gets x},c)}\right)\right)\leq \prod\limits_{c\in E}\left(\mathbb{P}[\neg c\mid \sigma](1+\eta)^{Z(\sigma,c)}\right).
\end{aligned}    
\end{equation}
Moreover, by $u=\nextvar{\sigma}$, we have $\sigma(u) = \hollowstar\neq \star$.
Thus, $u\not\in \vst{\sigma}$.
Meanwhile, by $\sigma_{u\leftarrow x}(u) = x\neq \star$, we also have $u\not\in \vst{u\leftarrow x}$ for each $x\in Q_u$.
Thus, $U\setminus \vst{\sigma} = U\setminus \vst{\sigma_{u\gets x}}$.
Combining with \eqref{eq-sum-musigmauxtimespivconst-simplify}, we have
\begin{equation*}
\begin{aligned}
&\sum\limits_{x\in Q_u}\left(\mu^{\sigma}_{u}(x)\prod\limits_{v\in U\setminus \vst{\sigma_{u\gets x}}}\left(1-q_v\cdot \theta_v\right)\prod\limits_{c\in E}\left(\mathbb{P}[\neg c\mid \sigma_{u\gets x}](1+\eta)^{Z(\sigma_{u\gets x},c)}\right)\right)\\
= & \left(\prod\limits_{v\in U\setminus \vst{\sigma}}\left(1-q_v\cdot \theta_v\right)\right) \cdot \sum\limits_{x\in Q_u}\left(\mu^{\sigma}_{u}(x)\prod\limits_{c\in E}\left(\mathbb{P}[\neg c\mid \sigma_{u\gets x}](1+\eta)^{Z(\sigma_{u\gets x},c)}\right)\right)\\
\leq & \prod\limits_{v\in U\setminus \vst{\sigma}}\left(1-q_v\cdot \theta_v\right)\prod\limits_{c\in E}\left(\mathbb{P}[\neg c\mid \sigma](1+\eta)^{Z(\sigma,c)}\right)
\end{aligned}
\end{equation*}
Then \eqref{eq-core-gaddition} and the lemma are proved.
\end{proof}

Now we can prove \Cref{WTprob}.
\begin{proof}[Proof of \Cref{WTprob}]
Let $\pth(\sigma) = \left(\sigma_0,\sigma_1,\dots,\sigma_{\ell}\right)$.
We show the lemma by a structural induction on $\pth(\sigma)$. The base case is when $\ell = 0$.
By \Cref{pathdef}, we have $\sigma_{\ell} = \sigma_0 = \sigma$. 
In this case,
$\+{E}^{\sigma}_T$ is the deterministic event $U=\vst{\sigma}\land E\subseteq \csfrozen{\sigma}$.
If $U\neq\vst{\sigma}$ or  $E\not\subseteq \csfrozen{\sigma}$, we have 
$\Pr{\+{E}^{\sigma}_T} =0$.
In addition, by \eqref{eq-definition-g}
one can verify that $g(\sigma,T)\geq 0$. Then the lemma is immediate.
Otherwise, we have $U=\vst{\sigma}\land E\subseteq \csfrozen{\sigma}$
and the event $\+{E}^{\sigma}_T$ happens.
By \eqref{eq-chidef} and $\sigma=\sigma_{\ell}$,
we have $\chi(\sigma)=\chi(\sigma_{\ell},\sigma)=\chi(\sigma,\sigma)=1$.
Thus,
\begin{equation}\label{eq-pr-esigmattimesefpath}
\Pr{ \+{E}^{\sigma}_T}\cdot\E{\chi(\sigma)\mid \+{E}^{\sigma}_T}=\Pr{ \+{E}^{\sigma}_T} = 1.
\end{equation}
Meanwhile, by $E\subseteq \csfrozen{\sigma}$ and \Cref{def:cbad}, we have $E\subseteq \csfrozen{\sigma}\subseteq \cfrozen{\sigma}$. 
Thus, each $c\in E$ is $\sigma$-frozen.
By \Cref{remark:one-sided-error-frozen-oracle},
we have
$\mathbb{P}\left[\neg c\mid \sigma\right]\geq 0.99\pprime$.
Thus, 
\begin{equation}\label{eq-099alpha-leq1}
(0.99\pprime)^{-1}\mathbb{P}[\neg c\mid \sigma](1+\eta)^{Z(\sigma,c)}\geq (0.99\pprime)^{-1}\cdot 0.99\pprime\cdot (1+\eta)^{Z(\sigma,c)} \geq 1.
\end{equation}
In addition, by \eqref{eq-definition-g} and $U=\vst{\sigma}$,
we have 
\begin{align*}
     g(\sigma,T)&=\prod\limits_{c\in E}\left((0.99\pprime)^{-1}\mathbb{P}[\neg c\mid \sigma](1+\eta)^{Z(\sigma,c)}\right).
\end{align*}
Combining with \eqref{eq-pr-esigmattimesefpath} and \eqref{eq-099alpha-leq1}, 
we have 
\begin{align*}
     g(\sigma,T)\geq 1\geq \Pr{ \+{E}^{\sigma}_T}\cdot\E{\chi(\sigma)\mid \+{E}^{\sigma}_T}.
\end{align*}
The base case is proved.

For the induction step, we assume that $\ell(\sigma)\geq 1$, which by \Cref{rp-1} of \Cref{pathdef}, says that $\nextvar{\sigma}=u\neq \perp$ for some $u\in V$. According to \Cref{rp-2} of \Cref{pathdef},
we have 
\begin{align*}
    \forall x\in \qus{u},\quad \Pr{\sigma_{1} = \sigma_{u\gets x}}&=\induceddist{\sigma}{u}(x).
\end{align*}
Thus, by the law of total probability, we have
\begin{equation}
\begin{aligned}\label{WTprob-eq4}
&\Pr{ \+{E}^{\sigma}_T}\cdot\E{\chi(\sigma)\mid \+{E}^{\sigma}_T}\\
= &\sum\limits_{x\in \qus{u}}\left(\Pr{\sigma_1=\sigma_{u\gets x}}\cdot \Pr{ \+{E}^{\sigma}_T\mid \sigma_1=\sigma_{u\gets x}}\cdot \E{\chi(\sigma)\mid \+{E}^{\sigma}_T\land \left(\sigma_1=\sigma_{u\gets x}\right)}\right)\\
=&\sum\limits_{x\in \qus{u}}\left(\induceddist{\sigma}{u}(x)\cdot\Pr{ \+{E}^{\sigma}_T\mid \sigma_1=\sigma_{u\gets x}}\cdot \E{\chi(\sigma)\mid \+{E}^{\sigma}_T\land \left(\sigma_1=\sigma_{u\gets x}\right)}\right).
\end{aligned}
\end{equation}
Moreover, by \eqref{eq-chidef} we have
\begin{equation}
\begin{aligned}\label{eq-ehpath-sigma}
   \E{\chi(\sigma)\mid   \+{E}^{\sigma}_T\land\left(\sigma_1=\sigma_{u\gets x}\right)} &=
   \E{\chi(\sigma_{\ell},\sigma_{0})\mid   \+{E}^{\sigma}_T\land\left(\sigma_1=\sigma_{u\gets x}\right)}
   \\&= \left(2-q_u\cdot \theta_u\right)\E{ \chi(\sigma_{\ell},\sigma_{1})\mid \+{E}^{\sigma}_T \land \left(\sigma_1=\sigma_{u\gets x}\right)}.
\end{aligned}
\end{equation}
In addition, for each $x\in \qus{u}$, given $\sigma_1 = \tau \triangleq \sigma_{u\gets x}$ , 
we have the subsequence  $(\sigma_1,\sigma_2,\dots,\sigma_{\ell})$  is identically  distributed as $\pth(\tau)$ by the Markov property of the \pth{} process.
Thus, we have $\sigma_{\ell}$ is identically distributed as $\tau_{\ell(\tau)}$.
Combining with the definition of ${\+E}^{\sigma}_T$ in \Cref{def:cbad}, we have
\begin{equation*}
    \E{ \chi(\sigma_{\ell},\sigma_{1})\mid  \+{E}^{\sigma}_T \land \left(\sigma_1=\tau \right)}=\E{ \chi(\tau_{\ell(\tau)},\tau)\mid  \+{E}^{\tau}_T} = \E{ \chi(\tau)\mid  \+{E}^{\tau}_T}.
\end{equation*}
Combining with \eqref{eq-ehpath-sigma}, we have 
\begin{equation*}
\begin{aligned}
   \E{\chi(\sigma)\mid   \+{E}^{\sigma}_T \land (\sigma_1=\sigma_{u\gets x}) }
   =&\left(2-q_u\cdot \theta_u\right)\E{ \chi(\sigma_{\ell},\sigma_{1})\mid \+{E}^{\sigma}_T \land (\sigma_1=\sigma_{u\gets x})}\\
   =&\left(2-q_u\cdot \theta_u\right)\E{ \chi(\sigma_{\ell},\sigma_{1})\mid \+{E}^{\sigma}_T \land (\sigma_1=\tau)}\\
    =&\left(2-q_u\cdot \theta_u\right)\E{ \chi(\tau)\mid  \+{E}^{\tau}_T}\\
    =&\left(2-q_u\cdot \theta_u\right)\E{ \chi(\sigma_{u\gets x})\mid  \+{E}^{\sigma_{u\gets x}}_T}.
\end{aligned}
\end{equation*}
Combining with \eqref{WTprob-eq4}, we have
\begin{equation*}
\begin{aligned}
&\Pr{ \+{E}^{\sigma}_T}\cdot\E{\chi(\sigma)\mid \+{E}^{\sigma}_T}\\
=&\sum\limits_{x\in \qus{u}}\left(\induceddist{\sigma}{u}(x)\cdot\Pr{ \+{E}^{\sigma}_T\mid \sigma_1=\sigma_{u\gets x}}\cdot \E{\chi(\sigma)\mid \+{E}^{\sigma}_T\land \left(\sigma_1=\sigma_{u\gets x}\right)}\right)\\
=&\left(2-q_u\cdot \theta_u\right)\sum\limits_{x\in \qus{u}}\left(\induceddist{\sigma}{u}(x)\cdot\Pr{ \+{E}^{\sigma}_T\mid \sigma_1=\sigma_{u\gets x}}\cdot \E{\chi(\sigma_{u\leftarrow x})\mid \+{E}^{\sigma_{u\gets x}}_T}\right).
\end{aligned}
\end{equation*}
In addition, for each $x\in \qus{u}$, 
recall that $\sigma_{\ell}$ is identically distributed as $\sigma_{u\gets x}$ given $\sigma_1 = \sigma_{u\gets x}$.
Combining with \Cref{def:cbad},
we have 
$\Pr{ \+{E}^{\sigma}_T\mid \sigma_1=\sigma_{u\gets x}} = \Pr{ \+{E}^{\sigma_{u\leftarrow x}}_T}$.
Therefore,
\begin{equation}
\begin{aligned}\label{eq-esigmatehpath-twopart}
\Pr{\+{E}^{\sigma}_T}\E{\chi(\sigma)\!\mid \!\+{E}^{\sigma}_T}
=\left(2-q_u\cdot \theta_u\right)\!\sum\limits_{x\in \qus{u}}\left(\induceddist{\sigma}{u}(x)\Pr{\+{E}^{\sigma_{u\leftarrow x}}_T} \E{\chi(\sigma_{u\leftarrow x})\!\mid\! \+{E}^{\sigma_{u\gets x}}_T}\right).
\end{aligned}
\end{equation}


We then show the induction step for two cases respectively, namely the case when $u\in U$ and the case when $u\notin U$.
At first we assume $u\in U$.
Given $x\in Q_u$ and $\tau = \sigma_{u\leftarrow x}$, by $\tau(u)=x$,
we also have $\tau_{\ell(\tau)}(u)=x\neq \star$.
Thus $u\not\in \vst{\tau_{\ell(\tau)}}$.
Combining with $u\in U$,
we have $U \neq \vst{\tau_{\ell(\tau)}}$.
Combining with \Cref{def:cbad},
we have $\+{E}^{\tau}_T$ does not happen.
In summary, for each $x\in Q_u$, $\+{E}^{\sigma_{u\leftarrow x}}_T$ does not happen 
and $\Pr{\+{E}^{\sigma_{u\leftarrow x}}_T} = 0$.
Combining with \eqref{eq-esigmatehpath-twopart} and \eqref{eq-def-induceddist}, we have 
\begin{equation}\label{eq-pesigmattimesehpath-uinvvart}
\begin{aligned}
    \Pr{ \+{E}^{\sigma}_T}\cdot \E{\chi(\sigma)\mid \+{E}^{\sigma}_T}&=
    \left(2-q_u\cdot \theta_u\right)\cdot \induceddist{\sigma}{u}(\star) \cdot \Pr{ \+{E}^{\sigma_{u\leftarrow \star}}_T}\cdot \E{\chi(\sigma_{u\gets \star})\mid  \+{E}^{\sigma_{u\gets \star}}_T}\\
    & =  
    (1-q_u\cdot \theta_u)\cdot \Pr{ \+{E}^{\sigma_{u\leftarrow \star}}_T}\cdot \E{\chi(\sigma_{u\gets \star})\mid  \+{E}^{\sigma_{u\gets \star}}_T}.
\end{aligned}
\end{equation}
In addition, by $\sigma\in\+{Q}^*$ is a partial assignment satisfying $\mathbb{P}[\neg c\mid \sigma]\leq \pprime q$ for all $c\in \mathcal{C}$,
one can also verify
$\mathbb{P}[\neg c\mid \sigma_{u\leftarrow x}]\leq \pprime q$ for all $c\in \mathcal{C}$ and $x\in \qus{u}$
by a similar argument as  \Cref{lemma:general-invariant}.
Thus by the induction hypothesis, for each $x\in \qus{u}$ we have
\begin{equation}\label{eq-pesigmauxttimesehpath-induct}
\begin{aligned}
\Pr{ \+{E}^{\sigma_{u\gets x}}_T}\cdot \E{\chi(\sigma_{u\gets x})\mid \+{E}^{\sigma_{u\gets x}}_T}
\leq g(\sigma_{u\gets x},T).
\end{aligned}
\end{equation}
Combining with \eqref{eq-pesigmattimesehpath-uinvvart} and \eqref{eq-definition-g},
we have
\begin{equation*}
\begin{aligned}
&\Pr{ \+{E}^{\sigma}_T}\cdot \E{\chi(\sigma)\mid \+{E}^{\sigma}_T}\leq(1-q_u\cdot \theta_u)\cdot g(\sigma_{u\gets \star},T)\\
  = &(1-q_u\cdot \theta_u) \prod\limits_{v\in U\setminus \vst{\sigma_{u\gets \star}}}\left(1-q_v\cdot \theta_v\right)\prod\limits_{c\in E}\left((0.99\pprime)^{-1}\mathbb{P}[\neg c\mid \sigma_{u\gets \star}](1+\eta)^{Z(\sigma_{u\gets \star},c)}\right)
\end{aligned}
\end{equation*}
In addition, by $u=\nextvar{\sigma}$ and \Cref{definition:boundary-variables}, we have $\sigma(u) = \hollowstar\neq \star$.
Thus, $u\not\in \vst{\sigma}$.
Meanwhile, by $\sigma_{u\leftarrow \star}(u) = \star$, we have $u\in \vst{u\leftarrow \star}$.
Thus, $\vst{\sigma_{u\leftarrow \star}} = \vst{\sigma}\biguplus \{u\}$.
Combining with $u\in U$,
we have $U\setminus \vst{\sigma} = \left(U\setminus \vst{\sigma_{u\gets \star}}\right)\biguplus \{u\}$.
Therefore, 
$$ (1-q_u\cdot \theta_u) \prod\limits_{v\in U\setminus \vst{\sigma_{u\gets \star}}}\left(1-q_v\cdot \theta_v\right) = \prod\limits_{v\in U\setminus \vst{\sigma}} \left(1-q_v\cdot \theta_v\right).$$
Thus, we have 
\begin{equation*}
\begin{aligned}
&\Pr{ \+{E}^{\sigma}_T}\cdot \E{\chi(\sigma)\mid \+{E}^{\sigma}_T}\leq\prod\limits_{v\in U\setminus \vst{\sigma}}\left(1-q_v\cdot \theta_v\right)\prod\limits_{c\in E}\left((0.99\pprime)^{-1}\mathbb{P}[\neg c\mid \sigma_{u\gets \star}](1+\eta)^{Z(\sigma_{u\gets \star},c)}\right).
\end{aligned}
\end{equation*}
By $\mathbb{P}[\neg c\mid \sigma_{u\gets \star}] = \mathbb{P}[\neg c\mid \sigma]$ and
$Z(\sigma_{u\gets \star},c) = Z(\sigma,c)$ for each $\sigma\in \qs$ and $c\in C$,
we have 
\begin{equation*}
\begin{aligned}
\Pr{ \+{E}^{\sigma}_T}\cdot \E{\chi(\sigma)\mid \+{E}^{\sigma}_T} \leq  \prod\limits_{v\in U\setminus \vst{\sigma}}\left(1-q_v\cdot \theta_v\right)\prod\limits_{c\in E}\left((0.99\pprime)^{-1}\mathbb{P}[\neg c\mid \sigma](1+\eta)^{Z(\sigma,c)}\right)=g(\sigma,T).
\end{aligned}
\end{equation*}
Then 
\eqref{eq-lemma-WTprob} is immediate.
This finishes the induction step for the case when $u\in U$.

In the following, we assume $u\not\in U$. Given $\tau = \sigma_{u\leftarrow \star}$,
by $\tau(u)=\star$, we also have $\tau_{\ell(\tau)}(u)=\star$.
Thus $u \in \vst{\tau_{\ell(\tau)}}$.
Combining with $u\not\in U$,
we have $U \neq \vst{\tau_{\ell(\tau)}}$.
Combining with \Cref{def:cbad},
we have $\+{E}^{\tau}_T = \+{E}^{\sigma_{u\leftarrow \star}}_T$ does not happen
and $\Pr{\+{E}^{\sigma_{u\leftarrow \star}}_T} = 0$.
Combining with \eqref{eq-esigmatehpath-twopart} and \eqref{eq-def-induceddist}, we have 
\begin{equation*}
\begin{aligned}
    \Pr{ \+{E}^{\sigma}_T}\cdot \E{\chi(\sigma)\mid \+{E}^{\sigma}_T}&=\left(2-q_u\cdot \theta_u\right)\sum\limits_{x\in Q_u}\left(\induceddist{\sigma}{u}(x)\Pr{\+{E}^{\sigma_{u\leftarrow x}}_T} \E{\chi(\sigma_{u\leftarrow x})\!\mid\! \+{E}^{\sigma_{u\gets x}}_T}\right)\\
    &=\sum\limits_{x\in Q_u}\left(\mu_{u}^{\sigma}(x)\cdot\Pr{\+{E}^{\sigma_{u\leftarrow x}}_T}\cdot \E{\chi(\sigma_{u\gets x})\mid \+{E}^{\sigma_{u\gets x}}_T}\right).
\end{aligned}
\end{equation*}
Combining with \eqref{eq-pesigmauxttimesehpath-induct} and \eqref{eq-definition-g}, we have
\begin{equation*}
\begin{aligned}
&\Pr{ \+{E}^{\sigma}_T}\cdot \E{\chi(\sigma)\mid \+{E}^{\sigma}_T}
 \leq \sum\limits_{x\in Q_u}\left(\mu_{u}^{\sigma}(x)\cdot g(\sigma_{u\gets x},T)\right)\leq g(\sigma,T),
\end{aligned}
\end{equation*}
where the last inequality is by \Cref{gaddition}.
This finishes the induction step for the  case when $u\notin U$.
The lemma is proved.
\end{proof}

Now we can prove \Cref{WTprob2} with \Cref{WTprob}.
\begin{proof}[Proof of \Cref{WTprob2}]
To prove the lemma, it is sufficient to prove \eqref{eq-lemma-WTprob1},
because \eqref{eq-lemma-WTprob2} is a special case of \eqref{eq-lemma-WTprob1} when $i = 0$.
In the following, we show \eqref{eq-lemma-WTprob1} by induction on $i$. The base case is when $i=t-1$. Conditioning on $X^i=X^{t-1}=\sigma$, by \eqref{eq:simulate-path} we have either $X^t_0 = \sigma$ or $X^t_0 = \sigma_{v_t\gets \star}$. By $\Pr{X^i=\sigma}>0$ we have $\sigma(v)\neq \star$ for any $v\in V$,
because by \eqref{eq:simulate-path}, 
$X^i$ is generated from $X^0 =\hollowstar^V$ with \Cref{Alg:main-eq}
and no variable in \Cref{Alg:main-eq} is set as $\star$.
If $X^t_0=\sigma$, 
by \eqref{eq:simulate-path} we have 
$X^t_{\ell} = \pth(X^t_0) = \pth(\sigma)$.
In addition, by \Cref{pathdef} we have $\pth(\sigma)=\sigma$.
Thus, we have $X^t_{\ell} =  \pth(\sigma) = \sigma$, 
$\vst{X^t_{\ell}}=\vst{\sigma}$,
and $\csfrozen{X^t_{\ell}}=\csfrozen{\sigma}$.
In addition, by \Cref{def:cbad} and that $\sigma(v)\neq \star$ for any $v\in V$,
we have $\vst{\sigma}=\emptyset$ and 
$\csfrozen{\sigma}=\emptyset$.
Thus, 
$\vst{X^t_{\ell}}=\emptyset$ and $\csfrozen{X^t_{\ell}}=\emptyset$. 
Combining with the definition of $\+{E}^{t}_T$ in \Cref{def:cbad},
we have $\+{E}^{t}_T$ does not happen.
Thus, we have 
\begin{align}\label{eq-exclude-ett-x0tsigma}
\Pr{ \+{E}^{t}_T\mid X^t_0=\sigma} = \Pr{  X^t_0=\sigma \mid \+{E}^{t}_T} = 0.
\end{align}
Therefore, we have
\begin{align*}
\quad &\Pr{ \+{E}^{t}_T\mid X^{t-1}=\sigma}\\
(\text{by the law of total probability})\quad = &\Pr{X^{t}_0 = \sigma\mid X^{t-1}=\sigma} \Pr{\+{E}^{t}_T\mid \left(X^{t-1}=\sigma\right) \land \left(X^{t}_0 = \sigma\right)} \\
&+ \Pr{X^{t}_0 = \sigma_{v_t\gets \star}\mid X^{t-1}=\sigma} \Pr{\+{E}^{t}_T\mid \left(X^{t-1}=\sigma\right)\land \left(X^{t}_0 = \sigma_{v_t\gets \star}\right)} \\
(\text{by } \eqref{eq-exclude-ett-x0tsigma})\quad= &\Pr{X^{t}_0 = \sigma_{v_t\gets \star}\mid X^{t-1}=\sigma} \Pr{\+{E}^{t}_T\mid \left(X^{t-1}=\sigma\right)\land\left(X^{t}_0 = \sigma_{v_t\gets \star}\right)} \\
(\text{by } \eqref{eq:simulate-path})\quad \leq & \Pr{r_t = 1}\Pr{\+{E}^{t}_T\mid \left(X^{t-1}=\sigma\right)\land\left(X^{t}_0 = \sigma_{v_t\gets \star}\right)}\\
(\text{by \Cref{def:cbad}})\quad = & \Pr{r_t = 1}\Pr{\+{E}^{\sigma_{v_t\gets \star}}_T}
\end{align*}
Similarly, by the law of total expectation we have
\begin{align*}
    &\E{\chi(X^t_0)\mid \+{E}^{t}_T\land \left(X^{t-1}=\sigma\right)}\\
    = &\Pr{X^{t}_0=\sigma \mid \+{E}^{t}_T \land \left(X^{t-1}=\sigma\right)}\E{\chi(X^t_0)\mid \+{E}^{t}_T\land \left(X^{t}_0=\sigma\right) \land \left(X^{t-1}=\sigma\right)}\\
    & + \Pr{X^{t}_0=\sigma_{v_t\gets \star} \mid \+{E}^{t}_T \land \left(X^{t-1}=\sigma\right)}\E{\chi(X^t_0)\mid \+{E}^{t}_T\land \left(X^{t}_0=\sigma_{v_t\gets \star}\right)\land \left(X^{t-1}=\sigma\right)}\\
(\text{by } \eqref{eq-exclude-ett-x0tsigma})\quad  = &\Pr{X^{t}_0=\sigma_{v_t\gets \star} \mid \+{E}^{t}_T \land \left(X^{t-1}=\sigma\right)}\E{\chi(X^t_0)\mid \+{E}^{t}_T\land \left(X^{t}_0=\sigma_{v_t\gets \star}\right) \land \left(X^{t-1}=\sigma\right)}\\
(\text{by \Cref{def:cbad}})\quad  = &\Pr{X^{t}_0=\sigma_{v_t\gets \star} \mid \+{E}^{t}_T \land \left(X^{t-1}=\sigma\right)}\E{\chi(X^t_0)\mid \+{E}^{\sigma_{v_t\gets \star}}_T}\\
    \leq &\E{\chi(X^t_0)\mid \+{E}^{\sigma_{v_t\gets \star}}_T}
\end{align*}
Combining the above two inequalities, we have
\begin{align*}
&\Pr{ \+{E}^{t}_T\mid X^i=\sigma}\cdot\E{\chi(X^t_0)\mid \+{E}^{t}_T\land X^i=\sigma}\\
\leq &\Pr{r_t = 1}\Pr{\+{E}^{\sigma_{v_t\gets \star}}_T}\E{\chi(X^t_0)\mid \+{E}^{\sigma_{v_t\gets \star}}_T}\\
(\text{by } \eqref{eq:simulate-path})\quad    = & (1-q_{v_t}\cdot\theta_{v_t})\Pr{ \+{E}^{\sigma_{v_t\gets \star}}_T}\E{\chi(\sigma_{v_t\gets \star})\mid \+{E}^{\sigma_{v_t\gets \star}}_T}\\
(\text{by \Cref{WTprob}})\quad    \leq & (1-q_{v_t}\cdot\theta_{v_t})g(\sigma_{v_t\gets \star},T)\\
(\text{by \eqref{eq-definition-g}})\quad \leq & g(\sigma,T).
\end{align*}
The base case is proved.

For the induction step, we assume $i<t-1$ and prove \eqref{eq-lemma-WTprob1} based on the hypothesis on $i+1$. 
Given $X^i=\sigma$, if $v_{i+1}$ is $\sigma$-fixed, 
by \eqref{eq:simulate-path} and \Cref{Alg:main-eq},
we have $X^{i+1}=X^i=\sigma$.
Then \eqref{eq-lemma-WTprob1} holds by the induction hypothesis. 
Otherwise, $v_{i+1}$ is not $\sigma$-fixed. Let $u=v_{i+1}$.
According to \Cref{Alg:main-eq} we have
\begin{equation}\label{eq-simulate}
     \forall x\in Q_{u},\quad \Pr{X^{i+1} = \sigma_{u\gets x}\mid X^i=\sigma}=\mu^{\sigma}_{u}(x).
\end{equation}
In addition, 
recall that \Cref{Alg:main-eq} generates the prefix $({X}^0,{X}^1,\dots,{X}^t)$ of the random partial assignments  $X^0,X^1,\dots,X^n$ maintained in \Cref{Alg:main} defined in \Cref{def-pas-main}.
Combining with \Cref{lemma:invariant-p-prime-q-bound} and  
we have $\mathbb{P}[\neg c\mid X^i]\leq \pprime q$ for all $c\in \mathcal{C}$.
Combining with $\Pr{X^i=\sigma}>0$,
we have $\mathbb{P}[\neg c\mid \sigma]\leq \pprime q$ for all $c\in \mathcal{C}$.
Combining with \Cref{localuniformitycor},
we have $\mu^{\sigma}_{u}(x)>0$ for each 
$x\in Q_{u}$.
Thus, we have
\begin{equation}\label{eq-prxi1eqsigmamux}
\begin{aligned}
&\quad\Pr{X^{i+1}=\sigma_{u\leftarrow x}} \\
(\text{by the chain rule})\quad &=\Pr{X^i=\sigma}\Pr{X^{i+1} = \sigma_{u\gets x}\mid X^i=\sigma}\\
(\text{by \eqref{eq-simulate}})\quad &=\Pr{X^i=\sigma}\mu^{\sigma}_{u}(x) \\
(\text{by $\Pr{X^i=\sigma}>0$, $\mu^{\sigma}_{u}(x)>0$})\quad&>0.
\end{aligned}
\end{equation}
Thus by the induction hypothesis, for each $x\in Q_{u}$ we have
\begin{equation}\label{eq-pesigmauxttimesehpath-induct2}
\begin{aligned}
\Pr{ \+{E}^{t}_T\mid X^{i+1}=\sigma_{u\gets x}}\cdot \E{\chi(X^t_0)\mid \+{E}^{t}_T\land X^{i+1}=\sigma_{u\gets x}}
\leq g(\sigma_{u\gets x},T).
\end{aligned}
\end{equation}
Combining with \eqref{eq-simulate} we have
\begin{equation*}
\begin{aligned}
&\Pr{ \+{E}^{t}_T\mid X^i=\sigma}\cdot\E{\chi(X^t_0)\mid \+{E}^{t}_T\land X^i=\sigma}\\
(\text{by \eqref{eq-simulate} and \eqref{eq-pesigmauxttimesehpath-induct2}})\quad \leq &\sum\limits_{x\in Q_u}\left(\mu_{u}^{\sigma}(x)\cdot g(\sigma_{u\gets x},T)\right)\\
 (\text{by \Cref{gaddition}}) \quad\leq  & g(\sigma,T),
\end{aligned}
\end{equation*}
which finishes the induction step.
Then \eqref{eq-lemma-WTprob1} and the lemma are proved.
\end{proof}

Now we can prove \Cref{badtree}.

\begin{proof}[Proof of \Cref{badtree}]
For each $c\in E$, 
we have $Z(\hollowstar^V,c)=\abs{\var{c}\setminus \Lambda(\hollowstar^V) } \leq k$. 
Combining with \eqref{eq-definition-g} we have
\begin{equation}
\begin{aligned}\label{eq-badtree-gstarvt-upperbound}
g(\hollowstar^V,T)=&\prod\limits_{v\in V}(1-q_v\cdot \theta_v)\prod\limits_{c\in E}\left((0.99\pprime)^{-1}\mathbb{P}[\neg c\mid \hollowstar^V](1+\eta)^{Z(\hollowstar^V,c)}\right)\\
\leq &\prod\limits_{v\in V}(1-q_v\cdot \theta_v)\prod\limits_{c\in E}\left((0.99\pprime)^{-1}p(1+\eta)^{k}\right).
\end{aligned} 
\end{equation}
In addition, by \eqref{eq:parameter-p-prime} and \eqref{eq:parameter-theta}, 
we have $\eta\leq (2k\Delta)^{-1}$ and
\begin{align}\label{eq-badtree-v-upperbound}
 \forall v\in V,\quad   1-q_v\cdot \theta_v\leq 1-q\theta\leq (8\mathrm{e}k\Delta)^{-1}.
\end{align}
By $8\mathrm{e}p\Delta^3\leq 0.99\alpha$ and $\eta\leq (2k\Delta)^{-1}$, we have
\begin{align}\label{eq-badtree-e-upperbound}
    (0.99\pprime)^{-1}p(1+\eta)^{k}\leq (4\mathrm{e}\Delta^3)^{-1}.
\end{align} 
Combining \eqref{eq-badtree-gstarvt-upperbound} with \eqref{eq-badtree-v-upperbound} and \eqref{eq-badtree-e-upperbound}, we have
\begin{align*}
    g(\hollowstar^V,T)    \leq  \left(8\mathrm{e}k\Delta\right)^{-\abs{ U }}\cdot \left(4\mathrm{e}\Delta^3\right)^{-\abs{ E }}.
\end{align*}
Combining with \Cref{WTprob2},
we have 
\begin{align*}
\Pr{ \+{E}^{t}_T}\cdot\E{\chi(X^t_0)\mid \+{E}^{t}_T}\leq g(\hollowstar^V,T) \leq \left(8\mathrm{e}k\Delta\right)^{-\abs{ U }}\cdot \left(4\mathrm{e}\Delta^3\right)^{-\abs{ E }}.
\end{align*}
The lemma is proved.
\end{proof}

We finish \Cref{sec:raw-tail-bound} by proving \Cref{badtree2}. Recall that \eqref{eq-badtree-e-upperbound} holds by \eqref{eq:parameter-p-prime},  \eqref{eq:parameter-theta} and $8\mathrm{e}p\Delta^3\leq 0.99\alpha$.
Thus \Cref{badtree2} is immediate by the following lemma, which is an analogy of \Cref{WTprob}.

\begin{lemma}
Recall the definition of $g(\cdot,\cdot)$ in \eqref{eq-definition-g}. Let $(X^0,X^1,\ldots, X^n)=\simulator(n)$. For any $0\leq i\leq n$, any partial assignment $\sigma\in \qs$ where $\Pr{X^i=\sigma}>0$ and 
any set of disjoint constraints
$T\subseteq \+{C}$, we have
\begin{equation}\label{eq-lemma-badtree-2}
    \Pr{T\subseteq \cfrozen{X^n}\mid X^i=\sigma} \leq g(\sigma,T).
\end{equation}
Specifically, 
\begin{equation}\label{eq-lemma-badtree-i-equal-0}
\Pr{T\subseteq \cfrozen{X^n}}\leq g(\hollowstar^V,T)=\prod\limits_{c\in T}\left((0.99\pprime)^{-1}(1+\eta)^k\right).
\end{equation}
 \end{lemma}
 
 \begin{proof}
 To prove the lemma, it is sufficient to prove \eqref{eq-lemma-badtree-2},
because \eqref{eq-lemma-badtree-i-equal-0} is a special case of \eqref{eq-lemma-badtree-2} when $i = 0$.
 We show \eqref{eq-lemma-badtree-2} by induction on $i$. The base case is when $i=n$. 
 For each $\sigma$,
 conditioning on $X^i=\sigma$, we have $X^n = X^i = \sigma$.
 Thus, if $T\not\subseteq \cfrozen{\sigma}$, we have $\Pr{T\subseteq \cfrozen{X^n}\mid X^i=\sigma}=\Pr{T\subseteq \cfrozen{\sigma}\mid X^i=\sigma}=0$ and \eqref{eq-lemma-badtree-2} is immediate by the non-negativity of $g(\cdot,\cdot)$. Otherwise, we have $T\subseteq \cfrozen{\sigma}$. According to \Cref{remark:one-sided-error-frozen-oracle},
we have $c$ is $\sigma$-frozen only if 
$\mathbb{P}\left[\neg c\mid \sigma\right]\geq 0.99\pprime$.
Therefore, we have
\begin{align*}
    g(\sigma,T)&=\prod\limits_{c\in T}\left((0.99\pprime)^{-1}\mathbb{P}[\neg c\mid \sigma](1+\eta)^{Z(\sigma,c)}\right)
    \geq \prod\limits_{c\in T}\left((0.99\pprime)^{-1}\mathbb{P}[\neg c\mid \sigma]\right)\\
    &\geq \prod\limits_{c\in T}\left((0.99\pprime)^{-1}\cdot 0.99\pprime\right)
    \geq 1=\Pr{T\subseteq \cfrozen{X^n}\mid X^i=\sigma}.
\end{align*}
Thus, we also have \eqref{eq-lemma-badtree-2} and the base case is proved.

For the induction step, we assume $i<n$ and prove \eqref{eq-lemma-badtree-2} based on the hypothesis on $i+1$. 
Given $X^i=\sigma$, if $v_{i+1}$ is $\sigma$-fixed, 
by \eqref{eq:simulate-path} and \Cref{Alg:main-eq},
we have $X^{i+1}=X^i=\sigma$.
Then \eqref{eq-lemma-badtree-2} holds by the induction hypothesis. Otherwise $v_{i+1}$ is not $\sigma$-fixed. Let $u=v_{i+1}$.
According to \Cref{Alg:main-eq} we have
\begin{equation}\label{eq-simulate-2}
     \forall x\in Q_{u},\quad \Pr{X^{i+1} = \sigma_{u\gets x}\mid X^i=\sigma}=\mu^{\sigma}_{u}(x).
\end{equation}
Following the proof of \eqref{eq-prxi1eqsigmamux} in \Cref{WTprob2},
it is easy to verify that 
$\Pr{X^{i+1}=\sigma_{u\leftarrow x}}>0$ for each $x\in Q_{u}$.
Thus by the induction hypothesis, for each $x\in Q_{u}$ we have
\begin{equation}\label{eq-lemma-badtree-4}
\begin{aligned}
\Pr{ T\subseteq \cfrozen{X^n}\mid X^{i+1}=\sigma_{u\gets x}}
\leq g(\sigma_{u\gets x},T).
\end{aligned}
\end{equation}
Thus, we have
\begin{align*}
&\Pr{ T\subseteq \cfrozen{X^n}\mid X^i=\sigma}\\
(\text{by \eqref{eq-simulate-2} and \eqref{eq-lemma-badtree-4}})\quad \leq &\sum\limits_{x\in Q_u}\left(\mu_{u}^{\sigma}(x)\cdot g(\sigma_{u\gets x},T)\right)\\
 (\text{by \Cref{gaddition}})\quad \leq & g(\sigma,T),
\end{align*}
which finishes the induction step. Then \eqref{eq-lemma-badtree-2} and the lemma are proved.
 \end{proof}

\section{Conclusion and Open Problems}
We give an algorithm for sampling uniform solutions to general constraint satisfaction problems (CSPs) in a local lemma regime. 
The algorithm runs in an expected  near-linear time in the number of variables and polynomial in other local parameters, including: domain size $q$, width $k$ and degree $\Delta$.

This gives, for the first time, a near-linear time sampling algorithm for general CSPs with constant $q,k,\Delta$, in a local lemma regime;
and this also gives, for the first time, a polynomial-time sampling algorithm for general CSPs in a local lemma regime without assuming any degree or width bound.

A crucial step of our sampling algorithm,
is a marginal sampler 
that can draw values of a variable according to the correct marginal distribution. 
Within a local lemma regime, this marginal sampler is a local algorithm whose cost is independent of the size of the CSP, and is polynomial in the local parameters $q,k$ and $\Delta$. 
This marginal sampler proves a thought-provoking point: 
\emph{within a local lemma regime, 
a locally defined sampling or inference problem can be solved at a local cost}. 

There are several open problems:
\begin{itemize}
    \item 
An open question is to improve the current LLL condition $p\Delta^5\lesssim 1$ for sampling general CSPs  closer the lower bound $p\Delta^{2}\gtrsim 1$.
We also believe that removing the extra $q, k$ factors in the current LLL condition may help us better understand the nature of sampling LLL.
    \item 
Another fundamental question is to generalize the current bound for CSPs to a general sampling Lov\'{a}sz local lemma with non-uniform distributions and/or asymmetric criteria.
\item
The recursive marginal sampler of Anand and Jerrum \cite{anand2021perfect} is a refreshingly novel idea for sampling. Here we see that its can solve an otherwise difficult to solve problem. It would be exciting to see what more this new idea can bring to the study of sampling LLL.  
\item
Despite the current technical barrier, Markov chain based algorithms have several advantages, such as their efficient parallelization~\cite{liu2021simple}. Therefore, it is still very worthwhile to have Markov chain based algorithms for sampling general CSPs.
For this to work, we may have to develop a way for dynamic projection of solution space, which may be of independent interest. 
\end{itemize}
    
\appendix

\ifarxiv{
\section*{Acknowledgement}
We thank Weiming Feng and Kewen Wu for helpful discussions.
Kun He wants to thank Xiaoming Sun for his support in doing this work.
Yitong Yin wants to thank Vishesh Jain for pointing to the notion of robust CSPs.
}{}

\bibliographystyle{alpha}
\bibliography{references} 

\clearpage

\section{A Bernoulli Factory for Margin Overflow}\label{sec:bernoulli-factory}
We present a generic solution to the following problem. 
Let $(\Phi, \sigma ,v)$ be the input to \Cref{Alg:Recur}, where  $\Phi=(V,\+{Q},\+{C})$ is a CSP formula, $\sigma\in \+{Q}^*$ is a feasible partial assignment, and $v\in V$ is a variable.
Assume \Cref{inputcondition-recur} for the $(\Phi, \sigma ,v)$.
The marginal distribution $\mu_v^\sigma$ over domain $Q_v$ is well-defined,
and by \Cref{localuniformitycor}, the following is satisfied for the parameters $\theta_v>0$ and $\zeta>0$ fixed as in \eqref{eq:parameter-theta}:
\begin{align}\label{eq:promise-margin-lower-bound-app}
\min\limits_{x\in Q_v}\mu^{\sigma}_v(x)\ge  \theta_v+\zeta.
\end{align}
Therefore, the distribution $\+{D}$ in \eqref{eq:definition-margin-overflow} is well-defined.
We reiterate its definition here:
\begin{align}\label{eq:definition-margin-overflow-app}
\forall x\in Q_v,\qquad
\+{D}(x)\triangleq\frac{\mu_v^{\sigma}(x)-\theta_v}{1-q_v\cdot \theta_v},\quad \text{ where }q_v\triangleq |Q_v|.
\end{align}
Our goal is to sample from this distribution by accessing an oracle for drawing independent samples from $\mu_v^{\sigma}$.
Such an oracle for $\mu_v^{\sigma}$ is realized by $\rejsamp{}(\Phi,\sigma,\{v\})$ defined in \Cref{Alg:rej}.

Such a problem of simulating a new coin by making black-box accesses to an old coin, while the distribution of the new coin is a function of the old, is known as the Bernoulli factory problem \cite{vonNeumann1951}.

\subsection{Construction and correctness of the Bernoulli factory}\label{sec:bernoulli-factory-construction}
For $\xi\in[0,1]$, we denote by $\+{O}_\xi$ a coin with probability of heads $\xi$.
Formally, $\+{O}_\xi$ is an oracle that, upon each call, independently returns 1 with probability $\xi$ and 0 with probability $1-\xi$.

We write $\nu=\mu_{v}^{\sigma}$ for short. 
We construct the following two types of basic oracles:
\begin{itemize}
\item
for each $x\in Q_v$, an $\+{O}_{\nu(x)}$ is constructed as $\+{O}_{\nu(x)}=\one{\rejsamp{}(\Phi,\sigma,\{v\})=x}$;
\item an $\+{O}_{\theta_v}$ is constructed as $\+{O}_{\theta_v}=\one{r<\theta_v}$ for $r\in[0,1)$ chosen uniformly at random.
\end{itemize}

For each  $x\in Q_v$, we apply the  \emph{Bernoulli factory for subtraction} in \cite{Nacu2005FastSO}, denoted by \subbf{},
such that it constructs a new coin $\+{O}_{\nu(x)-\theta_v}=\subbf{}(\+{O}_{\nu(x)}, \+{O}_{\theta_v},\zeta)$ with probability of heads $\nu(x)-\theta_v$.

We then apply the \emph{Bernoulli race} in \cite{Dughmi17Bernoulli}, denoted by \berrace{},
such that the subroutine $\berrace{}\left(\{\+{O}_{\nu(x)-\theta_v}\}_{x\in Q_v}\right)$ returns a random value $I\in Q_v$ satisfying that $I=x$ with probability proportional to  $\nu(x)-\theta_v=\mu_v^{\sigma}(x)-\theta_v$, i.e.~$I$ is distributed as $\+{D}$ defined in \eqref{eq:definition-margin-overflow-app}.
This achieves our goal.

Although these constructions are not new, for rigorousness, we restate the precise constructions.

The Bernoulli race \cite{Dughmi17Bernoulli} subroutine \berrace{}$\left(\left\{\+{O}^{1},\+{O}^{2},\ldots, \+{O}^{q}\right\}\right)$ is given accesses to a list of coins $\+{O}^{1}=\+{O}_{\xi_1},\+{O}^{2}=\+{O}_{\xi_2},\ldots, \+{O}^{q}=\+{O}_{\xi_q}$ with unknown $\xi_1,\xi_2,\ldots,\xi_q\in[0,1]$.
Its goal is to return a random $I\in [q]$ such that $I=i$ with probability $\xi_i/\sum_{j=1}^q\xi_j$. 
This can be achieved by independently repeating the following until a value is returned:
\begin{itemize}
\item choose $I\in[q]$ uniformly at random;
\item if a draw of $\+{O}^I$ returns 1 then return $I$.
\end{itemize}
The correctness of this procedure was given in \cite{Dughmi17Bernoulli}.
\begin{proposition}[\text{\cite[Theorem 3.3]{Dughmi17Bernoulli}}]\label{brcor}
Given access to a list of coins $L=\{\+{O}^1,\+{O}^2,\dots,\+{O}^q\}$, where for each $i\in [q]$, the probability of heads for $\+{O}^i$ is $\xi_i$, the \berrace{}$(L)$ defined above terminates with probability 1 and returns a random $I\in[q]$ such that $\Pr{I=i}=\xi_i/\sum_{j=1}^q\xi_j$ for every $i\in[q]$.
\end{proposition}

To define the Bernoulli factory for subtraction, we further need to construct a linear Bernoulli factory, which transforms $\+{O}_\xi$ to $\+{O}_{C\xi}$ for a $C>1$ with the promise that $C\xi\leq 1$.
We adopt the construction of linear Bernoulli factory in \cite{Hub16Bernoulli} described in \Cref{Alg:lbf}.
Its correctness is guaranteed as follows.

\begin{algorithm}
  \caption{$\linbf{}(\+{O}, C, \zeta)$\cite{Hub16Bernoulli}} \label{Alg:lbf}
  \KwIn{a coin $\+{O}=\+{O}_{\xi}$ with unknown $\xi$,  $C>1$ and a slack $\zeta>0$, with promise that $C\xi\leq 1-\zeta$;}  
  \KwOut{a random value $\textsf{Bernoulli}(C\xi)$;}
  $k\gets 4.6/\zeta,\zeta\gets\min\{\zeta,0.644\},i\gets 1$\;
  \Repeat{$i=0$ or $R=0$}
  {
    \Repeat{$i=0$ or $i\geq k$}
    {
        draw $B\gets \+{O}$, $G\gets \textsf{Geometric}\left(\frac{C-1}{C}\right)$\;
        \tcp{\small$G$ is drawn according to geometric distribution with parameter $\frac{C-1}{C}$}
        $i\gets i-1+(1-B)G$\;
    }
    \If{$i\geq k$}
    {
        draw $R\gets \textsf{Bernoulli}\left((1+\zeta/2)^{-i}\right)$\;
        $C\gets C(1+\zeta/2),\zeta\gets \zeta/2,k\gets 2k$\;
    }
  }
    \textbf{return} $\one{i=0}$\;
\end{algorithm}

\begin{proposition}[\text{\cite[Theorem 1]{Hub16Bernoulli}}]\label{linbfcor}
Given access to a coin $\+{O}_{\xi}$, given as input $C>1$ and $\zeta>0$, with promise that $C\xi\leq 1-\zeta$, 
$\linbf{}(\+{O}, C, \zeta)$ terminates with probability $1$ and returns a draw of $\+{O}_{C\xi}$.
\end{proposition}

A Bernoulli factory for subtraction, $\subbf{}(\+{O}_{\xi_1}, \+{O}_{\xi_2},\zeta)$, is given in \cite{Nacu2005FastSO}, which transforms two coins $\+{O}_{\xi_1}, \+{O}_{\xi_2}$ with the promise that $\xi_1-\xi_2\geq \zeta>0$, to a new coin $\+{O}_{\xi_1-\xi_2}$.
We implement this procedure using the linear Bernoulli factory defined above:
%
\begin{itemize}
\item  $\subbf{}\left(\+{O}_{\xi_1}, \+{O}_{\xi_2},\zeta\right)=1-\linbf{}\left(\+{O}_{(1-\xi_1+\xi_2)/2},2,\zeta\right)$,
\end{itemize}
where the coin $\+{O}_{(1-\xi_1+\xi_2)/2}$ is realized with $\+{O}_{1/2}$, $\+{O}_{\xi_1}$ and $\+{O}_{\xi_2}$ as follows: if $\+{O}_{1/2} = 1$, return $1 - \+{O}_{\xi_1}$; otherwise, return $\+{O}_{\xi_2}$.
The correctness of this procedure is guaranteed as follows.

\begin{proposition}[\text{\cite[Proposition 14, (iv)]{Nacu2005FastSO}}]\label{subbfcor}
Given access to two coins $\+{O}_{\xi_1}$ and $\+{O}_{\xi_2}$, and given as input $\zeta>0$, with promise that $\xi_1-\xi_2\geq \zeta$, 
$\subbf{}\left(\+{O}_{\xi_1}, \+{O}_{\xi_2},\zeta\right)$ terminates with probability $1$ and returns a draw of $\+{O}_{\xi_1-\xi_2}$.
\end{proposition}

Recall the coins $\+O_{\theta_v}$ and $\+O_{\nu(x)}$ for $x\in Q_v$ where $\nu=\mu_v^{\sigma}$.
We further assume that the values in $Q_v$ are enumerated in an arbitrary order as $Q_v=\{x_1,x_2,\ldots,x_{q_v}\}$.
To draw a sample from the distribution $\+{D}$ defined in \eqref{eq:definition-margin-overflow-app}, we construct the following Bernoulli factory:
\begin{itemize}
\item  for each $i\in [q_v]$,  let $\+{O}^i=\subbf{}\left(\+{O}_{\nu(x_i)}, \+{O}_{\theta_v},\zeta\right)$;
\item  draw $I\gets \berrace{}\left(\left\{\+{O}^1,\+{O}^2,\ldots,\+{O}^{q_v}\right\}\right)$, and return the $I$-th value $x_I$ in $Q_v$.
\end{itemize}
The parameter $\zeta>0$ in above is as fixed in \eqref{eq:parameter-theta} and satisfies the promise \eqref{eq:promise-margin-lower-bound-app} assuming \Cref{inputcondition-recur} (due to \Cref{localuniformitycor}).
Thus by Propositions \ref{brcor}, \ref{linbfcor} and \ref{subbfcor}, the above procedure terminates with probability 1 and returns an $x_I$ distributed as $\+{D}$ defined as in \eqref{eq:definition-margin-overflow-app}.
This proves \Cref{bercorrect}.

\subsection{Efficiency of the Bernoulli Factory}
We now bound the efficiency of the Bernoulli factory constructed above. 
In this analysis, we need to explicitly bound the costs for realizations of the the basic oracles $\+{O}_{\nu(x)}$ for $x\in Q_v$ through the rejection sampling $\+{O}_{\nu(x)}=\one{\rejsamp{}(\Phi,\sigma,\{v\})=x}$, whose complexity is measured in terms of both the computation cost and the query complexity for the evaluation oracle in \Cref{definition:evaluation-oracle}.


%
%

{Recall the simplification and decomposition of CSP  
defined in \Cref{sec:rejection-sampling}.
Let $\Phi^\sigma=(V^\sigma,\+{Q}^{\sigma},\+{C}^{\sigma})$ denote the simplification of $\Phi$ under partial assignment $\sigma\in\qs$, and $H^{\sigma}=H_{\Phi^{\sigma}}=(V^\sigma,\+{C}^{\sigma})$ its hypergraph representation. 
Recall that for each $v\in V^{\sigma}$, $H_v^{\sigma}=(V_v^\sigma,\+{C}_v^{\sigma})$ denotes the connected component in $H^{\sigma}$ that contains the vertex/variable $v$.
Let $\Phi^\sigma_v$ be its corresponding formula. }

We show the following theorem for upper bound on the complexity of the Bernoulli factory. 
\begin{theorem}\label{beref}
Assuming \Cref{inputcondition-recur} for the input $(\Phi, \sigma ,v)$, 
the Bernoulli factory algorithm constructed in \Cref{sec:bernoulli-factory-construction} costs in expectation:
\begin{itemize}
    \item  $O\left(q^2k^2\Delta^6(\abs{\+C^{\sigma}_v}+1) (1-\mathrm{e}\pprime q)^{-\abs{\+C^{\sigma}_v}}\right)$ queries to the evaluation oracle in \Cref{definition:evaluation-oracle};
    \item $O\left(q^3k^3\Delta^6(\abs{\+C^{\sigma}_v}+1) (1-\mathrm{e}\pprime q)^{-\abs{\+C^{\sigma}_v}}\right)$ in computation.
\end{itemize}
\end{theorem}

To prove \Cref{beref}, we first bound the cost for realizing the basic oracles  $\+{O}_{\nu(x)}$ for $x\in Q_v$. 
%

\begin{lemma}\label{lemma:csot-BF-basic-oracle}
Assume \Cref{inputcondition-recur} for the input $(\Phi, \sigma ,v)$.
It takes at most $\Delta(\abs{\+C^\sigma_v}+1)$ queries to the evaluation oracle in \Cref{definition:evaluation-oracle} and 
$O\left(k\Delta(\abs{\+C^{\sigma}_v}+1)\right)$ computation cost
for preprocessing the oracles  $\+{O}_{\nu(x)}$ for all $x\in Q_v$.
And upon each query, $\+{O}_{\nu(x)}$ returns using at most $\abs{ \+{C}_v^\sigma}(1-\mathrm{e}\pprime q)^{-\abs{\+C^{\sigma}_v}}$ queries to the evaluation oracle in expectation and 
$O\left((qk(\abs{ \+{C}_v^\sigma}+1))(1-\mathrm{e}\pprime q)^{-\abs{\+C^{\sigma}_v}}\right)$ computation cost in expectation.
\end{lemma}
\begin{proof}
The oracle $\+{O}_{\nu(x)}$ is computed as $\+{O}_{\nu(x)}=\one{\rejsamp{}(\Phi,\sigma,\{v\})=x}$.

First, observe that $\rejsamp{}(\Phi,\sigma,\{v\})$ is equivalent to $\rejsamp{}(\Phi_v^\sigma,\sigma_{V\setminus\Lambda(\sigma)},\{v\})$.
By using a depth-first search in $H^{\sigma}$, the connected component $\Phi_v^\sigma=(V_v^{\sigma},\+{Q}^{\sigma}_v,\+{C}^{\sigma}_v)$ can be constructed using at most $\Delta(\abs{\+C^\sigma_v}+1)$ queries to the evaluation oracle and $O\left(\Delta\abs{\+C^{\sigma}_v}+\Delta\abs{V^\sigma_v}\right)=O\left(k\Delta(\abs{\+C^{\sigma}_v}+1)\right)$ computation cost, because $\abs{V^\sigma_v}\le k\abs{\+C^{\sigma}_v}+1$.
This is the preprocessing cost.

Then a query to oracle $\+{O}_{\nu(x)}$ is reduced to a calling to $\rejsamp{}(\Phi_v^\sigma,\sigma_{V\setminus\Lambda(\sigma)},\{v\})$ using \Cref{Alg:rej}.
In fact, \Cref{line-rs-find} can be skipped and $K=1$ in \Cref{line-rs-for} since the component $\Phi_v^\sigma$ containing $v$ has been explicitly constructed in the preprocessing.
It is well known that the expected number of trials (the \textbf{repeat} loop in \Cref{line-rs-until}) taken by the rejection sampling until success is given by $\mathbb{P}_{\Phi^{\sigma}_v}[\Omega_{\Phi^{\sigma}_v}]^{-1}$, 
where $\Omega_{\Phi^{\sigma}_v}$ is the set of satisfying assignments of $\Phi^{\sigma}_v$ and hence  $\mathbb{P}_{\Phi^{\sigma}_v}[\Omega_{\Phi^{\sigma}_v}]$ gives the probability that a uniform random assignment is satisfying for $\Phi^{\sigma}_v$. 
By \Cref{locallemma}, assuming \Cref{inputcondition-recur},
\[
\mathbb{P}_{\Phi^{\sigma}_v}[\Omega_{\Phi^{\sigma}_v}]
\ge 
(1-\mathrm{e}\pprime q)^{\abs{\+C^{\sigma}_v}}.
\]
The rejection sampling in $\rejsamp{}(\Phi_v^\sigma,\sigma_{V\setminus\Lambda(\sigma)},\{v\})$ takes 
$(1-\mathrm{e}\pprime q)^{-\abs{\+C^{\sigma}_v}}$ trials in expectation.
And within each trial, it is easy to verify that it uses at most $\abs{\+{C}_v^\sigma}$ queries to the evaluation oracle and $O\left(k\abs{ \+{C}_v^\sigma}+q\abs{ V_v^\sigma}\right)=O\left(qk(\abs{ \+{C}_v^\sigma}+1))\right)$ computation cost.
This proves the lemma.
\end{proof}






We then state known results for the efficiency of Bernoulli factories.

\begin{proposition}[\text{\cite[Theorem 3.3]{Dughmi17Bernoulli}}]\label{bref}
Given access to a list of coins $L=\{\+{O}^1,\+{O}^2,\dots,\+{O}^q\}$, where for each $i\in [q]$, the probability of heads for $\+{O}^i$ is $\xi_i$, the expected number of queries to $\+{O}^1,\+{O}^2,\dots,\+{O}^q$ for  executing  \berrace{}$(L)$ is at most ${q}/\left({\sum_{i=1}^{q}\xi_i}\right)$.
\end{proposition}

\begin{proposition}[\text{\cite[Theorem 1]{Hub16Bernoulli}}]\label{linbfef}
Given access to a coin $\+{O}_{\xi}$, given as input $C>1,\zeta>0$, with the promise $C\xi\leq 1-\zeta$, 
the expected number of queries to $\+{O}_{\xi}$ made in $\linbf{}(\+{O}, C, \zeta)$ is
at most $ 9.5C/\zeta$.
\end{proposition}

By \Cref{linbfef}, we have the following complexity bound for the Bernoulli factory for subtraction.

\begin{corollary}\label{subef}
Given access to two coins $\+{O}_{\xi_1},\+{O}_{\xi_2}$, 
and given as input $\zeta>0$, with the promise $\xi_1-\xi_2\geq \zeta$, 
the expected number of queries to $\+{O}_{\xi_1},\+{O}_{\xi_2}$ for executing $\subbf{}\left(\+{O}_{\xi_1}, \+{O}_{\xi_2},\zeta\right)$ is  
at most $\frac{39\zeta^{-1}}{1-(\xi_1-\xi_2)}$.
\end{corollary}




\begin{proof}[Proof of \Cref{beref}]
Recall that our Bernoulli factory algorithm is $\berrace{}\left(\left\{\+{O}^1,\+{O}^2,\ldots,\+{O}^{q_v}\right\}\right)$ where $\+{O}^i=\subbf{}\left(\+{O}_{\nu(x_i)}, \+{O}_{\theta_v},\zeta\right)$ for the $i$-th value $x_i\in Q_v$ and $\nu=\mu_v^\sigma$.
By \Cref{bref} and \Cref{subef},
the total number of queries to the basic oracles $\+{O}_{\nu(x)}$ for $x\in Q_v$ 
is bounded by:
\begin{align*}
\frac{q_v}{\sum\limits_{x\in Q_v}\left( \mu_v^{\sigma}(x)-\theta_v\right)}\cdot \left(\max_{x\in Q_v}\frac{39\zeta^{-1}}{1-(\mu_v^{\sigma}(x)-\theta_v)}\right)
\le
\frac{39}{\zeta^{2}(1-2\eta-\zeta)}
=O\left(q^2k^2\Delta^6\right),
\end{align*}
where the inequality is due to $\mu_v^{\sigma}(x)\le \theta_v+2\eta+\zeta$ by \Cref{localuniformitycor} assuming \Cref{inputcondition-recur}.
The theorem then follows by applying \Cref{lemma:csot-BF-basic-oracle}
%
and observing that the preprocessing costs are paid only once in the beginning.
\end{proof}

\section{Basic Properties of Variable/Constraint Attributes along \pth{}}\label{sec:monotonicity-proofs}
In this section, we prove the technical lemma (\Cref{cor-mono}) regarding the attributes of various variable/constraint sets $\vstar{\sigma},\cfrozen{\sigma},\ccon{\sigma},\csfrozen{\sigma},\+C^{\sigma}_v$ along $\pth$,
where $\cfrozen{\sigma}$ is defined in \Cref{definition:frozen-fixed}, $\+C^{\sigma}_v$ in \Cref{sec:rejection-sampling},
$\ccon{\sigma}$ in \Cref{definition:boundary-variables},
and $\vstar{\sigma}, \csfrozen{\sigma}$ in \Cref{def:cbad}.

Recall in \Cref{definition:boundary-variables}:   for any $\sigma\in \qs$, $\hfix{\sigma}$ denotes the sub-hypergraph of $H^\sigma$ induced by $V^{\sigma}\cap\vfix{\sigma}$. 
Recall in \Cref{sec:rejection-sampling} that for each $v\in V^{\sigma}$,
$H_v^{\sigma}=(V_v^\sigma,\+{C}_v^{\sigma})$ denotes the connected component in $H^{\sigma}$ that contains the vertex/variable $v$. 
For each $c\in \+C$, we denote the simplified constraint of $c$ under $\sigma$ as~$c^{\sigma}$.

Note that \Cref{cor-mono} consists of three parts: monotonicity property, upper bound on  $\abs{\+C^{\sigma_{\ell}}_v},\abs{\ccon{\sigma_{\ell}}})$, and upper bound on length of $\pth(\sigma)$. We prove \Cref{cor-mono} by showing these three parts in order.

\subsection{Proof of the monotonicity property}
The following lemma will be used in the proof of the monotonicity property in \Cref{cor-mono}.

\begin{lemma}\label{property-u}
Given $\sigma\in \qs$ with $\nextvar{\sigma} = u \neq \perp$, it holds that
$u\not \in \vfix{\sigma}$.
\end{lemma}
\begin{proof}
By $ u = \nextvar{\sigma} \neq \perp$ and the definition of $\nextvar{\sigma}$ in \Cref{definition:boundary-variables},
we have $u\in \vinf{\sigma}$.
Combining with the definition of $\vinf{\sigma}$,
we have $u \not \in \vstar{\sigma}$.
Combining with $\vstar{\sigma} \subseteq V^{\sigma}\cap \vfix{\sigma}$,
we have $u\not \in \vfix{\sigma}$.
\end{proof}

The next lemma states a basic monotonicity property when extending some partial assignment $\sigma$ on $\nextvar{\sigma}$.

\begin{lemma}\label{lem-mono}
Given $\sigma\in \qs$ with $\nextvar{\sigma} = u \neq \perp$ and $a\in Q_u\cup \{\star\}$, let $\tau=\sigma_{u\gets a}$. 
it holds that 
\[
\vst{\sigma}\subseteq \vst{\tau},\quad
\+{C}^{\sigma}_{\+{P}}\subseteq \+{C}^{\tau}_{\+{P}},
\]
where $\+{P}$ can be any attribute $\+{P}\in\{ \mathsf{frozen},\,\, \star\text{-}\mathsf{con},\,\, \star\text{-}\mathsf{frozen} \}$.
\end{lemma}
\begin{proof}
At first, we prove $\vstar{\sigma}\subseteq \vstar{\tau}$.
By $\nextvar{\sigma} = u \neq \perp$ and \Cref{property-u},
we have $u\not \in \vfix{\sigma}$.
Combining with the definition of $\vfix{\sigma}$ in \Cref{definition:frozen-fixed}, we have $\sigma(u)\neq \star$. Therefore it is straightforward by \Cref{def:cbad} that $\vstar{\sigma}\subseteq \vstar{\tau}$.

Now we prove $\cfrozen{\sigma}\subseteq \cfrozen{\tau}$.
For each $c\in \cfrozen{\sigma}$, we have 
$\var{c}\subseteq \vfix{\sigma}$.
Combining with $u\not \in \vfix{\sigma}$, we have $u\not\in \var{c}$, which says $\tau_{\var{c}}=\sigma_{\var{c}}$, hence $c\in \cfrozen{\tau}$ by the consistency assumption of frozen oracle in \Cref{assumption:frozen-oracle}.
In summary, we have $\cfrozen{\sigma}\subseteq \cfrozen{\tau}$.

In the next, we prove $\ccon{\sigma}\subseteq \ccon{\tau}$.
For each $c\in \ccon{\sigma}$, by \Cref{definition:boundary-variables}, there exists some vertex $v\in \var{c}\cap \vcon{\sigma}$.
By $v\in \vcon{\sigma}$, we have there exists some variable $v'$ such that $\sigma(v')=\star$ and $v$ and $v'$ are connected in $\hfix{\sigma}$.
We claim that  
$\hfix{\sigma}$ is a sub-hypergraph of $\hfix{\tau}$.
Then we have $v$ and $v'$ are connected in $\hfix{\tau}$.
In addition, recall $u\not \in \vfix{\sigma}$. Combining with $v'\in \vfix{\sigma}$, we have $u\neq v'$ and then $\tau(v') = \sigma(v') = \star$.
Combining with $v$ and $v'$ are connected in $\hfix{\tau}$, we have $v\in \vcon{\tau}$. 
Combining with $v\in \var{c}$, we have $c\in \ccon{\tau}$.
In summary, we have $\ccon{\sigma}\subseteq \ccon{\tau}$.

Now we prove the claim that $\hfix{\sigma}$ is a sub-hypergraph of $\hfix{\tau}$.
At first, we show
\begin{align}\label{relation-v-hfix}
V^{\sigma}\cap\vfix{\sigma} \subseteq  V^{\tau}\cap\vfix{\tau}.
\end{align}
Obviously, $\Lambda^+(\sigma)\subseteq \Lambda^+(\tau)$.
Combining with $\cfrozen{\sigma}\subseteq \cfrozen{\tau}$,
we have $\vfix{\sigma} \subseteq \vfix{\tau}$.
In addition, we have 
$$V^{\sigma}\cap\vfix{\sigma} = \left(V^{\sigma}\setminus \{u\}\right)\cap\vfix{\sigma} \subseteq  V^{\tau}\cap\vfix{\sigma}\subseteq V^{\tau}\cap\vfix{\tau},$$
where the first relation is by $u\not\in \vfix{\sigma}$,
the second is by $V^{\sigma}\setminus \{u\} \subseteq V^{\tau}$ and 
the last is by $\vfix{\sigma} \subseteq \vfix{\tau}$.
Thus, we have \eqref{relation-v-hfix} holds.
In addition, let $\+C_{\textsf{fix}}^{\sigma}$ be the hyperedge set of $\hfix{\sigma}$ and $\+C_{\textsf{fix}}^{\tau}$ be the hyperedge set of $\hfix{\tau}$. We can show that
\begin{align*}
&\quad \+C_{\textsf{fix}}^{\sigma}\\
(\text{definitions of $\hfix{\sigma}$ and $\+C_{\textsf{fix}}^{\sigma}$})&= \{c^\sigma\mid \left(\var{c^{\sigma}} \subseteq V^{\sigma}\cap\vfix{\sigma}\right)\wedge\left(\text{$c$ is not satisfied by $\sigma$}\right)\} \\
(\text{by $u\not\in \vfix{\sigma}$})& = \{c^\sigma\mid \left(\var{c^{\sigma}} \subseteq V^{\sigma}\cap\vfix{\sigma}\right)\wedge(u\not\in\var{c})\wedge\left(\text{$c$ is not satisfied by $\sigma$}\right)\}\\
(\text{by $c^{\sigma} = c^{\tau}$ if $u\not\in \var{c}$})& \subseteq \{c^\tau\mid \left(\var{c^{\tau}} \subseteq V^{\sigma}\cap\vfix{\sigma}\right)\wedge(u\not\in\var{c})\wedge\left(\text{$c$ is not satisfied by $\tau$}\right)\}\\
(\text{by \eqref{relation-v-hfix}})& \subseteq \{c^\tau\mid \left(\var{c^{\tau}} \subseteq V^{\tau}\cap\vfix{\tau}\right)\wedge(u\not\in\var{c})\wedge\left(\text{$c$ is not satisfied by $\tau$}\right)\}\\
&\subseteq\{c^\tau\mid \left(\var{c^{\tau}} \subseteq V^{\tau}\cap\vfix{\tau}\right)\wedge\left(\text{$c$ is not satisfied by $\tau$}\right)\} \\
(\text{definitions of $\hfix{\tau}$ and $\+C_{\textsf{fix}}^{\tau}$})&= \+C_{\textsf{fix}}^{\tau}
\end{align*}
Combining with \eqref{relation-v-hfix}, we have proven the claim that $\hfix{\sigma}$ is a sub-hypergraph of $\hfix{\tau}$.

By $\cfrozen{\sigma}\subseteq \cfrozen{\tau}$ and $\ccon{\sigma} \subseteq \ccon{\tau}$,
we have 
$$\csfrozen{\sigma} = \cfrozen{\sigma}\bigcap \ccon{\sigma} \subseteq \cfrozen{\tau}\bigcap \ccon{\tau} = \csfrozen{\tau}.$$
\end{proof}

By definition of $\pth(\cdot)$ and \Cref{lem-mono}, the monotonicity property in \Cref{cor-mono} is immediate by induction.

\subsection{Proof of the upper bound on $\abs{\+C^{\sigma_{\ell}}_v}$ and $\abs{\ccon{\sigma_{\ell}}}$}
The following lemma will be used in the proof of the upper bound on $\abs{\+C^{\sigma_{\ell}}_v}$ and $\abs{\ccon{\sigma_{\ell}}}$ in \Cref{cor-mono}.

\begin{lemma}\label{lem-cxlv-cxlstar}
Let $\sigma\in \qs$ be a partial assignment with exactly one variable $v\in V$ having $\sigma(v)=\star$ and $\pth(\sigma) = (\sigma_0,\sigma_1,\dots,\sigma_{\ell})$. For each $c^{\sigma_{\ell}}\in \+C^{\sigma_{\ell}}_v$ and each variable $u\in \var{c^{\sigma_{\ell}}}$, it holds that $c\in \ccon{\sigma_{\ell}}$ and $u\in \vfix{\sigma_{\ell}}$.
\end{lemma}
\begin{proof}
For each simplified constraint $c^{\sigma_{\ell}}\in \+C^{\sigma_{\ell}}_v$, 
by the definition of $\+C^{\sigma_{\ell}}_v$,
we have there exists a connected path $c^{\sigma_{\ell}}_1,c^{\sigma_{\ell}}_2,\dots,c^{\sigma_{\ell}}_{t}\in \+C^{\sigma_{\ell}}_v$ such that $v \in \var{c^{\sigma_{\ell}}_1}$, $c^{^{\sigma_{\ell}}}_{t}=c^{^{\sigma_{\ell}}}$ and $c^{\sigma_{\ell}}_i$ intersects $c^{\sigma_{\ell}}_{i+1}$ for each $1\leq i<t$.
Let $\text{dist}(v,c^{\sigma_{\ell}},\sigma_{\ell})$, or $\text{dist}(v,c^{\sigma_{\ell}})$ for short, denote the length of the shortest connected path from $v$ to $c^{\sigma_{\ell}}$ in $H^\sigma_{\ell}$.
We prove the lemma by induction on $\text{dist}(v,c^{\sigma_{\ell}})$.

For the base case when $\text{dist}(v,c^{\sigma_{\ell}})=1$, we have $v\in \var{c^{\sigma_{\ell}}}$.
In addition, by $\sigma_{\ell}(v)=\sigma(v)=\star$,
we have $v\in V^{\sigma_{\ell}}\cap\vfix{\sigma_{\ell}}$.
Thus, we have $v\in \vcon{\sigma_{\ell}}$.
Combining with $v\in \var{c^{\sigma_{\ell}}}$,
we have $c\in \ccon{\sigma_{\ell}}$ by  \Cref{definition:boundary-variables}.
In addition, we have $u\in \vfix{\sigma_{\ell}}$ for each $u\in \var{c^{\sigma_{\ell}}}$.
Because otherwise, $u \not \in \vfix{\sigma_{\ell}}$.
We have $u\in \var{c^{\sigma_{\ell}}}\setminus \vfix{\sigma_{\ell}}\subseteq V^{\sigma_{\ell}} \setminus\vcon{\sigma_{\ell}}$ by
$\vcon{\sigma_{\ell}}\subseteq \vfix{\sigma_{\ell}}$.
In addition, by $u,v\in \var{c^{\sigma_{\ell}}}$, $v\in \vcon{\sigma_{\ell}}$,
and $c^{\sigma_{\ell}}\in \+C_{v}^{\sigma_{\ell}} \subseteq \+C^{\sigma_{\ell}}$, we have $u\in \vinf{\sigma_{\ell}}$.
Therefore, we have $\vinf{\sigma_{\ell}}\neq \emptyset$
and $\nextvar{\sigma_{\ell}}\neq \perp$.
Thus, by \Cref{pathdef}, there must be another partial assignment $\sigma_{\ell+1}$ generated from $\sigma_{\ell}$,
which is contradictory with $\pth(\sigma)=(\sigma_{0},\dots\sigma_{\ell})$.

For the induction step, we assume $\text{dist}(v,c^{\sigma_{\ell}})=t>1$. Let $c^{\sigma_{\ell}}_1,c^{\sigma_{\ell}}_2,\dots,c^{\sigma_{\ell}}_{t} \in \+C^{\sigma_{\ell}}_v$ be a connected path in $H^{\sigma_{\ell}}_v$ with $v \in \var{c^{\sigma_{\ell}}_0}$, $c^{\sigma_{\ell}}_{t}=c^{\sigma_{\ell}}$ and $c^{\sigma_{\ell}}_i$ intersects $c^{\sigma_{\ell}}_{i+1}$ for each $1\leq i< t$.  
By the induction hypothesis
and $c_i^{\sigma_{\ell}}\in \+C_v^{\sigma_{\ell}}$,
we have $c_i\in \ccon{\sigma_{\ell}}$ for each $i<t$. 
Choose some $w\in \var{c^{\sigma_{\ell}}_{t-1}}\cap \var{c^{\sigma_{\ell}}_t}$.
By the induction hypothesis, we have $w\in \vfix{\sigma_{\ell}}$.
Combining with $w\in V^{\sigma_{\ell}}$,
we have $w\in \vfix{\sigma_{\ell}}\cap V^{\sigma_{\ell}}$.
Recall that $\sigma_{\ell}(v)=\star$.
Combining with $w\in \vfix{\sigma_{\ell}}\cap V^{\sigma_{\ell}}$ and that
$w\in \var{c^{\sigma_{\ell}}_{t-1}}$ is connected to $v\in \var{c^{\sigma_{\ell}}_1}$ by the path 
$c^{\sigma_{\ell}}_1,c^{\sigma_{\ell}}_2,\dots,c^{\sigma_{\ell}}_{t-1}$ in $H^{\sigma_{\ell}}_v$,
we have $w\in \vcon{\sigma_{\ell}}$.
Combining with $w\in \var{c^{\sigma_{\ell}}_t}$,
we have $c_{\ell}\in \ccon{\sigma_{\ell}}$.
In addition, we have $u\in \vfix{\sigma_{\ell}}$ for each $u\in \var{c^{\sigma_{\ell}}}$.
Because otherwise, $u \not \in \vfix{\sigma_{\ell}}$.
We have $u\in \var{c^{\sigma_{\ell}}}\setminus \vfix{\sigma_{\ell}}\subseteq V^{\sigma_{\ell}} \setminus\vcon{\sigma_{\ell}}$ by
$\vcon{\sigma_{\ell}}\subseteq \vfix{\sigma_{\ell}}$.
In addition, by $u,w\in \var{c^{\sigma_{\ell}}}$, $w\in \vcon{\sigma_{\ell}}$,
and $c^{\sigma_{\ell}}\in \+C_{v}^{\sigma_{\ell}} \subseteq \+C^{\sigma_{\ell}}$, we have $u\in \vinf{\sigma_{\ell}}$.
Therefore, similar to the base case one can also reach a contradiction.
This completes the induction step and the proof of the lemma.
\end{proof}

Because $\abs{\+C^{\sigma_{\ell}}_v} \leq \abs{\ccon{\sigma_{\ell}}}$ is immediate by \Cref{lem-cxlv-cxlstar}, it is sufficient to show that $\abs{\ccon{\sigma_{\ell}}}\leq \Delta \cdot \abs{\csfrozen{\sigma_{\ell}}}+\Delta\cdot \abs{\vst{\sigma_{\ell}}} $. We show this by proving that
for each $c\in \ccon{\sigma_{\ell}
}$, either there exists some $u\in \var{c}$ such that $u\in \vst{\sigma_{\ell}}$, or there exists some $c'\in \csfrozen{\sigma_{\ell}}$ such that $\var{c}\cap \var{c'}\neq \emptyset$.

For each $c\in \ccon{\sigma_{\ell}
}$,
by \Cref{definition:boundary-variables},
we have there exists some $u\in \vcon{\sigma_{\ell}}\cap \var{c^{\sigma_{\ell}}}$.
By $u\in \vcon{\sigma_{\ell}}$,
we have $u\in  V^{\sigma_{\ell}}\cap \vfix{\sigma_{\ell}}$.
By $u\in V^{\sigma_{\ell}}$,
we have $u\not\in \Lambda(\sigma_{\ell})$.
Combining with 
$u\in \vfix{\sigma_{\ell}}$,
we have either $\sigma_{\ell}(u) = \star$ or
$u\in c'$ for some $c'\in \cfrozen{\sigma_{\ell}}$.
If $\sigma_{\ell}(u) = \star$,
we have $u\in \vstar{\sigma_{\ell}}$.
Otherwise, $u\in c'$ for some $c'\in \cfrozen{\sigma_{\ell}}$.
In addition, we also have $c'\in \ccon{\sigma_{\ell}}$ by $u\in \vcon{\sigma_{\ell}}$ and $u\in \var{c'}$.
Combining with $c'\in \cfrozen{\sigma_{\ell}}$,
we have $c'\in \csfrozen{\sigma_{\ell}}$.
This completes the proof of the upper bound on $\abs{\+C^{\sigma_{\ell}}_v}$ and $\abs{\ccon{\sigma_{\ell}}}$ in \Cref{cor-mono}.

\subsection{Proof of the upper bound on length of $\pth(\sigma)$}

 Fix any $0\leq i\leq \ell$.  We claim that 
for each $0\leq  j <i$,
\begin{enumerate}
    \item either there exist some $c_j,c_j'$ such that 
$\nextvar{\sigma_j}\in \var{c_j}$, $c_j'\subseteq \csfrozen{\sigma_{i}}$, and $\var{c_j}\cap\var{c_j'}\neq \emptyset$;\label{shortpath2-item1}
    \item  or there exist some $c_j,u_j$ such that
$\nextvar{\sigma_j},u_j\in \var{c_j}$ and $u_j\in \vst{\sigma_i}$. \label{shortpath2-item2}
\end{enumerate}
Therefore, for each $0\leq j<i$, $\nextvar{\sigma_{j}}$ is in a constraint $c$ where either
$\var{c}\cap \var{c'}\neq \emptyset \text{ for some }c'\in \csfrozen{\sigma_{i}}$, or $u\in \var{c}$ for some $u\in \vstar{\sigma_i}$.
By \Cref{lem-mono} and induction, we have $\csfrozen{\sigma_i}\subseteq \csfrozen{\sigma_{\ell}}$ and $\vstar{\sigma_i}\subseteq \vstar{\sigma_{\ell}}$ for each $0\leq i\leq \ell$.  Combining with $\abs{\var{c}}\leq k$, we have
\begin{align*}
\ell &\leq k\cdot \abs{ \{c\in \+{C}: \var{c}\cap \var{c'}\neq \emptyset \text{ for some }c'\in \csfrozen{\sigma_{\ell}}\text{ or } u\in \var{c} \text{ for some }u\in \vst{\sigma_{\ell}} \}}\\
&\leq k\Delta\cdot \left(\abs{\csfrozen{\sigma_{i}}}+\abs{ \vst{\sigma_{i}}}\right),
\end{align*}
which proves the upper bound on length of $\pth(\sigma)$ in \Cref{cor-mono}.

Now we prove the claim.
Note that by $\pth(\sigma)=(\sigma_0,\dots,\sigma_{\ell})$, $0\leq i\leq \ell$ and \Cref{pathdef}, we have $\nextvar{\sigma_{j}}\neq \perp$ for each $0\leq j<i$.
Assume that $\nextvar{\sigma_{j}} = u_j$.
By \Cref{definition:boundary-variables},
we have $u_j \in \vinf{\sigma_{j}}\neq \emptyset$.
Combining with the definition of $\vinf{\sigma_{j}}$,
we have there exists some $c_j\in \+{C}^{\sigma_j}$, $w_j\in \vcon{\sigma_j}$ such that $u_j,w_j\in\vbl(c_j)$.
By $w_j\in \vcon{\sigma_{j}}$,
we have $w_j\in  V^{\sigma_{j}}\cap \vfix{\sigma_{j}}$.
By $w_j\in V^{\sigma_{j}}$,
we have $w_j\not\in \Lambda(\sigma_{j})$.
Combining with 
$w_j\in \vfix{\sigma_{j}}$,
we have either $\sigma_j(w_j) = \star$ or
$w_j\in \widehat{c}_j$ for some $\widehat{c}_j\in \cfrozen{\sigma_{j}}$.
If $\sigma_j(w_j) = \star$,
we have $w_j\in \vst{\sigma_j}\subseteq \vst{\sigma_i}$ and $c_j,w_j$ satisfies \Cref{shortpath2-item2}.
Otherwise, $w_j\in \var{\widehat{c}_j}$ for some $\widehat{c}_j\in \cfrozen{\sigma_{j}}$.
In addition, by $w_j\in \vcon{\sigma_{j}}$ and $w_j\in \var{\widehat{c}_j}$, 
we have $\widehat{c}_j\in \ccon{\sigma_{j}}$.
Combining with $\widehat{c}_j\in \cfrozen{\sigma_{j}}$,
we have 
$\widehat{c}_j\in \csfrozen{\sigma_{j}}$. By $w_j\in\vbl(c_j)$ and $w_j\in \var{\widehat{c}_j}$, we have $\var{c_j}\cap\var{\widehat{c}_j}\neq \emptyset$ and $c_j,\widehat{c}_j$ satisfies \Cref{shortpath2-item1}.
This justifies the claim and finishes the proof of \Cref{cor-mono}.

\end{document}